\documentclass[12pt, a4paper]{article}
\usepackage[a4paper, margin=1in]{geometry}
\usepackage[utf8]{inputenc}
\usepackage{graphicx, hyperref, array}
\usepackage[T1]{fontenc}
\usepackage{microtype, courier, xcolor}
\usepackage{amsmath, amsthm, amssymb, mathtools}
\newtheorem{theorem}{Theorem}
\newtheorem{observation}{Observation}
\newtheorem{lemma}{Lemma}
\newtheorem{corollary}{Corollary}
\newtheorem{conjecture}{Conjecture}
\newtheorem{claim}{Claim}

\title{The Parameterized Complexity of Problems on Outer \(k\)-Planar Graphs}
\author{Xiaobin Ren\thanks{Email: \href{mailto:ren.xiaobin@icloud.com}{ren.xiaobin@icloud.com}} 
\and Hans L. Bodlaender \thanks{Department of Information and Computing Sciences, Utrecht University. Email: \href{mailto:h.l.bodlaender@uu.nl}{h.l.bodlaender@uu.nl}} }
\date{}

\begin{document}
\maketitle

\begin{abstract}

A graph is \emph{outer \(k\)-planar} if it admits a straight-line drawing in which all vertices lie on a circle and every edge is crossed by at most \(k\) other edges. We study the parameterized complexity of a broad collection of graph problems on outer \(k\)-planar graphs, with \(k\) as the parameter. Many graph problems are known to be XALP-hard when parameterized by treewidth or outerplanarity, and XNLP-hard when parameterized by pathwidth. We show that only a few such problems, including \textsc{Binary CSP} and \textsc{Scattered Set}, remain intractable on outer \(k\)-planar graphs, whereas a large class of the others become fixed-parameter tractable in this setting, assuming that an outer \(k\)-planar drawing of the input graph is given. These include \textsc{List Coloring}, \textsc{Capacitated Dominating Set}, \textsc{Capacitated Vertex Cover}, \textsc{Target Outdegree Orientation}, and \textsc{Target Set Selection}, among others.

In addition to the algorithmic and complexity results, we establish several structural results. We show that outer \(k\)-planar graphs have mim-width at most \(k+2\), that graphs of cut-width at most \(k\) are outer \(2k\)-planar, and that graphs of feedback edge set number at most \(k\) are outer \(6k\)-planar. We also show that many graph parameters are incomparable with outer \(k\)-planarity, thereby clarifying its position within the graph parameter hierarchy.


\end{abstract}

\pagenumbering{arabic}

\section{Introduction}
The study of graphs that admit drawings with edge crossings, under controlled and well-structured restrictions, has long been an intriguing topic in algorithmic graph theory and the theory of graph drawing.
A fundamental question in this area is how to restrict edge crossings in a meaningful way while preserving useful combinatorial or geometric structure. One of the most natural models in this direction is that of $k$-planar graphs. A graph is called \emph{$k$-planar} if it admits a drawing in which every edge is crossed by at most $k$ other edges. The study of such graphs dates back to the 1960s, when Ringel~\cite{Ringel1965} investigated their structural properties and established the first results on the chromatic number of $1$-planar graphs. Over the past several decades, this line of research has developed into a broad and active area, now commonly studied under the name of \emph{beyond-planar} graph classes. In recent years, numerous structural, algorithmic, and complexity-theoretic results have been obtained for several beyond-planar graph classes; see, for example,~\cite{1-planar, recog-k-planar, cyclomatic, local-crossing-3, local-crossing-4, book-2}.

Among the many beyond-planar graphs that have been introduced, outer $k$-planar graph is a particularly natural one. A graph is \emph{outer $k$-planar} if it admits a straight-line convex drawing in which all vertices lie on the outer face and every edge is crossed by at most $k$ other edges. This notion combines two simple restrictions: all vertices are placed in convex position, e.g. on a circle, and the number of crossings per edge is bounded. Outer $k$-planar graphs have been studied for a long time. Already in the 1990s, Pach and T\'{o}th showed that any outer $k$-planar graph on $n$ vertices has at most $4.1\sqrt{k}\,n$ edges~\cite{pach}. For \(k\ge 5\), the constant factor was later improved to 2.46 by Aichholzer et al.~\cite{outer-k-edge}. More recent research has focused on structural and graph drawing aspects. These include bounds on degeneracy, chromatic number, separator size, and treewidth, as well as algorithmic results for recognition and related drawing problems~\cite{beyond-survery, firman_et_al, outerkrecog}. Wood and Telle~\cite{Wood2007} showed that every outer $k$-planar graph has treewidth at most $3k+11$. This upper bound was later improved by Firman et al.~\cite{firman_et_al} to $1.5k+2$, who also established a lower bound of $k+2$ on the treewidth for even values of $k$. More recently, Pyzik~\cite{pyzik} proved that every outer $k$-planar graph has treewidth at least $1.5k+0.5$ for every odd value of $k$. Regarding recognition, Kobayashi et al.~\cite{outerkrecog} recently proved that testing outer $k$-planarity is XNLP-hard by giving a reduction from the \textsc{Bandwidth} problem. For an overview of more results on outer $k$-planar graphs, especially those concerning graph drawing, we refer the reader to~\cite{outer-1, beyond-survery, outerkrecog}.

In parameterized complexity, many graph problems are known to be XALP-hard or XNLP-hard when parameterized by structural measures such as treewidth or pathwidth, which implies such problems are also W[$t$]-hard for every $t \in \mathbb{N}$; see the recent line of work by Bodlaender et al.~\cite{xnlp1,xnlp2,xalp}. For a parameter, if every graph with this parameter bounded by $k$ also has treewidth bounded by $f(k)$, for some computable function $f$, then such graph classes can be viewed as subclasses of bounded-treewidth graphs in the parameterized setting. A natural question that has been frequently studied in this context is whether such hardness results with respect to treewidth persist on restricted subclasses of bounded-treewidth graphs, or whether the complexity drops to lower levels of the W-hierarchy, or even to fixed-parameter tractability.

A well-studied example is given by graphs of outerplanarity $k$, that is, $k$-outerplanar graphs, which have treewidth at most $3k-1$~\cite{BOD1998}. Bodlaender and Szilá\'{a}gyi~\cite{outerkplanar} showed that many problems that are XALP-hard when parameterized by treewidth remain XALP-hard on planar graphs when parameterized by outerplanarity. To avoid ambiguity, we emphasize that $k$-outerplanar graphs and outer $k$-planar graphs are different notions. A graph is $k$-outerplanar if it becomes empty after $k$ iterations of removing all vertices on the outer face. In particular, $k$-outerplanar graphs are planar, whereas outer $k$-planar graphs are not necessarily planar.

Since outer $k$-planar graphs also have treewidth $O(k)$, motivated by this line of work, we investigate whether such hardness results with respect to treewidth persist on outer $k$-planar graphs, when parameterized by the crossing bound $k$. More importantly, we ask whether the geometric structure of outer $k$-planar drawings can be exploited to provide additional useful information that makes these hard problems tractable. In this thesis, we give an affirmative answer to this question. We systematically study a broad range of classical graph problems and provide the first parameterized complexity results for these problems on outer $k$-planar graphs.

\subsection{Contributions}\label{contribution}
The problems studied in this paper, together with their main complexity results, are summarized in Table~\ref{tab:xalp-overview}. All of these problems are known to be XNLP-complete when parameterized by pathwidth, and XALP-complete when parameterized by treewidth or outerplanarity~\cite{outerkplanar, xalp, xnlp1, xnlp2, xnlp-flows}. These include, in particular, \textsc{Binary CSP}, \textsc{List Coloring}, \textsc{Precoloring Extension}, \textsc{Scattered Set}, \textsc{All-or-Nothing Flow}, \textsc{Target Set Selection}, several outdegree orientation problems, as well as capacitated variants of dominating set and vertex cover, among others.

\begin{table}[ht!]
\centering
\small
\renewcommand{\arraystretch}{1.5}
\begin{tabular}{|>{\raggedright\arraybackslash}m{4.6cm}|l|l|l|l|}
\hline
&\multicolumn{1}{c|}{$\operatorname{lcr}^\circ$} & \multicolumn{1}{c|}{outerplanarity} & \multicolumn{1}{c|}{treewidth} & \multicolumn{1}{c|}{pathwidth} \\ \hline

\textsc{Binary CSP}
& XALP-c.~(\ref{binary-csp})  & XALP-c.~\cite{outerkplanar} & XALP-c.~\cite{xalp} & XNLP-c.~\cite{xnlp1} \\ \hline

\textsc{Scattered Set}
& XALP-c. (\ref{scattered-set})  & XALP-c.~\cite{outerkplanar} & XALP-c.~\cite{outerkplanar} & XNLP-c.~\cite{outerkplanar} \\ \hline

\textsc{List Coloring}
& \multicolumn{1}{c|}{FPT (\ref{lc-pre})}  & XALP-c.~\cite{outerkplanar} & XALP-c.~\cite{xalp} & XNLP-c.~\cite{xnlp1} \\ \hline

\textsc{Precoloring Extension}
& \multicolumn{1}{c|}{FPT (\ref{lc-pre})}  & XALP-c.~\cite{outerkplanar} & XALP-c.~\cite{xalp} & XNLP-c.~\cite{xnlp1} \\ \hline

\textsc{Capacitated (Red-Blue) Dominating Set}
& \multicolumn{1}{c|}{FPT (\ref{CDS})}  & XALP-c.~\cite{outerkplanar} & XALP-c.~\cite{xalp} & XNLP-c.~\cite{xnlp2} \\ \hline

\textsc{Capacitated Vertex Cover}
& \multicolumn{1}{c|}{FPT (\ref{CVC})}  & XALP-c.~\cite{outerkplanar} & XALP-c.~\cite{xalp} & XNLP-c.~\cite{xnlp2} \\ \hline

\textsc{All-or-Nothing Flow}
& \multicolumn{1}{c|}{FPT (\ref{TOO})}  & XALP-c.~\cite{outerkplanar} & XALP-c.~\cite{xalp} & XNLP-c.~\cite{xnlp-flows} \\ \hline

\textsc{Target Outdegree Orientation}
& \multicolumn{1}{c|}{FPT (\ref{TOO})}  & XALP-c.~\cite{outerkplanar} & XALP-c.~\cite{xalp} & XNLP-c.~\cite{xnlp-flows} \\ \hline

\textsc{Min./Max.\ Outdegree Orientation}
& \multicolumn{1}{c|}{FPT (\ref{TOO})}  & XALP-c.~\cite{outerkplanar} & XALP-c.~\cite{xalp} & XNLP-c.~\cite{xnlp-flows} \\ \hline

\textsc{Circulating Orientation}
& \multicolumn{1}{c|}{FPT (\ref{TOO})}  & XALP-c.~\cite{outerkplanar} & XALP-c.~\cite{xalp} & XNLP-c.~\cite{xnlp-flows} \\ \hline

$f$-\textsc{Dominating Set}
& \multicolumn{1}{c|}{FPT (\ref{f-k-dom})}  & XALP-c.~\cite{outerkplanar} & XALP-c.~\cite{outerkplanar} & XNLP-c.~\cite{outerkplanar} \\ \hline

$q$-\textsc{Dominating Set}
& \multicolumn{1}{c|}{FPT (\ref{f-k-dom})}  & XALP-c.~\cite{outerkplanar} & XALP-c.~\cite{outerkplanar} & XNLP-c.~\cite{outerkplanar} \\ \hline 

\textsc{Target Set Selection}
& \multicolumn{1}{c|}{FPT (\ref{TSS})}  & XALP-c.~\cite{outerkplanar} & XALP-c.~\cite{outerkplanar} & XNLP-c.~\cite{outerkplanar} \\ \hline
\end{tabular}
\caption{Overview of parameterized complexity results for selected graph problems considered in this paper, under different structural parameters. $\operatorname{lcr}^\circ(G)$ denotes the convex local crossing number of a graph \(G\). Previously established hardness results and their references are indicated.}
\label{tab:xalp-overview}
\end{table}

In Table~\ref{tab:xalp-overview}, we use $\operatorname{lcr}^\circ(G)$ to denote the \emph{convex local crossing number} of $G$, that is, the smallest $k$ such that $G$ is outer $k$-planar, following the standard terminology in the study of beyond-planar graphs~\cite{schaefer, firman_et_al}. On the negative side, we show that certain problems that are XALP-complete when parameterized by treewidth, such as \textsc{Binary CSP} and \textsc{Scattered Set}, remain XALP-complete on outer $k$-planar graphs, when parameterized by $k$. On the positive side, assuming that an outer $k$-planar drawing of the input graph is given, we show that many problems that are XALP-complete with respect to treewidth become fixed-parameter tractable. Our approaches are mainly based on the triangulation technique of Firman et al.~\cite{firman_et_al}, combined with standard tools such as representative sets; see, for example, \cite[Chapter~12.3]{paralg-book} and \cite[Lemma 4.2]{CHANG2026492}. The key ingredient underlying our FPT algorithms is the following lemma.

\begin{lemma}[{\cite[Lemma 6]{firman_et_al}}]\label{crossing-lemma}
For every $k \ge 0$, given an outer $k$-planar drawing of a graph $G$ with $n$ vertices, the outer cycle admits a triangulation such that each triangulation edge crosses the original graph at most $k$ times. Moreover, given the intersection graph of the edges of $G$, such a triangulation can be constructed in $O(nk)$ time.
\end{lemma}

The triangulation technique described in Lemma~\ref{crossing-lemma} was originally introduced by the authors in~\cite{firman_et_al} as an intermediate tool for proving upper bounds on the treewidth and separation number of outer $k$-planar graphs. We show that it can play a much more important role in the design of FPT algorithms for a large number of problems on outer $k$-planar graphs. Roughly speaking, the weak dual of the triangulated outer cycle forms a binary tree, where each tree node corresponds to a triangle in the triangulation. This tree provides a natural structure for bottom-up dynamic programming. More precisely, we can recursively merge subpolygons of the triangulated graph along the skeleton of the weak dual, thereby combining the vertices lying on the corresponding portions of the outer face, until the entire graph is eventually reconstructed.

Moreover, the triangulated graph has well-behaved geometric properties. It satisfies separator properties similar to those of standard tree decompositions, while Lemma~\ref{crossing-lemma} bounds the number of original graph edges crossing each triangulation edge, which in turn bounds the number of neighbors that many vertices can have outside the subpolygons in dynamic programming. This allows us to naturally adapt several XP algorithms on standard tree decompositions, not limited to those considered in this paper, into FPT algorithms on triangulated outer \(k\)-planar graphs. In Section~\ref{fpt-algorithms}, we describe in detail how to exploit Lemma~\ref{crossing-lemma} to achieve this.

Note that the XALP-hardness results established in~\cite{outerkplanar} cannot be transferred to outer $k$-planar graphs in a straightforward manner. This is because there is no general algorithm that transforms an arbitrary $k$-outerplanar graph into an outer $f(k)$-planar graph, where \(f\) is a computable function. As a counterexample, consider the complete bipartite graph $K_{2,n}$. This graph is $2$-outerplanar, but it is easy to verify that, for any placement of its vertices on the outer face, some edge must have $\Theta(n)$ crossings. Therefore, any graph containing an induced subgraph $K_{2,m}$ with $m$ not bounded by $k$ cannot be outer $k$-planar. This observation also explains why many reduction gadgets used to prove W[1]-hardness or XALP-hardness for many problems fail in the setting of outer $k$-planar graphs. For instance, all known hardness proofs for the \textsc{List Coloring} problem, including those of Fellows et al.~\cite{some-colorful}, inevitably introduce an induced subgraph of the form $K_{2,n}$ in their gadgets. This suggests that some problems that are hard in the classical parameterized setting may become tractable on outer $k$-planar graphs.

In addition to the XALP-completeness results and FPT algorithms, we also establish several structural results for outer \(k\)-planar graphs. Using Lemma~\ref{crossing-lemma}, we prove an upper bound of \(k+2\) on the mim-width of outer \(k\)-planar graphs (Theorem~\ref{mim-proof}). We show that bounded cut-width implies bounded outer \(k\)-planarity, which also yields the same conclusion for bandwidth (Theorems~\ref{cut-width-2k} and Corollary~\ref{bw-2k2}), and that bounded feedback edge set number likewise implies bounded outer \(k\)-planarity (Theorem~\ref{fes-6k}). Furthermore, we prove that many classical graph parameters are incomparable with outer \(k\)-planarity (see Table~\ref{tab:unbounded-parameters}). Together with the relationships among these well-studied graph parameters, this allows us to determine the relative position of outer \(k\)-planarity within the graph parameter hierarchy (see Figure~\ref{fig:para-hierarchy}).

\subsection{Organization}
The remainder of the paper is organized as follows. Section~\ref{section-pre} introduces the necessary preliminaries, including basic notation, definitions of several graph parameters, the complexity classes XNLP and XALP, as well as basic concepts from graph drawing and properties of outer \(k\)-planar graphs. Section~\ref{section-triangulation} briefly reviews the triangulation method of Firman et al.~\cite{firman_et_al} and presents the upper bound proof for mim-width. Section~\ref{XALP-results} contains XALP-completeness proofs for problems that remain intractable on outer \(k\)-planar graphs. Section~\ref{fpt-algorithms} introduces the main framework for designing dynamic programming algorithms on outer \(k\)-planar graphs, together with the underlying intuition, and presents all algorithms and reductions used to establish the fixed-parameter tractability of many problems. Section~\ref{sec:hierarchy} presents the structural results, mainly concerning boundedness and unboundedness relationships among various graph parameters. Finally, Section~\ref{sec:conclusion} concludes with directions for future research and a number of open problems.

\section{Preliminaries}\label{section-pre}
We assume the reader to be familiar with basic notions from algorithms, graph theory, and parameterized complexity, such as classes FPT, XP, W[1], W[2], \dots, W[P] (see \cite{DF99, DF13, flum06, paralg-book}).

Let $G$ be a graph. We write $V(G)$ and $E(G)$ for the vertex set and edge set of~$G$, respectively. For a vertex $v \in V(G)$, let $N_G(v)$ denote the set of neighbors of~$v$ in~$G$, and whenever the graph is clear from the context, we omit the subscript~$G$. The subgraph of~$G$ induced by a set of vertices $S \subseteq V$ is denoted by~$G[S]$. A vertex is called a \emph{leaf} if it has exactly one neighbor. Unless stated otherwise, all graphs considered in this paper are simple.

A pair of vertex sets $(A,B)$ is called a \emph{separation} of a graph $G$ if $A \cup B = V(G)$ and there is no edge between $A \setminus B$ and $B \setminus A$. A separation $(A,B)$ is said to be \emph{balanced} if both $|A \setminus B|$ and $|B \setminus A|$ are at most $2n/3$, where $n = |V(G)|$. For a (balanced) separation $(A,B)$, the set $A \cap B$ is called a (balanced) \emph{separator}, and $|A \cap B|$ is referred to as the \emph{order} of $(A,B)$.

We denote by $\mathbb{N}$ the set of natural numbers $\{0,1,2,\ldots\}$. For a positive integer $n$, let $[n] = \{1,2,\dots,n\}$, and for positive integers $a \le b$, we write \([a,b] = \{i \in \mathbb{N} \mid a \le i \le b \}\). Given a function $F : A \to B$ and a subset $A' \subseteq A$, we denote by $F|_{A'}$ the restriction of $F$ to $A'$, that is, the function from $A'$ to $B$ defined by $F|_{A'}(x) = F(x)$ for all $x \in A'$.

A \emph{drawing} of a graph maps each vertex to a distinct point in the plane, and each edge to a Jordan curve connecting the points corresponding to its endpoints. These curves are not allowed to pass through any vertex other than their endpoints, and any two such curves intersect in at most one point. A \emph{convex drawing} $\Gamma$ of a graph is a straight-line drawing in which all vertices are placed on distinct points of a circle, referred to as the \emph{circle of $\Gamma$}. An \emph{outer $k$-planar drawing} $\Gamma$ is a convex drawing in which every edge is crossed by at most $k$ other edges. Traversing the circle of $\Gamma$ in counterclockwise order induces a cyclic ordering of the vertices. We say that two pairs of vertices $\{v_1, w_1\}$ and $\{v_2, w_2\}$ are \emph{intertwined} in $\Gamma$ if their vertices appear in the cyclic order $(v_1, v_2, w_1, w_2)$ or $(v_2, v_1, w_2, w_1)$. Observe that two edges $\{v_1, w_1\}$ and $\{v_2, w_2\}$ cross in $\Gamma$ if and only if they have four distinct endpoints and are intertwined.

Given an outer $k$-planar drawing $\Gamma$ of a graph $G$, we define the \emph{outer cycle} of $G$ as the cycle obtained by connecting the vertices of $G$ one by one in the order in which they appear along the circle of $\Gamma$, even if some of the corresponding edges are not present in $G$.

In the definition of convex drawings, one could equivalently allow the vertices to be placed on any convex curve (instead of a circle) and the edges to be drawn as Jordan curves (instead of straight-line segments), provided that all curves lie inside the convex region and any two of them intersect in at most one point. We have the following simple observation.

\begin{observation}[\cite{firman_et_al}]
Let $G$ be a graph, let $\Gamma$ be an outer $k$-planar drawing of $G$, and let $c$ be a curve connecting two vertices $v$ and $w$ of $G$ that is entirely contained inside the circle of $\Gamma$. If $c$ crosses at most $\ell$ edges of $\Gamma$, then the straight-line segment $\overline{vw}$ also crosses at most $\ell$ edges of $\Gamma$.
\end{observation}
\begin{proof}
Every edge $\{v',w'\}$ of $G$ that is intertwined with $\{v,w\}$ must be crossed by the curve $c$. Therefore, $c$ has at least one crossing with every edge that is crossed by the straight-line segment $\overline{vw}$, which implies the claim.
\end{proof}

\subsection{XNLP and XALP}

A \emph{parameterized problem} is a language $L \subseteq \Sigma^* \times \mathbb{N}$ for some finite alphabet $\Sigma$. The class XNLP consists of parameterized problems that can be solved by a nondeterministic algorithm in time $f(k)n^{O(1)}$ and space $g(k)\log n$, for some computable functions $f$ and $g$, where $k$ is the parameter and $n$ is the size of the input. The class XALP is defined analogously as the class of problems that can be solved by a nondeterministic Turing machine with an additional stack in time $f(k)n^{O(1)}$ and space $g(k)\log n$. For more background on these two classes, as well as other alternative equivalent definitions of the class XALP, we refer to a series of recent works~\cite{xnlp1, xnlp2, xalp, xnlp-flows}.

As in classical complexity theory, the notions of hardness and completeness for a class are defined via reductions. A \emph{parameterized reduction} from a problem $L_1 \subseteq \Sigma_1^* \times \mathbb{N}$ to a problem $L_2 \subseteq \Sigma_2^* \times \mathbb{N}$ is a function $f : \Sigma_1^* \times \mathbb{N} \to \Sigma_2^* \times \mathbb{N}$ such that the following conditions are satisfied:
\begin{itemize}
    \item For all $(x,k) \in \Sigma_1^* \times \mathbb{N}$, we have $(x,k) \in L_1$ if and only if $f(x,k) \in L_2$;
    \item There exists a computable function $g : \mathbb{N} \to \mathbb{N}$ such that if $f(x,k) = (y,k')$, then $k' \le g(k)$.
\end{itemize}
A \emph{fixed-parameter tractable reduction} or \emph{fpt-reduction} is a parameterized reduction for which there is a deterministic algorithm that computes \(f(x, k)\) in time \(g(k)n^{O(1)}\). If the function $f(x,k)$ can be computed in space $O(g(k) + \log n)$, then the reduction is called a \emph{parameterized logspace reduction}, or \emph{pl-reduction}. If $f(x,k)$ can be computed in time $g(k)\,n^{O(1)}$ and space $O(h(k)\log n)$, then the reduction is called a \emph{parameterized tractable logspace reduction}, or \emph{ptl-reduction}.

The classes XNLP and XALP are closed under both pl-reductions and ptl-reductions. Note that a deterministic algorithm that uses \(O(g(k) + \log n)\) space necessarily runs in \(2^{O(g(k)+\log n)}\) time,  which is FPT with respect to \(k\). If there is a pl-reduction or ptl-reduction from problem \(A\) to problem \(B\), and \(B\) is in XNLP (respectively, XALP), then \(A\) is in XNLP (respectively, XALP) as well. A problem is XNLP-hard if every problem in XNLP admits a reduction to it, and XNLP-complete if it is both XNLP-hard and in XNLP; the same applies to the class XALP.

Hardness for the classes XNLP and XALP also has implications for the space requirements of deterministic parameterized algorithms. This is closely related to the following Conjecture~\ref{spsc} by Pilipczuk and Wrochna~\cite{spsc, xnlp1}, now known as the \emph{Slice-wise Polynomial Space Conjecture} (SPSC). The conjecture is considered plausible, as all known deterministic dynamic programming algorithms for those XNLP-hard or XALP-hard problems require DP tables of size $O(n^{f(k)})$.

\begin{conjecture}[SPSC \cite{spsc, xnlp1}]\label{spsc}
If a problem is XNLP-hard or XALP-hard, then there is no algorithm that solves it in XP time while simultaneously using fixed-parameter tractable space.
\end{conjecture}

Membership in XNLP (respectively, XALP) can often be established by transforming standard dynamic programming algorithms on path decompositions (respectively, tree decompositions) into nondeterministic algorithms that use only logarithmic space. Intuitively, the logarithmic space bound is achieved by guessing the relevant table entry when needed, rather than storing the entire table. Examples of such membership proofs can be found in~\cite{xnlp1,xnlp2,xalp}. In our hardness proofs, we omit the details showing the logarithmic space bounds of the reductions. These bounds can be obtained using standard techniques, where values are recomputed on demand instead of being stored explicitly. This may increase the running time by a polynomial factor, but ensures that the space usage of the reduction remains within the required logarithmic bounds.

\subsection{Graph Parameters}
\paragraph{Treewidth.}
A \emph{tree decomposition} of a graph $G=(V,E)$ is a pair 
$\mathcal{T} = (T, \{X_t\}_{t \in V(T)})$, where $T$ is a tree and each node $t \in V(T)$ is assigned a subset $X_t \subseteq V(G)$, called a \emph{bag}, such that the following three conditions hold:
\begin{itemize}
    \item $\bigcup_{t \in V(T)} X_t = V(G)$.
    \item For every edge $uv \in E(G)$, there exists a node $t \in V(T)$ such that $\{u,v\} \subseteq X_t$.
    \item For every vertex $v \in V(G)$, the set 
    $\{ t \in V(T) : v \in X_t \}$ induces a connected subtree of $T$.
\end{itemize}
The \emph{width} of a tree decomposition $\mathcal{T}$ is 
$\max_{t \in V(T)} |X_t| - 1$. The \emph{treewidth} of $G$, denoted $\mathrm{tw}(G)$, is the minimum width over all tree decompositions of $G$. A \emph{path decomposition} of a graph $G$ is a tree decomposition $\mathcal{T} = (T, \{X_t\}_{t \in V(T)})$ in which $T$ is a path. The \emph{pathwidth} of $G$ is the minimum width over all path decompositions of $G$. Kloks~\cite{Kloks} introduced the notion of \emph{nice} tree and path decompositions. In this paper, we mainly use a generalized variant of nice tree and path decompositions. Moreover, in all proofs, whenever a problem is parameterized by treewidth, we assume that a tree decomposition of the input graph is also given as part of the input.

A \emph{generalized path decomposition} can be described via a sequence of operations on terminal graphs. A \emph{terminal graph} is a triple $(V,E,X)$, together with a binary relation $\prec \,\subset X \times X$ that forms a strict total order on the set $X$. The elements of $X$ are called \emph{terminals}. We consider the following five types of operations on terminal graphs.

\begin{itemize}
\item \textbf{Introduce.} Given a terminal graph $G=(V,E,X)$, the introduce operation adds a new isolated vertex $v$ to both $V$ and $X$, where $v$ becomes the smallest element with respect to the ordering on $X$. Formally, $G'=(V\cup\{v\},E,\{v\}\cup X)$, with $v \prec x$ for all $x\in X$.

\item \textbf{Forget.} Given a terminal graph $G=(V,E,X)$ with $X\neq\emptyset$, the largest element of $X$ is removed from the set of terminals. Formally, $G'=(V,E,X\setminus\{x\})$, with $y\prec x$ for all $y\in X\setminus\{x\}$.

\item \textbf{Add-Edge$(i)$.} Given a terminal graph $G=(V,E,X)$ with $|X|>i$, we add an edge between the $i$th and the $(i+1)$st terminal. The sets $V$ and $X$ remain unchanged, and a single edge is added to $E$.

\item \textbf{Swap$(i)$.} Given a terminal graph $G=(V,E,X)$ with $|X|>i$, the operation swaps the $i$th and the $(i+1)$st terminals in the ordering of $X$.

\item \textbf{Join.} Given two terminal graphs $G_1=(V_1,E_1,X)$ and $G_2=(V_2,E_2,X)$ that intersect only on their terminals, that is, $V_1\cap V_2=X$ and $E_1\cap E_2=\emptyset$, we construct the graph $G=(V_1\cup V_2,E_1\cup E_2,X)$. This operation fuses the two terminal graphs by identifying, for each $i$, the $i$th terminal of the first graph with the $i$th terminal of the second graph.
\end{itemize}
Note that by applying a sequence of Swap operations, the terminals can be reordered arbitrarily. It follows that a nice path decomposition of a graph $G=(V,E)$ of pathwidth at most $k$ (see~\cite{Kloks}) can be transformed into a sequence of Introduce, Forget, Add-Edge, and Swap operations that constructs the terminal graph $G=(V,E,\emptyset)$ using $O(kn)$ operations, while ensuring that every intermediate terminal graph contains at most $k+1$ terminals. Similarly, a nice tree decomposition can be transformed into a sequence of Introduce, Forget, Add-Edge, Swap, and Join operations. From such a sequence, one can recover a \emph{generalized tree decomposition}, that is, a rooted tree in which each bag consists of the current set of terminals and is labeled by one of the operation types Introduce, Forget, Add-Edge, Swap, or Join. We omit the details and one can check that these transformations can be carried out in logarithmic space.

\paragraph{Mim-width.}
A \emph{branch decomposition} of a graph $G$ is a pair $(T,\delta)$, where $T$ is a subcubic tree, i.e. a tree of maximum degree at most 3, and $\delta$ is a bijection from $V(G)$ to the leaves of $T$. Every edge $e \in E(T)$ partitions the leaves of $T$ into two sets $L_e$ and $\overline{L}_e$, corresponding to the two components of $T - e$. This induces a partition $(A_e, \overline{A}_e)$ of $V(G)$, where $\delta(A_e) = L_e$ and $\delta(\overline{A}_e) = \overline{L}_e$. Let $G[A_e,\overline{A}_e]$ denote the bipartite subgraph of $G$ consisting of all edges with one endpoint in $A_e$ and the other in $\overline{A}_e$. A matching $F \subseteq E(G)$ is \emph{induced} if there is no edge between vertices belonging to different edges of $F$. We denote by $\mathrm{cutmim}_G(A_e,\overline{A}_e)$ the maximum size of an induced matching in $G[A_e,\overline{A}_e]$. The \emph{mim-width} of a branch decomposition $(T,\delta)$, denoted by $\mathrm{mimw}_G(T,\delta)$, is the maximum value of $\mathrm{cutmim}_G(A_e,\overline{A}_e)$ over all edges $e \in E(T)$. The \emph{mim-width} of $G$, denoted by $\mathrm{mimw}(G)$, is the minimum value of $\mathrm{mimw}_G(T,\delta)$ over all branch decompositions $(T,\delta)$ of $G$.

\paragraph{Cut-width and bandwidth.}
Let \(G\) be a graph. For a linear ordering \(\pi=(v_1,\ldots,v_n)\) of \(V(G)\), the cut-width of \(\pi\) is the maximum, over all \(i\in\{1,\ldots,n-1\}\), of the number of edges with one endpoint in \(\{v_1,\ldots,v_i\}\) and the other endpoint in \(\{v_{i+1},\ldots,v_n\}\). The cut-width of \(G\), denoted by \(\operatorname{cutw}(G)\), is the minimum of this value over all linear orderings of \(V(G)\). For a linear ordering \(\pi=(v_1,\ldots,v_n)\) of \(V(G)\), the bandwidth of \(\pi\) is the maximum value of \(|i-j|\) over all edges \(v_i v_j\in E(G)\). The bandwidth of \(G\), denoted by \(\operatorname{bw}(G)\), is the minimum of this value over all linear orderings of \(V(G)\).

\paragraph{Cyclomatic number.}
For a connected graph \(G\), its \emph{cyclomatic number} or \emph{feedback edge set number} is equal to \( |E(G)|-|V(G)|+1 \). Equivalently, it is the minimum number of edges whose deletion makes the graph acyclic. For a disconnected graph, the feedback edge set number is the sum of this quantity over all connected components, or equivalently \( |E(G)|-|V(G)|+\operatorname{cc}(G) \), where \(\operatorname{cc}(G)\) is the number of connected components of \(G\).

\section{Triangulation of Outer \(k\)-Planar Graphs}\label{section-triangulation}
In this section, we briefly review the triangulation technique of Firman et al.~\cite{firman_et_al}, and show how to use it to obtain an upper bound on the mim-width of outer \(k\)-planar graphs. Lemma~\ref{crossing-lemma} states that, given an outer \(k\)-planar drawing, we first connect all vertices on the outer face in their circular order to obtain an outer cycle. The polygon formed by this outer cycle admits a triangulation with the following property: each triangulation edge, which we refer to as a \emph{triangulation link} or \emph{cut link}, is crossed by at most \(k\) edges of the original graph. Moreover, such a triangulation can be efficiently computed in FPT time with respect to the parameter \(k\). Note that some triangulation links may coincide with edges of the original graph, as shown in Figure~\ref{fig:triang-tree}.

\begin{figure}[ht!]
    \centering
    \includegraphics[width=0.9\linewidth]{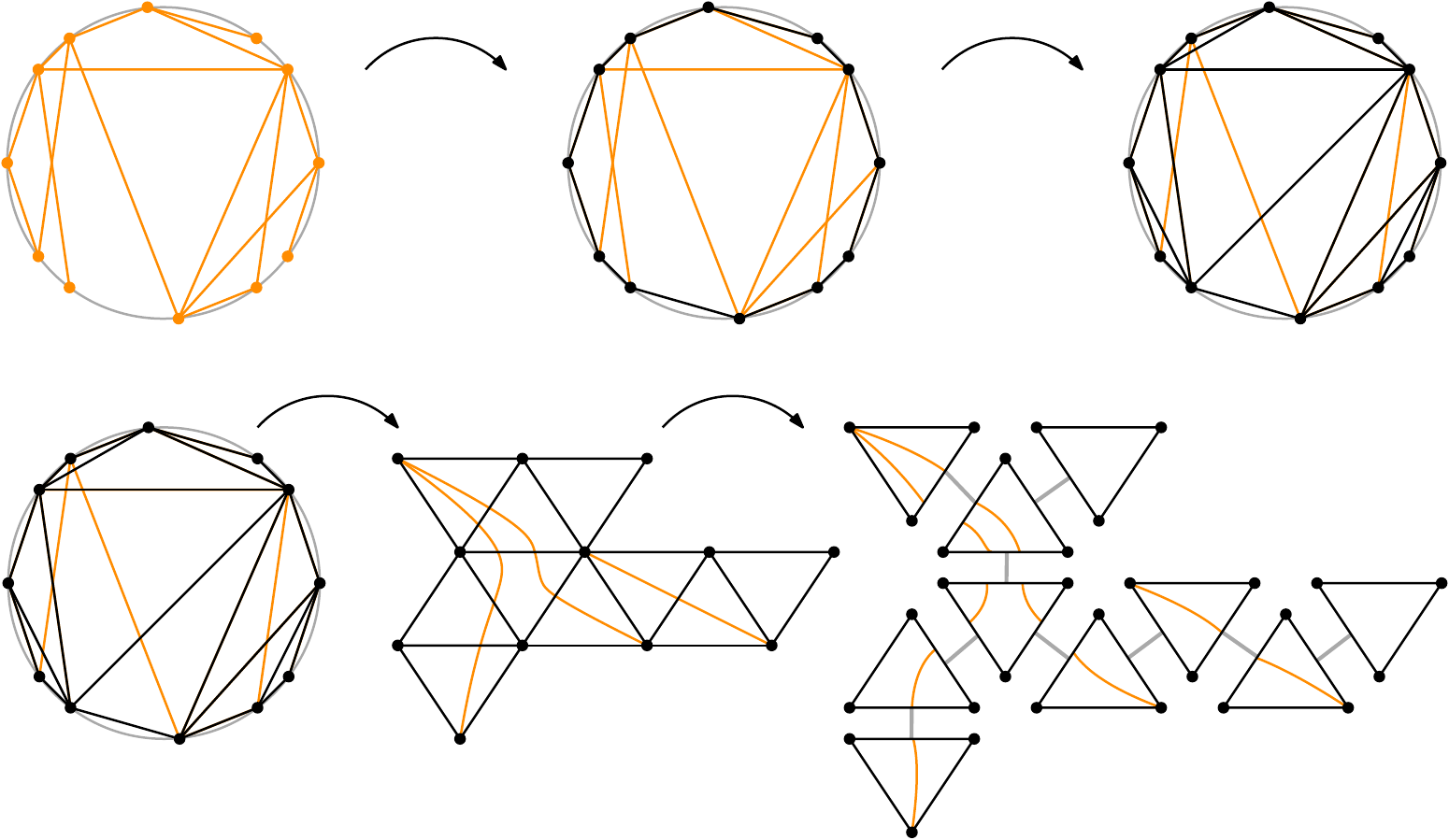}
    \caption{After triangulating the outer cycle, the weak dual of this triangulation is a tree, which means we can split the drawing into a tree of triangles, each of which crosses at most \(3k\)  edges of the original graph, which are marked in orange. All triangulation links are marked in black.}
    \label{fig:triang-tree}
\end{figure}

After triangulating the polygon enclosed by the outer cycle, the weak dual of the resulting triangulation is a tree. Furthermore, splitting the drawing yields a tree structure consisting of a collection of triangles, where each triangle is crossed by at most \(3k\) edges of the original graph. If we place the three vertices of each triangle into a bag, and for every edge that crosses some triangles we add one of its endpoints to all bags corresponding to the triangles it crosses, we obtain a na\"ive tree decomposition of width \(3k + 2\). With additional refined lifting operations, the authors of~\cite{firman_et_al} showed that this can be improved to obtain a treewidth upper bound of $1.5k + 2$.

Using the same triangulation technique, the authors of~\cite{firman_et_al} also proved that every outer $k$-planar graph admits a balanced separator of size at most $k + 2$. To see this intuitively, consider a triangulation link \(\lambda_t=a_tb_t\) that separates the triangulated graph into two parts. Let \(P_t\) denote the subpolygon on one side of \(\lambda_t\). By abuse of notation, we also use \(P_t\) to denote the set of vertices contained in this subpolygon, including the endpoints \(a_t\) and \(b_t\). By Lemma~\ref{crossing-lemma}, there are at most $k + 2$ vertices in $P_t$ that have neighbors outside $P_t$, since the link $\lambda_t$ is crossed by at most $k$ edges. Let $A_t$ be the set of endpoints of these crossing edges that lie inside $P_t$, then \(|A_t| \le k\). Note that $(P_t,\, (V(G)\setminus P_t) \cup \{a_t,b_t\} \cup A_t)$ forms a separation of $G$, whose separator $\{a_t, b_t\} \cup A_t$ has size at most $k+2$. Moreover, such a separator always exists and is balanced.

\paragraph{An Upper Bound on Mim-Width}\label{mimw-bound}
It is known that any graph of treewidth $k$ has mim-width at most $k+1$~\cite{mim-width-bound}. Combining this with the current best upper bound of $1.5k+2$ on the treewidth of outer $k$-planar graphs~\cite{firman_et_al} yields a trivial bound of $1.5k+3$ on their mim-width. We show that Lemma~\ref{crossing-lemma} can be exploited to obtain a tighter bound, namely that every outer $k$-planar graph has mim-width at most $k+2$.

\begin{theorem}\label{mim-proof}
Let \(G\) be an outer \(k\)-planar graph, and let \(\Gamma\) be an outer \(k\)-planar drawing of \(G\). Assume that \(\Gamma\) admits a triangulation of the outer cycle such that every triangulation link is crossed by at most \(k\) edges of \(G\). Then \(\operatorname{mimw}(G)\le k+2\).
\end{theorem}

\begin{proof}
Fix such a triangulation \(H\) of the outer cycle. Let \(D\) be the weak dual of \(H\). Since \(H\) is a triangulation of a polygon, \(D\) is a binary tree. For every node $t$ of $D$, let $\triangle_t$ denote the corresponding triangle in \(H\). Choose an outer-face link \(e_0=u_0v_0\) of \(H\), let \(\triangle_r\) be the unique triangle of \(H\) incident with \(e_0\), and root \(D\) at node \(r\). For every non-root triangle \(\triangle_t\), let \(\lambda_t=a_tb_t\) be the triangulation link shared by \(\triangle_t\) and its parent triangle, and we call \(\lambda_t\) the \emph{parent link} of \(\triangle_t\). Let \(c_t\) be the \emph{third vertex} of \(\triangle_t\). We write $P_t$ for the set of vertices on one side of $\lambda_t$ contained in a subpolygon, including the endpoints $a_t$ and $b_t$. After splitting this subpolygon, the weak dual of the resulting triangles corresponds to a subtree of $D$ rooted at $t$. For convenience, we also write \(S_t\coloneq P_t\setminus\{a_t,b_t\}\). For the root node \(r\), we have \(P_r\coloneq V(G)\) and \(S_r\coloneq V(G)\setminus\{u_0,v_0\}\).

We first construct a branch decomposition \((T,\delta)\) of \(G\). Observe that every vertex of \(G\) other than \(u_0\) and \(v_0\) appears as the third vertex of a unique triangle of \(H\). Indeed, if \(v\notin\{u_0,v_0\}\), then among all triangles of \(H\) containing \(v\), there is a unique one that is closest to the root triangle \(\triangle_r\); in that triangle, \(v\) cannot belong to the parent link, hence it is the third vertex.

For every node \(t\) of \(D\), we recursively construct a tree \(T_t\) with a distinguished vertex \(x_t\), called the attachment vertex of \(T_t\), such that the leaves of \(T_t\) are in bijection with the vertex set \(S_t\). Moreover, the leaf corresponding to vertex \(c_t\) will be denoted by \(\ell_t\).

If \(t\) is a leaf of \(D\), then \(S_t=\{c_t\}\). In this case, \(T_t\) consists of the single edge \(x_t\ell_t\). If \(t\) has exactly one child \(s\), then \(T_t\) is obtained from \(T_s\) by adding a new vertex \(x_t\) adjacent to \(x_s\) and to \(\ell_t\). The leaves of \(T_t\) are then exactly the leaves of \(T_s\) together with \(\ell_t\), hence they are in bijection with \(S_s\cup\{c_t\}=S_t\). If \(t\) has exactly two children \(s_1\) and \(s_2\), then we create two new vertices \(x_t\) and \(y_t\), and we add the edges \(x_t x_{s_1}, x_t y_t, y_t \ell_t\), and \(y_t x_{s_2}\). Thus \(T_t\) is obtained by joining \(T_{s_1}\), \(T_{s_2}\), and the leaf \(\ell_t\) through the edge \(x_t y_t\); see Figure~\ref{fig:mim-width} for an example. Again the leaves of \(T_t\) are exactly the leaves corresponding to the vertex set \(S_{s_1}\cup\{c_t\}\cup S_{s_2}=S_t\).

\begin{figure}[ht!]
    \centering
    \includegraphics[width=0.8\linewidth]{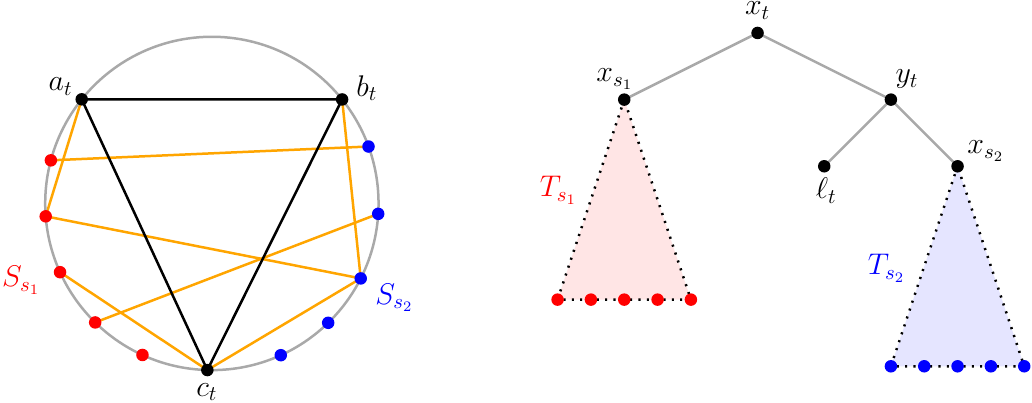}
    \caption{An example illustrating the bottom-up construction of the subcubic tree $T_t$ corresponding to a non-root node $t$, by joining $T_{s_1}$, $T_{s_2}$, and the leaf $\ell_t$. The left shows the graph \(G\), and the right shows the corresponding branch decomposition. The blue and red vertices represent the leaves of the subtrees $T_{s_1}$ and $T_{s_2}$, respectively. All vertices in \(S_t\) are mapped to leaves of \(T_t\). Some links, vertices, and edges of \(G[V(G)\setminus P_t]\) are not drawn for better visualization.}
    \label{fig:mim-width}
\end{figure}

Finally, we construct \(T\) by taking \(T_r\), adding two new leaves \(\ell_{u_0}\) and \(\ell_{v_0}\), and adding one new root vertex \(z\) adjacent to \(x_r\), \(\ell_{u_0}\), and \(\ell_{v_0}\). We define \(\delta\) by \(\delta(u_0)=\ell_{u_0}, \delta(v_0)=\ell_{v_0}\), and \(\delta(c_t)=\ell_t\ \text{ for every triangle }\triangle_t\). By construction, \(T\) is a subcubic tree rooted at \(z\), every internal vertex has degree at most \(3\), every leaf corresponds to exactly one vertex of \(G\), and every vertex of \(G\) corresponds to exactly one leaf of \(T\). Hence \((T,\delta)\) is a branch decomposition of \(G\).

We now bound the \(\operatorname{cutmim}\) value of every edge of \(T\). Let \(e\in E(T)\), and let \(X_e\subseteq V(G)\) be the set corresponding to one component of \(T-e\). By construction, every such set \(X_e\) is of one of the following forms. First, \(X_e=\{u_0\}\), \(X_e=\{v_0\}\), or \(X_e=\{c_t\}\) for some triangle \(\triangle_t\). In this case, \(\operatorname{cutmim}_G(X_e,\overline{X_e})\le 1 \le k + 2\). Second, \(X_e=S_t=P_t\setminus\{a_t,b_t\}\) for some non-root triangle \(\triangle_t\), where \(a_tb_t\) is the parent link of \(\triangle_t\). This includes both the case where \(e\) is the edge \(x_t x_{s_1}\) that attaches \(T_{s_1}\) to the rest of the branch decomposition and the case where \(e = y_t x_{s_2}\) is an internal edge of the branch decomposition separating the whole subtree of a child triangle from the rest. Third, if \(t\) has two children and \(e\) is the edge \(x_ty_t\), then one side of \(T-e\) corresponds to \(\{c_t\}\cup S_{s_2}=P_{s_2}\setminus\{b_t\}\), where \(s_2\) is the child corresponding to the link \(c_tb_t\). Finally, if \(e\) is the edge \(z x_r\), then \(X_e=\{u_0,v_0\}\) and \(\operatorname{cutmim}_G(X_e,\overline{X_e}) \le 2 \le k+2\).

Therefore, apart from singleton sets and the set \(\{u_0, v_0\}\), every set \(X_e\) has the form \(X_e = P\setminus Z\), where \(P\) is the vertex set on one side of some triangulation link \(xy\), and \(Z\subseteq\{x,y\}\). Fix such a set \(X_e=P\setminus Z\), and consider an induced matching in the bipartite graph \(G[X_e,\overline{X_e}]\). Every edge of this matching is an edge of \(G\) with one endpoint in \(X_e\) and the other in \(\overline{X_e}\).

We distinguish two kinds of matching edges. Suppose first that a matching edge \(uv\) is not incident with \(x\) or \(y\). Then \(u\) lies strictly on one side of the link \(xy\), while \(v\) lies strictly on the other side. Since all vertices lie on a circle and all edges are drawn as straight lines, the endpoints of \(uv\) and \(xy\) appear in alternating order on the circle. Hence \(uv\) crosses the link \(xy\). By the choice of the triangulation, at most \(k\) edges of \(G\) cross the link \(xy\). Therefore there are at most \(k\) matching edges not incident with \(x\) or \(y\). Now consider matching edges incident with \(x\) or \(y\). Any matching can contain at most one edge incident with \(x\) and at most one edge incident with \(y\). Hence any matching in \(G[X_e,\overline{X_e}]\) contains at most two matching edges of this type.

Combining the two observations, every induced matching in \(G[X_e,\overline{X_e}]\) has size at most \(k+2\). Therefore \(\operatorname{cutmim}_G(X_e,\overline{X_e})\le k+2\). Since this holds for every edge \(e\in E(T)\), we conclude that \(\operatorname{mimw}_G(T,\delta)\le k+2\). By the definition of mim-width, \(\operatorname{mimw}(G)\le \operatorname{mimw}_G(T,\delta)\le k+2\).
\end{proof}

\section{XALP-Completeness Results}\label{XALP-results}
In this section, we show that \textsc{Binary CSP} and \textsc{Scattered Set} remain XALP-complete on outer $k$-planar graphs. As discussed in Section~\ref{contribution}, both problems are known to be XALP-complete when parameterized by treewidth, and since every outer \(k\)-planar graph has treewidth at most \(1.5k+2\)~\cite{firman_et_al}, membership in XALP for these two problems parameterized by outer \(k\)-planarity follows immediately. Therefore, we focus on establishing their XALP-hardness via reductions. The reduction gadgets used in our XALP-hardness proof for \textsc{Binary CSP} are adapted from~\cite{outerkplanar}, following the construction based on generalized tree decompositions. The XALP-hardness proof for \textsc{Scattered Set} follows a similar strategy, and we borrow from~\cite{outerkplanar} the idea of using pendant path gadgets to forbid certain vertices from being selected into the scattered set.

\subsection{Binary CSP}\label{binary-csp}
Constraint Satisfaction Problems (CSPs) are problems defined by a set of variables, domains (possible values), and constraints that limit valid combinations, aiming to find a state where all constraints are satisfied. A Binary Constraint Satisfaction Problem (Binary CSP) is a type of problem where every constraint restricts the values of at most two variables, represented as a graph where vertices are variables and edges are constraints. Formally, it is defined as follows.

\vspace{3mm}
\noindent\fbox{
\begin{minipage}{0.97\textwidth}
\textsc{Binary CSP}\\
\textbf{Input:} A graph $G = (V,E)$, a set $\mathcal{C}$, a set $C(v) \subseteq \mathcal{C}$ for each $v \in V$, and a set
$C(u,v) \subseteq \mathcal{C} \times \mathcal{C}$ for each ordered pair $(u,v) \in V \times V$ such that $uv \in E$. \\
\textbf{Question:} Is there a function $f \colon V \rightarrow \mathcal{C}$ such that for every $v \in V$, $f(v) \in C(v)$,
and for every edge $uv \in E$, we have $(f(u), f(v)) \in C(u,v)$?
\end{minipage}}
\vspace{3mm}

We call the elements of $\mathcal{C}$ \emph{colors}, the sets $C(v)$ \emph{domains} (or \emph{vertex constraints}), and the sets $C(u,v)$ \emph{edge constraints}. We emphasize that $C(u,v)$ contains ordered pairs, that is, even though the graph is undirected, the ordering of vertices matters for the edge constraints. We may assume without loss of generality that $C(u,v) \subseteq C(u) \times C(v)$ for every edge $uv \in E$.

\begin{theorem}\label{bincsp-thm}
\textsc{Binary CSP} is XALP-complete on outer \(k\)-planar graphs parameterized by \(k\). 
\end{theorem}

\begin{proof}
We show a reduction from \textsc{Binary CSP} parameterized by treewidth. Let $G$ be an input graph with a nice tree decomposition of width $k$. From this instance, we construct an outer \(f(k)\)-planar graph $H$. As a first step, we transform the given nice tree decomposition into a generalized tree decomposition in which every bag is of one of the following types: Introduce, Forget, Add-Edge, Swap, or Join. For each bag $X$, let $\prec_X$ denote the ordering of terminals associated with $X$. For every bag $X$ and every terminal $v \in X$, we introduce a vertex $v_X$ in $H$. We draw the vertices of $H$ in the plane as follows. For each bag $X$, the vertices corresponding to the terminals in $X$ are placed on a single horizontal line, ordered from left to right according to $\prec_X$. For each bag $X$ and its child $Y$, the vertices corresponding to $X$ are drawn above those corresponding to $Y$. Intuitively, the drawing still maintains the tree decomposition skeleton.

Let $B_X = \{ v_X \mid v \in X \} \subseteq V(H)$. For every bag $X$ of the tree decomposition of $G$, we add edges to $H$ as follows. For each vertex $v \in V(G)$, if both bag $X$ and its child $Y$ contain $v$, then we add the edge $v_X v_Y$ to $H$. Furthermore, if $X$ is an Add-Edge bag corresponding to an edge $vw \in E(G)$, then we add the edge $v_X w_X$ to $H$, as illustrated in Figure \ref{fig:binary-csp}. 

\begin{figure}[ht!]
    \centering
    \includegraphics[width=0.65\linewidth]{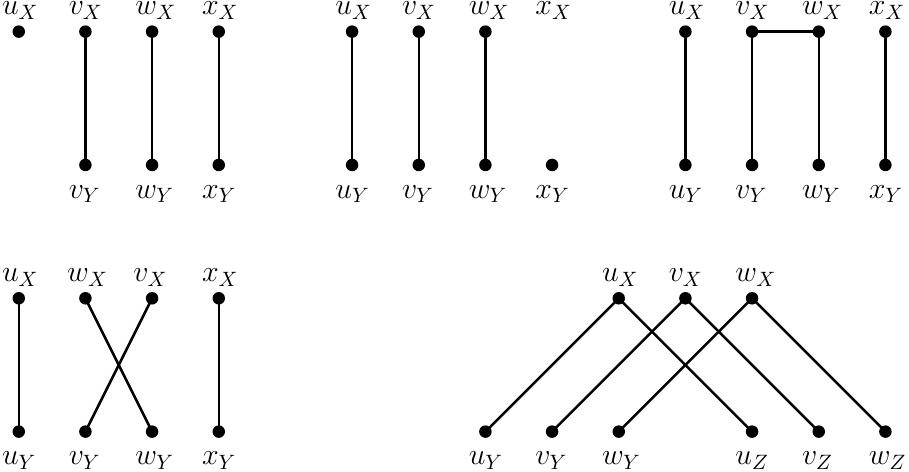}
    \caption{Constructing the graph \(H\) from a generalized tree decomposition of \(G\) by adding edges between vertices in parent and child bags, as well as between adjacent vertices inside Add-Edge bags. From left to right, top to bottom, the representations of the five bags are: Introduce $u$, Forget $x$, Add-Edge $vw$, Swap $vw$, and Join.}
    \label{fig:binary-csp}
\end{figure}

We now define the constraints of the \textsc{Binary CSP} instance on $H$. For each bag $X$ and each vertex $v_X \in B_X$, we set $C(v_X) = C(v)$. For each edge $v_X v_Y$ in $H$, we define \(C(v_X, v_Y) = \{(d,d) \mid d \in C(v)\}.\) For each edge $v_X w_X$ in $H$ corresponding to an original edge $vw \in E(G)$, we set \(C(v_X, w_X) = C(v,w).\) In any valid assignment of an instance on $H$, all vertices corresponding to the same vertex of $G$ must receive the same color. Moreover, for every edge of $G$, there is a corresponding Add-Edge bag in the decomposition, where the edge constraint is enforced. Consequently, any valid assignment for $H$ induces a valid assignment for $G$, and vice versa.

It remains to show that the graph $H$ admits an outer $f(k)$-planar drawing, where $k$ is the treewidth of $G$. We rely on the fact that every tree has an outerplanar drawing, that is, given a tree, all its nodes can be placed on a circle such that no tree edge is crossed. For example, consider a pre-order traversal of the bags of the tree decomposition. Along the outer face, we reserve a contiguous interval \(I_X\) for each bag \(X\), and place all vertices in \(B_X\) within this interval clockwise, following the order $\prec_X$. The intervals themselves are arranged clockwise along the outer face according to the pre-order traversal of the tree bags. See Figure~\ref{fig:bincsp-cross} for an example.

\begin{figure}[ht!]
    \centering
    \includegraphics[width=0.9\linewidth]{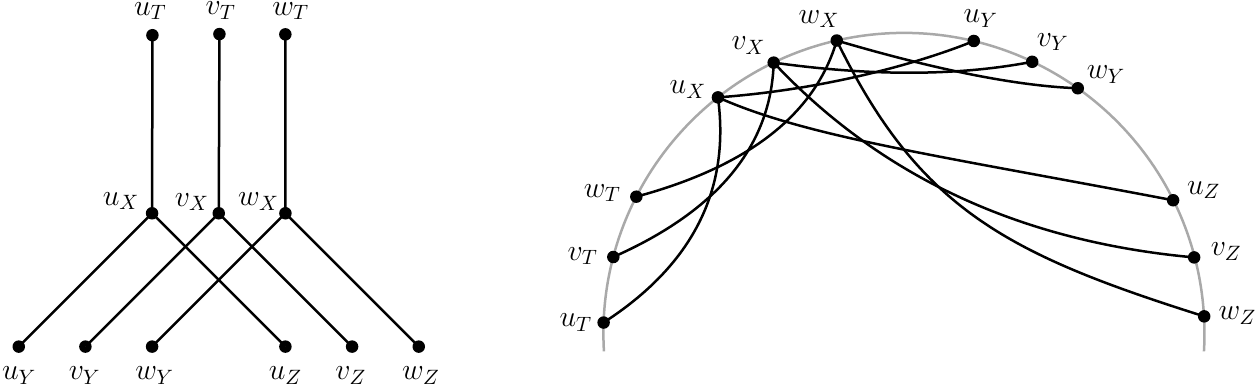}
    \caption{An example of edges incident with the vertices of a join bag \(X\), with treewidth \(k=2\). For each bag, a consecutive interval on the outer face is reserved for placing the at most \(k+1\) vertices contained in the bag. All intervals are arranged clockwise according to a pre-order traversal of the tree decomposition. In this example, every edge is crossed at most \(3k=6\) times.}
    \label{fig:bincsp-cross}
\end{figure}

Note that the tree decomposition is binary, which implies that a tree bag containing at most $k + 1$ vertices can be adjacent to at most three other bags. Consequently, all vertices of a single bag are connected to at most $3k+3$ other vertices in neighboring bags. The maximum number of crossings occurs on edges incident to vertices in join bags. We analyze the worst-case crossings on such edges to obtain an upper bound on crossings per edge in \(H\).

Observe that for every edge of the form \(v_{X}v_{Y}\) in $H$ connecting vertices belonging to two adjacent bags \(X\) and \(Y\), its two endpoints \(v_X\) and \(v_Y\) lie in two distinct intervals \(I_X\) and \(I_Y\) on the outer face. Each interval contains at most \(k+1\) vertices, which together are incident with at most \(3k+3\) edges. Hence, the vertices contained in the two intervals \(I_X\) and \(I_Y\) together are incident with at most \(5k+5\) distinct edges, where up to \(k+1\) of them are edges between vertices of \(B_X\) and \(B_Y\). Observe that every edge that may cross \(v_X v_Y\) must belong to this set of at most \(5k+5\) edges. Indeed, since the tree decomposition admits an outerplanar drawing, each of the at most \(k+1\) edges between vertices of \(B_X\) and \(B_Y\) cannot be crossed by any edge not contained in this set of at most \(5k+5\) edges. Finally, edges of the form $v_X w_X$ introduced at Add-Edge bags do not create crossings, since the two endpoints \(v_X\) and \(w_X\) are adjacent and lie within the same interval on the outer face. Therefore, each edge of $H$ is crossed at most $O(k)$ times.
\end{proof}

\subsection{Scattered Set}\label{scattered-set}
In this section, we consider the \textsc{Scattered Set} problem, which can be seen as a generalization of the classic \textsc{Independent Set} problem and is formally defined as follows.

\vspace{3mm}
\noindent\fbox{
\begin{minipage}{0.97\textwidth}
\textsc{Scattered Set}\\
\textbf{Input:} Graph \(G=(V,E)\), \(d \in \mathbb{N}\), \(m \in \mathbb{N}\). \\
\textbf{Question:}  Is there a set \(I \subseteq V\) of size \(m\) such that for any pair of distinct vertices \(u, v \in I\) the distance between \(u\) and \(v\) is at least \(d\)? 
\end{minipage}}
\vspace{3mm}

For any fixed $d \geq 3$, the problem parameterized by solution size \(m\) is W[1]-hard, even when restricted to bipartite graphs~\cite{Eto2014}. When $d = 2$, the problem coincides with \textsc{Independent Set}, which is known to be W[1]-complete~\cite{wh2}. For fixed $d$, the problem becomes fixed-parameter tractable when parameterized by treewidth, using standard dynamic programming techniques~\cite{scattered-set-fpt}. In this paper, we consider the setting where $d$ is part of the input and is not fixed.

\begin{theorem}
\textsc{Scattered Set} is XALP-complete on outer \(k\)-planar graphs parameterized by \(k\).
\end{theorem}

\begin{proof}
We show a reduction from \textsc{Scattered Set} parameterized by treewidth.  Let \(G\) be an input graph with a nice tree decomposition of width \(k\). From this instance, we construct an outer \(f(k)\)-planar graph \(H\) for a function \(f\). The proof idea is similar to that of Theorem~\ref{bincsp-thm}, and again relies on the fact that every tree admits an outerplanar embedding, which allows us to construct a drawing that avoids unnecessary edge crossings as much as possible.

Our goal is to show that graph $G$ contains a $d$-scattered set of size $m$ if and only if the constructed graph $H$ contains a $d'$-scattered set of size $m'$. As before, we first construct for $G$ a generalized tree decomposition \(T\) whose bags are of type Introduce, Forget, Add-Edge, Swap, or Join. For every bag $X$ and every terminal $v \in X$, we introduce a vertex $v_X$ in $H$. Let $B_X = \{ v_X \mid v \in X \} \subseteq V(H)$. For a bag $X$, the vertices in $B_X$ are arranged in \(H\) from left to right on a single horizontal line according to the order $\prec_X$. For each bag $X$ and its child $Y$ in the tree decomposition, the vertices in $B_X$ are drawn above the vertices in $B_Y$.

For every bag $X$ of the tree decomposition of $G$, we add edges to $H$ as follows. For each vertex $v \in V(G)$, if both bag $X$ and its child $Y$ contain $v$, then we add the edge $v_X v_Y$ to $H$. Unlike the proof of Theorem~\ref{bincsp-thm}, if $X$ is an Add-Edge bag corresponding to an edge $vw \in E(G)$, we connect $v_X$ and $w_X$ by a long path gadget instead of a single edge. In addition, we attach pendant paths to certain vertices of $H$ to ensure that specific forbidden vertices cannot be selected in any $d'$-scattered set, and such pendant path gadgets are also used in~\cite{outerkplanar}. A \emph{pendant path} of length $n$ attached to a vertex $v$ is a path $P=(u_0,u_1,\dots,u_n)$ such that $u_0 = v$, each internal vertex $u_i$ for $1 \le i \le n-1$ has degree $2$, and the terminal vertex $u_n$ has degree $1$.

For each vertex $v \in V(G)$, we select one occurrence of $v$ in an arbitrary bag $X$ that introduces $v$, and designate the corresponding vertex $v_X \in V(H)$ as the representative of $v$ in $H$. We denote this representative by $R_v$. All other copies of $v$ in \(H\) are referred to as non-representative vertices. Let $|V(T)|$ denote the size of \(T\), i.e. the total number of bags. Observe that the distance between any two bags in \(T\) is at most $|V(T)|$. We define \(A = d \cdot |V(T)| + 1\) and  \(d' = d \cdot A\). If $X$ is an Add-Edge bag corresponding to an edge $vw \in E(G)$, we connect $v_X$ and $w_X$ in \(H\) by a long path of length $A$. For all non-representative vertices in graph \(H\), as well as all vertices in the long path gadgets, we attach to each of them two vertex-disjoint pendant paths of length $d' - 1$.

The intuition behind this construction is as follows. If, in $G$, we want to reach a vertex $v$ from a vertex $u$ along a shortest path \(P_G\) of length $\ell$, then we must traverse all edges on \(P_G\), say $e_1, e_2, \ldots, e_{\ell}$. Now consider starting from the representative vertex $R_u$ and attempting to reach $R_v$ along a shortest path \(P_H\) in $H$. By the properties of tree decompositions, for every vertex $v \in V(G)$, the set of bags containing $v$ induces a connected subtree of $T$. This implies that in the constructed graph $H$, all copies of $v$ also form a connected subtree \(T_v\), since the copies of $v$ in adjacent bags are connected by edges. In $H$, the vertices $R_u$ and $R_v$ lie in two distinct connected subtrees \(T_u\) and \(T_v\), and different subtrees are connected only via the long path gadgets introduced in the Add-Edge bags. Therefore, any path from the subtree containing $R_u$ to the subtree containing $R_v$ must traverse some of these long path gadgets.

\begin{figure}[ht!]
    \centering
    \includegraphics[width=0.95\linewidth]{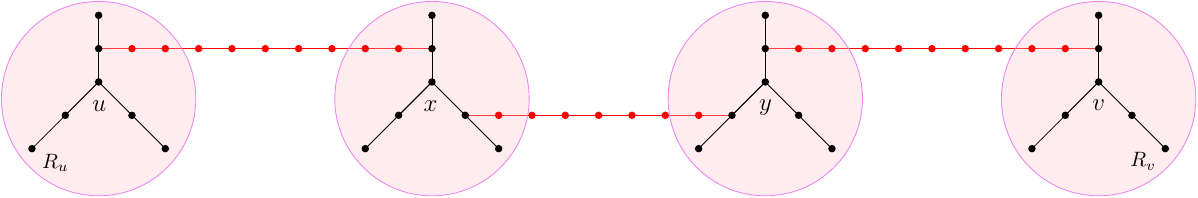}
    \caption{Suppose that the shortest path from \(u\) to \(v\) in \(G\) is \(u \rightarrow x \rightarrow y \rightarrow v\). In \(H\), all copies of each vertex \(i \in V(G)\) form a connected subtree \(T_i\). Different subtrees are connected by long path gadgets. There exists a path from \(R_u\) to \(R_v\) in \(H\) that traverses the long path gadgets corresponding to the Add-Edge bags for the edges \(ux\), \(xy\), and \(yv\). Each subtree \(T_i\) can be viewed as a super-vertex, represented by a pink circle. Long path gadgets are shown in red.}
    \label{fig:scatter-distance}
\end{figure}

By making these long path gadgets sufficiently long, while still keeping the construction polynomial in size, we ensure that the distance between two vertices \(R_u\) and \(R_v\) in \(H\) is determined mainly by the number of long path gadgets traversed when moving from \(R_u\) to \(R_v\), rather than by the distances traveled inside any subtree containing all copies of a particular vertex. For example, in Figure~\ref{fig:scatter-distance}, one may think of the subtree \(T_i\) containing all copies of a vertex \(i \in V(G)\) as a super-vertex \(S_i\) in \(H\). Two such super-vertices \(S_i\) and \(S_j\) are connected by a long path gadget if and only if the vertices \(i\) and \(j\) are adjacent in the original graph \(G\). Since the distance between any two vertices inside the same subtree \(T_i\) is much smaller than the length of a long path gadget, we can scale up distances between vertices of \(G\) in the graph \(H\), while preserving the relative ordering of distances between different pairs of vertices.

Formally, consider any two vertices \(u,v \in V(G)\), and suppose that there exists a shortest path \(P_G\) of length \(\ell\) with edges \(e_1,e_2,\ldots,e_\ell\) from \(u\) to \(v\) in \(G\). We now bound the distance between the corresponding representative vertices \(R_u\) and \(R_v\) in \(H\). To obtain a lower bound, observe that any shortest path \(P_H\) from \(R_u\) to \(R_v\) must traverse at least \(\ell\) long path gadgets. Indeed, every edge \(uv \in E(G)\) corresponds uniquely to a long path gadget contained in an Add-Edge bag for edge \(uv\). If \(P_H\) traverses fewer than \(\ell\) long path gadgets when moving from \(R_u\) to \(R_v\), then this would imply the existence of a path from \(u\) to \(v\) in \(G\) using fewer than \(\ell\) edges, contradicting the assumption that \(P_G\) is a shortest path. Hence, we have \(A \cdot \ell \le \mathrm{dist}_H(R_u, R_v)\).

To obtain an upper bound for \(\mathrm{dist}_H(R_u, R_v)\), consider a path in \(H\) starting from \(R_u\), then traversing the long path gadgets corresponding to the edges \(e_1,e_2,\ldots,e_\ell \in E(G)\), and finally reaching \(R_v\). Such a path exists and traverses exactly $\ell$ long path gadgets of length $A$ introduced at Add-Edge bags. In addition, it includes the distances incurred when moving from one Add-Edge bag to the next along tree edges inside the subtree \(T_i\), for vertices \(i\) appearing on the path \(P_G\) other than \(u\) and \(v\). There are $\ell - 1$ such transitions, together with the distances from $R_u$ to the first Add-Edge bag and from the last Add-Edge bag to $R_v$. Hence, there are $\ell + 1$ such additional transitions in total, and each of them has length at most $|V(T)|$. Therefore, \(A \cdot \ell \le \mathrm{dist}_H(R_u, R_v) \le A \cdot \ell + (\ell+1)\cdot |V(T)|\). Now consider any two vertices $u, v \in V(G)$. If $\mathrm{dist}_G(u,v) \le d-1$, then \(\mathrm{dist}_H(R_u, R_v) \le (d-1)\cdot A + d\cdot |V(T)|\). By the definition of $A$, we have $d\cdot |V(T)| < A$, and hence \(\mathrm{dist}_H(R_u, R_v) < (d-1)\cdot A + A = d\cdot A = d'\). On the other hand, if $\mathrm{dist}_G(u,v) \ge d$, then \(\mathrm{dist}_H(R_u, R_v) \ge d\cdot A = d'\). Therefore, we map and scale up the distance between any two vertices $u$ and $v$ in $G$ to the distance between their corresponding representatives $R_u$ and $R_v$ in $H$, so that $\mathrm{dist}_G(u,v) \ge d$ if and only if $\mathrm{dist}_H(R_u,R_v) \ge d'$.

\begin{figure}[ht!]
    \centering
    \includegraphics[width=0.9\linewidth]{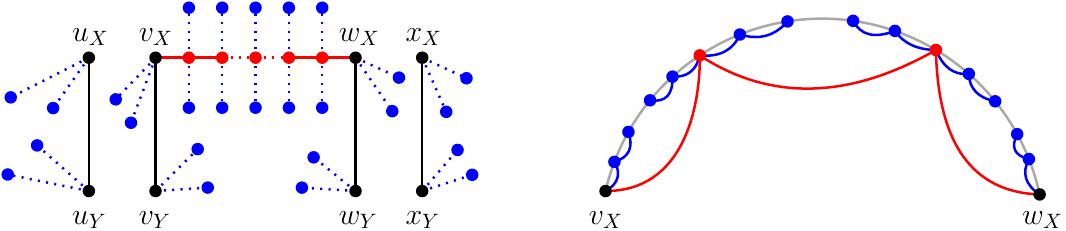}
    \caption{In an Add-Edge bag $X$, the vertices $v_X$ and $w_X$ are connected by a long path gadget (marked in red). All non-representative vertices and the vertices on the long path gadget are each attached to two pendant paths (marked in blue). The figure on the right illustrates how these pendant paths can be drawn on the outer face without introducing additional crossings.}
    \label{fig:scattered-set}
\end{figure}

Since we attach two pendant paths of length $d'-1$ to every non-representative vertex and to every vertex in the long path gadgets in $H$, none of these forbidden vertices can belong to a $d'$-scattered set. Indeed, if such a vertex is not selected, one can instead select the two terminal vertices of its pendant paths, whereas selecting the forbidden vertex prevents selecting all vertices on both pendant paths, resulting in a strictly smaller $d'$-scattered set. Therefore, if $G$ contains a $d$-scattered set $S$ of size $m$, then the representative vertices $R_v$ for each $v \in S$, together with all terminal vertices of the pendant paths, form a $d'$-scattered set of size $m'$ in $H$, and vice versa.

It remains to show that the graph $H$ is outer $f(k)$-planar, where $k$ is the treewidth of graph \(G\). Compared with the construction in the proof of Theorem~\ref{bincsp-thm}, the newly added pendant paths do not introduce additional crossings. Indeed, for any vertex $v$ to which two pendant paths are attached, these paths can be placed on the outer face locally to the left and right of $v$ without crossing any other edges. Moreover, since the edge in every Add-Edge bag \(X\) between $v_X$ and $w_X$ does not cross any edges, extending it into a path of length $A$ can be done by placing the vertices of the path consecutively along the outer face between $v_X$ and $w_X$, as shown in Figure~\ref{fig:scattered-set}.
\end{proof}

\section{FPT Algorithms on Triangulated Graphs}\label{fpt-algorithms}
In this section, we present algorithms for several problems that are FPT on outer $k$-planar graphs, when parameterized by $k$. We begin by explaining the intuition behind these algorithms, using classic treewidth-based dynamic programming as a guiding example. As discussed in Section~\ref{section-triangulation}, after applying the triangulation method of Firman et al.~\cite{firman_et_al} to the polygon enclosed by the outer cycle, the weak dual of the resulting triangulation is a tree. Once rooted, this tree provides a skeleton for designing bottom-up dynamic programming algorithms.

Observe that every triangulation link not on the outer face partitions the polygon enclosed by the outer cycle into two subpolygons, and the tree nodes of the weak dual corresponding to all triangles contained in one side form exactly a rooted subtree of the weak dual. Consider a triangle corresponding to a non-root node of the weak dual. Its three edges partition the polygon enclosed by the outer cycle into four regions. The two subpolygons defined by the two lower edges of the triangle together with the corresponding portions of the outer cycle correspond to the two subproblems in the dynamic programming, as illustrated in Figure~\ref{fig:lc-triang}.

\begin{figure}[ht!]
    \centering
    \includegraphics[width=0.9\linewidth]{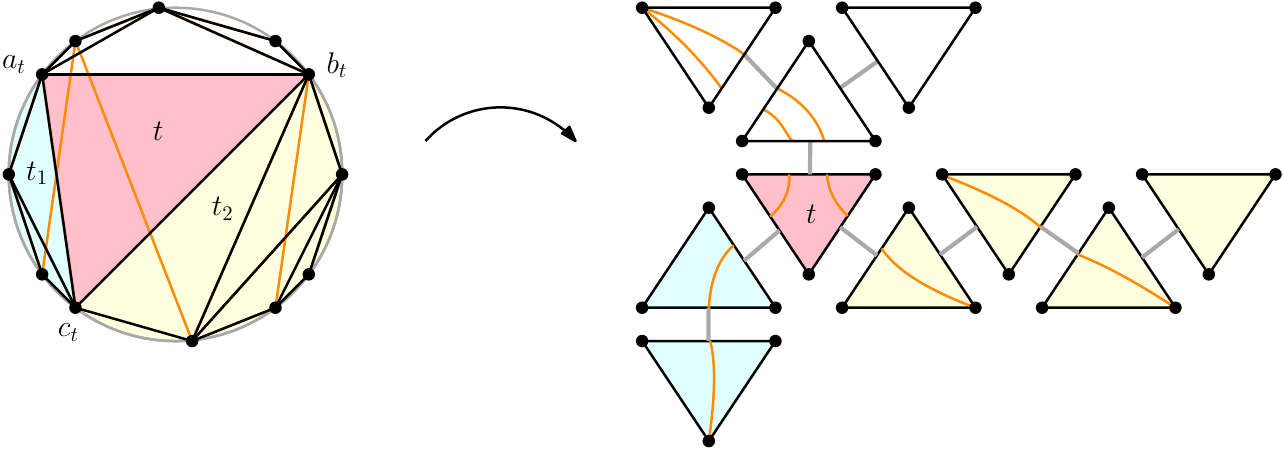}
    \caption{The triangulation link \(\lambda_t=a_tb_t\) separates the convex polygon into two subpolygons. The lower subpolygon consists of a collection of triangles whose corresponding nodes in the weak dual form a rooted subtree of the weak dual. Vertices in such a subpolygon can represent a subproblem in the dynamic programming. Subpolygons \(P_{t_1}\) and \(P_{t_2}\) denote the two subproblems induced by the links \(a_tc_t\) and \(b_tc_t\), respectively. Edges of the original graph are shown in orange.}
    \label{fig:lc-triang}
\end{figure}

The triangulated graph shown in Figure~\ref{fig:lc-triang} shares many similarities with standard tree decompositions. Classical dynamic programming on tree decompositions is based on the following principle. Given a tree decomposition \(T\) of a graph \(G\), each bag acts as a separator of size \(O(\mathrm{tw}(G))\). More importantly, the vertices contained in any subtree of \(T\) can interact with the rest of the graph only through the vertices in the root bag of that subtree. These vertices are typically referred to as the \emph{boundary vertices} of the corresponding subproblem. Consequently, for each subproblem represented by a subtree of \(T\), it suffices to store information about the boundary vertices in the DP table. This information is enough to combine partial solutions in a bottom-up manner and eventually obtain a final solution at the root of \(T\).

Outer $k$-planar graphs exhibit a similar property. Using the same notation as in Section~\ref{section-triangulation}, for a subpolygon \(P_t\) lying on one side of a triangulation link \(\lambda_t=a_tb_t\), we also use \(P_t\) to denote the set of all vertices contained in this subpolygon, including \(a_t\) and \(b_t\). Let \(A_t\) be the set of endpoints, lying in \(P_t\), of all edges that cross \(\lambda_t\). Note that \(\{a_t,b_t\}\cup A_t\) is a separator of \(G\) of size at most \(k+2\). Within $P_t$, only the vertices in $\{a_t, b_t\} \cup A_t$ may have edges connecting to vertices outside $P_t$, and hence they serve as the boundary vertices of the subproblem; see Figure~\ref{fig:separator}.

This means that for such a subpolygon $P_t$, it is sufficient to store information only for the vertices in $\{a_t, b_t\} \cup A_t$ in the DP table in order to combine with partial solutions from the outside. By processing the rooted weak dual tree in a bottom-up manner, we gradually merge subpolygons according to the tree skeleton, combining the corresponding boundary vertices along the outer cycle. In this process, each step corresponds to merging smaller subpolygons into larger ones, eventually reconstructing the entire polygon enclosed by the outer cycle. For instance, in Figure~\ref{fig:lc-triang}, merging the subpolygons $P_{t_1}$ and $P_{t_2}$ yields the subpolygon $P_t$, and this merging process continues recursively along the rooted weak dual tree until the full graph is obtained at the root. At the root, there exists an outer-face link $\lambda_0 = a_0 b_0$ that is not crossed by any edge, and the subpolygon defined by $\lambda_0$, together with $a_0$ and $b_0$, coincides with the entire graph $G$. Therefore, the final solution can be obtained at the root as well.

\begin{figure}[ht!]
    \centering
    \includegraphics[width=0.7\linewidth]{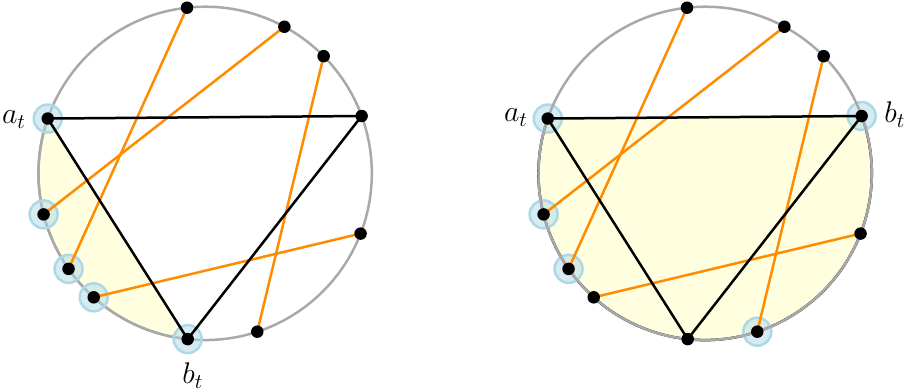}
    \caption{For the subpolygon (or equivalently, subproblem) $P_t$, including the vertices $a_t$ and $b_t$, lying below a triangulation link $\lambda_t = a_t b_t$, let $A_t$ be the set of endpoints, inside $P_t$, of the edges that cross $\lambda_t$. Then $(P_t,\, (V(G)\setminus P_t) \cup \{a_t,b_t\} \cup A_t)$ forms a separation of $G$, whose separator $\{a_t,b_t\} \cup A_t$ has size at most $k+2$. The vertices in the separator are highlighted with blue circles. For each subproblem in the dynamic programming, it suffices to store information only for these at most \(k+2\) boundary vertices. Edges of the original graph are shown in orange.}
    \label{fig:separator}
\end{figure}

The running time of treewidth-based algorithms is mainly determined by the size of the DP table at each bag, or more precisely, by the amount of information stored for each vertex in the bag. If $G$ has treewidth $k$ and each vertex in the bag requires $f(k)$ states for some function \(f\), then the overall running time is typically FPT. In contrast, if each vertex needs to store $O(n)$ possible values, the running time becomes $O(n^{k+1})$, which is only XP. Outer $k$-planar graphs admit an additional stronger geometric restriction. By Lemma~\ref{crossing-lemma}, for a subpolygon $P_t$ defined by a triangulation link $\lambda_t = a_t b_t$, only the two endpoints $a_t$ and $b_t$ may have many neighbors outside $P_t$, while every other vertex in $P_t$ has at most $k$ such neighbors. Consequently, among the boundary vertices, only $a_t$ and $b_t$ may have unbounded interaction with the outside, whereas each of the remaining boundary vertices has at most $k$ neighbors outside the subpolygon \(P_t\). This is a crucial observation: in many treewidth-based dynamic programming algorithms, the amount of information stored for each boundary vertex depends on its interactions with vertices outside the current subtree. For example, in the \textsc{Capacitated Dominating Set} problem, one needs to record, for each vertex in a bag, the amount of capacity that has already been used, or equivalently, the remaining capacity. In standard treewidth-based algorithms, this value can be as large as $O(n)$, since for a vertex \(v\) in a bag $X$, the number of neighbors of \(v\) outside the subtree rooted at $X$ can be as large as $O(n)$ and is not bounded by the treewidth. Therefore, this leads to an $n^{O(k)}$ XP algorithm. In essence, a tree decomposition only guarantees a small number of boundary vertices, at most $O(\mathrm{tw}(G))$, but does not impose any restriction on how many neighbors each boundary vertex may have outside the current subtree.

However, in outer $k$-planar graphs, for every subpolygon \(P_t\), except for two endpoints $a_t$ and $b_t$ of the triangulation link \(\lambda_t\), each boundary vertex has at most $k$ neighbors outside the subproblem (subpolygon \(P_t\)), and hence can contribute at most $k$ units of capacity usage in future computations. Therefore, it suffices to record its remaining capacity up to $O(k)$, while only $a_t$ and $b_t$ require values up to $O(n)$. As a result, the DP table size of each subproblem becomes $O(n^2 \cdot k^k)$, which yields an FPT algorithm. More generally, the geometric property captured in Lemma~\ref{crossing-lemma} allows us to reduce the DP table size of boundary vertices from XP to FPT in many problems. Moreover, even in cases where this geometric property alone is insufficient, it may still be possible to achieve such a reduction using additional techniques, such as representative sets. We will see an example of this in Section~\ref{lc-pre} for the \textsc{List Coloring} problem.

On the other hand, the geometric restriction from Lemma~\ref{crossing-lemma} is not helpful for all problems. In Section~\ref{XALP-results}, we showed that \textsc{Binary CSP} and \textsc{Scattered Set} remain XALP-complete. The reason is that, even though a boundary vertex has at most $k$ neighbors outside the subpolygon, it still needs to store $O(n)$ size information. For instance, in \textsc{Binary CSP}, each vertex must maintain all possible values in its domain, which can be of size $O(n)$, in order to ensure consistency with constraints involving external vertices. For \textsc{Scattered Set}, one must store, for each boundary vertex, the minimum distance to vertices that are currently selected in the scattered set, as in classical algorithms, leading to a DP table of size $d^{O(k)}$, where $d$ is part of the input and can be as large as \(O(n)\). Consequently, the running time remains XP.


\subsection{List Coloring and Precoloring Extension}\label{lc-pre}
We present an FPT algorithm for \textsc{List Coloring}, and show that \textsc{Precoloring Extension} is also fixed-parameter tractable by reducing it to \textsc{List Coloring}. Both problems can be viewed as special cases of \textsc{Binary CSP}. In particular, an instance of \textsc{List Coloring} can be interpreted as a \textsc{Binary CSP} instance in which, for each edge $uv$, the constraint is $C(u,v) = \{(a,b) : a \neq b\}$. Similarly, an instance of \textsc{Precoloring Extension} corresponds to a \textsc{Binary CSP} instance where each vertex has a domain of size either $1$ or $|C|$. Since vertices whose lists are larger than their degree can always be colored, \textsc{List Coloring} is fixed-parameter tractable when parameterized by treewidth plus maximum list size~\cite{JANSEN1997135}. We also remark that \textsc{List Coloring} can be solved in logarithmic space on trees~\cite{lc-tree}. Formally, the problem is defined as follows.

\vspace{3mm}
\noindent\fbox{
\begin{minipage}{0.97\textwidth}
\textsc{List Coloring}\\
\textbf{Input:} A graph $G = (V,E)$, a set of colors $\mathcal{C}$, and a color list $L(v) \subseteq \mathcal{C}$ for each $v \in V$. \\
\textbf{Question:}  Is there a function $f \colon V \rightarrow \mathcal{C}$ such that for every $v \in V$, $f(v) \in L(v)$, and for every edge $uv \in E$, $f(u) \neq f(v)$?
\end{minipage}}
\vspace{3mm}

In addition to the triangulation obtained by Lemma~\ref{crossing-lemma}, we employ the representative set technique introduced by Chang et al.~\cite{CHANG2026492} to compress the number of colorings of the boundary vertices. We first describe the intuition and key idea that lead to an FPT algorithm. For the \textsc{List Coloring} problem, classical approaches, such as dynamic programming on a tree decomposition, suffer from the fact that each vertex may have a list of size $O(n)$. As a result, the number of possible colorings for a bag of size $O(\mathrm{tw}(G))$ is $n^{O(\mathrm{tw}(G))}$. This also leads to a DP table of size $n^{O(\mathrm{tw}(G))}$, which explains why such algorithms for \textsc{List Coloring} typically run in XP time.

After triangulating an outer \(k\)-planar graph, consider a subpolygon \(P_t\) defined by a triangulation link \(\lambda_t = a_t b_t\). Within \(P_t\), only the two endpoints \(a_t\) and \(b_t\) may have many edges connecting to vertices outside \(P_t\), whereas every other vertex in \(P_t\) has at most \(k\) neighbors outside \(P_t\), by Lemma~\ref{crossing-lemma}. This observation suggests that if we can compress the number of coloring configurations of the boundary vertices in \(P_t \setminus \{a_t,b_t\}\) into a function \(g(k)\) that depends only on \(k\), while explicitly keeping all possible color assignments for the two endpoints \(a_t\) and \(b_t\), then the size of the DP table for each subproblem becomes \(O(n^2 \cdot g(k))\). Consequently, this yields an FPT algorithm with respect to \(k\). The following Lemma~\ref{lemma4.2} allows us to implement this idea.

For disjoint sets \(S, T\) of vertices in graph \(G\) and coloring functions \(g : S \to \mathcal{C}\) and \(h : T \to \mathcal{C}\), we say that \((S,g)\) is \emph{compatible} with \((T,h)\) if for every edge \(vw\) with \(v \in S\) and \(w \in T\), \(g(v) \neq h(w)\). If \(S\) and \(T\) are clear from the context, we simply say that colorings \(g\) and \(h\) are compatible.

\begin{lemma}[{\cite[Lemma 4.2]{CHANG2026492}}]\label{lemma4.2}
Let \(q\) and \(t\) be positive integers. Let \(G\) be a graph and let
\(S\) and \(T\) be disjoint sets of vertices in \(G\) such that
\(|\{v \in S : N_G(v) \cap T \neq \emptyset\}| \le q\) and
\(|\{v \in T : N_G(v) \cap S \neq \emptyset\}| \le t\).
Then for every set \(\mathcal{F}\) of proper colorings of \(G[S]\),
there is a subset \(\mathcal{F}^\ast \subseteq \mathcal{F}\) of size at most
\(2^{\frac{q(q+1)}{2}} t^{q-1}(t+1)\) such that for every proper coloring \(h\) of \(G[T]\), if there is a pair
\((S,g)\) with \(g \in \mathcal{F}\) that is compatible with \((T,h)\),
then there is a pair \((S,g^\ast)\) with \(g^\ast \in \mathcal{F}^\ast\)
that is compatible with \((T,h)\), and such a set \(\mathcal{F}^\ast\) can be computed in time
\(O(|\mathcal{F}| 2^{q^2} t^{q+2})\).
\end{lemma}

\begin{figure}[ht!]
    \centering
    \includegraphics[width=0.8\linewidth]{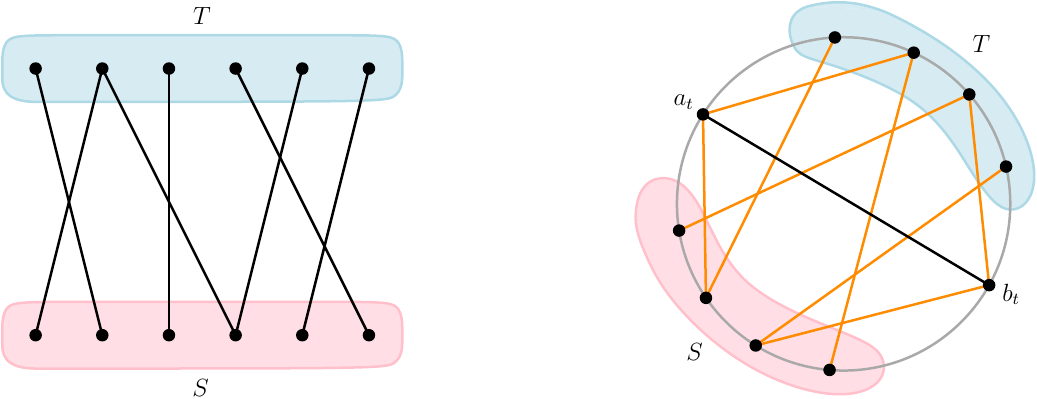}
    \caption{The endpoints of all edges crossing the triangulation link \(\lambda_t=a_tb_t\) on the two sides form disjoint sets, each of size at most \(k\). After fixing a color pair for \(a_t\) and \(b_t\), we use Lemma~\ref{lemma4.2} to compute a representative set \(\mathcal{F^\ast}\) for the set of all proper colorings of \(G[S]\) that can be extended to a proper \(L\)-coloring of \(G[P_t]\). The size of \(\mathcal{F^\ast}\) depends only on \(k\). Since there are at most \(O(n^2)\) possible color pairs for \(a_t\) and \(b_t\), the resulting DP table size of each subproblem is FPT.}
    \label{fig:rep-set}
\end{figure}

Lemma~\ref{lemma4.2} shows that whenever two vertex sets \(S\) and \(T\) are disjoint, any set \(\mathcal{F}\) of proper colorings of \(G[S]\) can be compressed into a representative set \(\mathcal{F}^{\ast}\) while preserving exactly the same compatibility with colorings of \(G[T]\). Moreover, the size of \(\mathcal{F}^{\ast}\) depends only on the number of vertices in \(S\) and \(T\) that are adjacent to vertices on the other side.

Consider the triangulation of an outer \(k\)-planar graph obtained from Lemma~\ref{crossing-lemma}. Let \(\lambda_t=a_tb_t\) be a triangulation link not on the outer face, which partitions the polygon into two sides, as illustrated in Figure~\ref{fig:rep-set}. Suppose that \(S\) and \(T\) denote the sets of vertices on the two sides that are incident with edges crossing \(\lambda_t\), and that the subpolygon \(P_t\) containing \(S\) corresponds to a subproblem of the dynamic programming. Fix a coloring of \(a_t\) and \(b_t\). Let \(\mathcal{F}\) be the set of all proper colorings of \(G[S]\) that can be extended to a proper \(L\)-coloring of \(G[P_t]\), where \(P_t\) includes the vertices \(a_t\) and \(b_t\). Using Lemma~\ref{lemma4.2}, we can compute a representative set \(\mathcal{F}^{\ast}\) for \(\mathcal{F}\) in a bottom-up manner. Since there are at most \(O(n^2)\) possible color pairs for \(a_t\) and \(b_t\), we compute at most \(O(n^2)\) such representative sets for a subproblem corresponding to \(P_t\).

The underlying idea is as follows. Suppose the input is a yes-instance, and let \(f\) be a proper \(L\)-coloring of \(G\). Restricting \(f\) to the subgraph \(G[(V(G)\setminus P_t)\cup\{a_t,b_t\}]\) yields a proper \(L\)-coloring \(f_1\). Restricting \(f\) to \(G[S]\) and \(G[T]\) yields proper \(L\)-colorings \(f_2\) and \(f_3\), respectively. Consider the color pair \((f_1(a_t),f_1(b_t))\). For this pair, we have already computed a representative set \(\mathcal{F}^{\ast}\) of all proper colorings of \(G[S]\) that can be extended to a proper \(L\)-coloring of \(G[P_t]\). By Lemma~\ref{lemma4.2}, there exists a coloring \(f_2^{\ast}\in\mathcal{F}^{\ast}\) that is compatible both with the coloring of the vertices in \(P_t\) and with the coloring \(f_3\) of \(G[T]\). Consequently, replacing \(f_2\) by \(f_2^{\ast}\) in \(f\) still yields a proper \(L\)-coloring of the entire graph \(G\). By Lemma~\ref{crossing-lemma}, both \(S\) and \(T\) have size at most \(k\). Therefore, the size of each representative set depends only on \(k\).

We now present the formal details of the algorithm and explain how these representative sets can be computed in a bottom-up fashion within FPT time, eventually yielding an FPT algorithm. Since the triangulation described in Lemma~\ref{crossing-lemma} can be computed efficiently in FPT time, in all subsequent algorithms presented in this section, in addition to an outer $k$-planar drawing of the input graph, we may assume that such a triangulation is also given as part of the input.

\begin{theorem}\label{list-coloring-thm}
Let $G$ be a connected graph together with an outer $k$-planar drawing and a triangulation \(H\) of the outer cycle such that every triangulation link is crossed by at most $k$ edges of $G$. Then \textsc{List Coloring} on $G$ is FPT when parameterized by \(k\).
\end{theorem}

\begin{proof}
Let $D$ be the weak dual of the triangulation $H$. For every node $t$ of $D$, let $\triangle_t$ denote the corresponding triangle. Choose an outer-face link $e_0=a_r b_r$, and root $D$ at the node \(r\), which corresponds to a unique triangle \(\triangle_r\) that is incident with $e_0$. If $t$ is not the root, let $\lambda_t=a_t b_t$ be the parent link of \(\triangle_t\), i.e. the triangulation link shared by $\triangle_t$ and its parent triangle. Let \(P_t\) be defined the same as before. We write $S_t \coloneq P_t\setminus\{a_t,b_t\}$. For a proper coloring $f$ of \(G\) and a vertex set $X\subseteq V(G)$, we write $f|_X$ for the restriction of $f$ to the set $X$.

For the root triangle $\triangle_r$, we regard the outer-face link $e_0=a_r b_r$ as its parent link, and then $P_r=V(G)$. For every non-root node $t$, let $A_t$ be the set of vertices in $S_t$ that are adjacent to some vertex in $V(G)\setminus P_t$. Since every edge joining $S_t$ with $V(G)\setminus P_t$ must cross $\lambda_t$, every such edge contributes exactly one endpoint to $A_t$, and therefore $|A_t|\leq k$. Similarly, if $C_t$ denotes the set of endpoints outside $P_t$ of the edges joining $S_t$ with $V(G)\setminus P_t$, then $|C_t|\leq k$ as well.

We now define the states. Fix a node $t$ and a coloring $\tau$ of the two endpoints of $\lambda_t$, that is, a function $\tau:\{a_t,b_t\}\to\mathcal{C}$ with $\tau(a_t)\in L(a_t)$ and $\tau(b_t)\in L(b_t)$. Let $Q[t,\tau]$ be the set of all colorings $g$ on $A_t$ such that there exists a proper $L$-coloring $f$ of $G[P_t]$ satisfying $f(a_t)=\tau(a_t)$, $f(b_t)=\tau(b_t)$, and $f|_{A_t}=g$. Thus every element of $Q[t,\tau]$ is \emph{valid} in the sense that it can be extended to a proper $L$-coloring of the whole subpolygon $G[P_t]$. For every node $t$ and every $\tau$, we compute a representative set $Q^*[t,\tau]\subseteq Q[t,\tau]$ satisfying the following property:
\[
(\star_{t,\tau})\quad
\begin{minipage}[t]{0.9\linewidth}
for every proper $L$-coloring $h$ of $G[(V(G)\setminus P_t)\cup\{a_t,b_t\}]$ with $h(a_t)=\tau(a_t)$ and $h(b_t)=\tau(b_t)$, if there exists a coloring $g\in Q[t,\tau]$ compatible with $h$, then there exists a coloring $g^*\in Q^*[t,\tau]$ compatible with $h$.
\end{minipage}
\]
Here compatibility means that for every edge $uv$ with $u\in A_t$ and $v\in C_t \cup \{a_t,b_t\}$, the two colors assigned to $u$ and $v$ are distinct. Since $A_t$ and $C_t$ are disjoint and both have size at most $k$, Lemma~\ref{lemma4.2} applies directly. In particular, for every node $t$ and every $\tau$, from any set of valid colorings on $A_t$ we can compute in time $O(|Q[t,\tau]|\cdot 2^{k^2}k^{k+2})$ a representative set of size at most \(\rho(k)=2^{k(k+1)/2}k^{k-1}(k+1)\). Thus we may always assume $|Q^*[t,\tau]|\leq \rho(k)$.

\paragraph{Leaf nodes.}
We compute these sets \(Q^*[t,\tau]\) bottom-up along the rooted tree $D$. Assume first that $t$ is a leaf. Then $\triangle_t$ is a triangle with vertices $\{a_t,b_t,c_t\}$, and $P_t=\{a_t,b_t,c_t\}$. For every coloring $\tau$ of $\{a_t,b_t\}$, we enumerate all colors $\gamma\in L(c_t)$ that are compatible with $\tau$ on the edges of $G[\{a_t,b_t,c_t\}]$. Every such choice determines a proper $L$-coloring of $G[P_t]$. If $c_t\in A_t$, this yields one coloring of $A_t$; otherwise it yields the unique empty coloring. Collecting all these colorings gives $Q[t,\tau]$, and then we apply Lemma~\ref{lemma4.2} to obtain $Q^*[t,\tau]$. Since there are at most $n^2$ choices for $\tau$ and at most $n$ choices for $\gamma$, the total running time for a leaf is $O(n^3\cdot 2^{k^2}k^{k+2})$.

\paragraph{Internal nodes.}
Now suppose that $t$ is an internal node in \(D\). Let $\triangle_t=\{a_t,b_t,c_t\}$, where $\lambda_t=\{a_t,b_t\}$ is the parent link. The other two sides of the triangle are the links $\lambda_{t_1}=\{a_t,c_t\}$ and $\lambda_{t_2}=\{c_t,b_t\}$ corresponding to the children of $t$ that exist in $D$. If one of these sides lies on the outer cycle, then the corresponding child is absent. We first describe the case that both children $t_1$ and $t_2$ exist, and the one-child case is obtained by deleting all references to the missing child.

Fix a coloring $\tau$ of $\{a_t,b_t\}$. We enumerate all colors $\gamma\in L(c_t)$ that are compatible with $\tau$ on the edges incident with $c_t$ in $G[\{a_t,b_t,c_t\}]$. Such a choice determines colorings $\tau_1$ and $\tau_2$ on the two child links by $\tau_1(a_t)=\tau(a_t)$, $\tau_1(c_t)=\gamma$, $\tau_2(c_t)=\gamma$, and $\tau_2(b_t)=\tau(b_t)$. The sets $Q^*[t_1,\tau_1]$ and $Q^*[t_2,\tau_2]$ have already been computed. We consider all pairs $(g_1,g_2)\in Q^*[t_1,\tau_1]\times Q^*[t_2,\tau_2]$.

We now check whether $(\gamma,g_1,g_2)$ can be combined into a valid coloring for the merged subgraph \(G[P_t]\). The only edges of $G[P_t]$ that are not already checked inside the two child subproblems are the edges with at least one endpoint in $A_{t_1}\cup A_{t_2}\cup\{a_t,b_t,c_t\}$ and not lying entirely in a single child region. Observe that every such edge is of one of the following types: it joins a vertex of $A_{t_1}$ with a vertex of $A_{t_2}$, or it joins a vertex of $A_{t_1}$ with $b_t$, or it joins a vertex of $A_{t_2}$ with $a_t$, or it is an edge of the triangle $\{a_t,b_t,c_t\}$. We therefore accept the triple $(\gamma,g_1,g_2)$ precisely when all these local constraints are satisfied. This test can be carried out in time $O(k^2)$.

Whenever a triple $(\gamma,g_1,g_2)$ is accepted, it yields a coloring $g$ of $A_t$ as follows. If $v\in A_t$ belongs to the left child region, then the color of $v$ is taken from $g_1$; if $v\in A_t$ belongs to the right child region, then the color of $v$ is taken from $g_2$; if $v=c_t\in A_t$, then its color is $\gamma$. We insert the obtained coloring $g$ into a temporary set $Q'[t,\tau]$. After all triples have been processed, we apply Lemma~\ref{lemma4.2} to $Q'[t,\tau]$ and obtain a representative set $Q^*[t,\tau]\subseteq Q'[t,\tau]$ of size at most $\rho(k)$.

For fixed $\tau$, there are at most $n$ choices for $\gamma$ and at most $\rho(k)^2$ pairs $(g_1,g_2)$. Hence $|Q'[t,\tau]|\leq n\cdot \rho(k)^2$, and the time to construct $Q'[t,\tau]$ is $O(n\cdot \rho(k)^2\cdot k^2)$. The subsequent compression step takes time
\(O(|Q'[t,\tau]|\cdot 2^{k^2}k^{k+2})=O(n\cdot \rho(k)^2\cdot 2^{k^2}k^{k+2})\).
Since there are at most $O(n^2)$ choices for $\tau$, the total running time for one internal node is \(O(n^3\cdot \rho(k)^2\cdot 2^{k^2}k^{k+2})\). Since the weak dual \(D\) of the triangulation has $O(n)$ nodes, the whole algorithm runs in time \(O(n^4\cdot \rho(k)^2\cdot 2^{k^2}k^{k+2})\), which is of the form $f(k)\cdot n^{O(1)}$.

\paragraph{Correctness.}
It remains to prove the correctness of the algorithm. For every node $t$ and every coloring $\tau$ of $\{a_t,b_t\}$, we claim that the computed representative set $Q^*[t,\tau]$ satisfies the property $(\star_{t,\tau})$. We prove this by induction on the  subtree of $D$ rooted at the node $t$.

\begin{claim}
The set $Q^*[t,\tau]$ satisfies the property $(\star_{t,\tau})$.
\end{claim}

\begin{proof}
If $t$ is a leaf node, then $Q[t,\tau]$ is computed explicitly from all proper $L$-colorings of the triangle $\triangle_t$, and $Q^*[t,\tau]$ is obtained from $Q[t,\tau]$ by Lemma~\ref{lemma4.2}. Hence $(\star_{t,\tau})$ holds immediately.

Assume now that $t$ is an internal node and that the claim holds for its children. Fix $\tau$ and let $h$ be a proper $L$-coloring of $G[(V(G)\setminus P_t)\cup\{a_t,b_t\}]$ with $h(a_t)=\tau(a_t)$ and $h(b_t)=\tau(b_t)$. Suppose that there is a coloring $g\in Q[t,\tau]$ compatible with $h$. By the definition of $Q[t,\tau]$, there exists a proper $L$-coloring $f$ of $G[P_t]$ such that $f(a_t)=\tau(a_t)$, $f(b_t)=\tau(b_t)$, and $f|_{A_t}=g$.

Let $c_t$ be the third vertex of $\triangle_t$, and let $\gamma=f(c_t)$. Define $\tau_1$ and $\tau_2$ on the child links by $\tau_1(a_t)=\tau(a_t)$, $\tau_1(c_t)=\gamma$, $\tau_2(c_t)=\gamma$, and $\tau_2(b_t)=\tau(b_t)$. Let $f_1$ and $f_2$ be the colorings of $G[P_{t_1}]$ and $G[P_{t_2}]$ induced by $f$, i.e., $f_1 = f|_{P_{t_1}}$ and $f_2 = f|_{P_{t_2}}$. Note that the colorings of child links $\tau_1:\{a_t,c_t\}\to\mathcal{C}$ and $\tau_2:\{b_t,c_t\}\to\mathcal{C}$ are uniquely determined by $\tau$ and $\gamma$.

Consider the left child first. The coloring of the vertices in $(V(G)\setminus P_{t_1})\cup\{a_t,c_t\}$ induced by $h$ together with $f$ is a proper $L$-coloring compatible with $f_1$. By the inductive hypothesis, there exists a representative coloring $g_1^* \in Q^*[t_1,\tau_1]$ that is also compatible with this outside coloring. Since every element of $Q^*[t_1,\tau_1]$ is valid, there exists a proper $L$-coloring $f_1^*$ of $G[P_{t_1}]$ extending $g_1^*$. Replacing $f_1$ by $f_1^*$ yields a new proper $L$-coloring \(f'\) of $G[P_t]$ that is compatible with $h$.

Now consider the right child. Using the new coloring \(f'\) of \(G[P_t]\) obtained after replacing the left child, the coloring of the vertices in $(V(G)\setminus P_{t_2})\cup\{b_t,c_t\}$ induced by $h$ together with $f'$ is a proper $L$-coloring compatible with $f_2$. By the inductive hypothesis, there exists a representative coloring $g_2^*\in Q^*[t_2,\tau_2]$ compatible with this coloring as well. Since every element of $Q^*[t_2,\tau_2]$ is valid, there exists a proper $L$-coloring $f_2^*$ of \(G[P_{t_2}]\) extending $g_2^*$. Hence the combination of $f_1^*$, $f_2^*$, the colors $\tau(a_t)$, $\tau(b_t)$, $\gamma$, and the outside coloring $h$ is still a proper \(L\)-coloring for \(G\).

Therefore the triple $(\gamma,g_1^*,g_2^*)$ passes the local compatibility test performed by the algorithm, and by selecting the coloring of vertices in \(A_t\), it contributes a coloring $g'\in Q'[t,\tau]$. Moreover, $g'$ is compatible with $h$, because it can be extended to a proper $L$-coloring of $G[P_t]$ that is compatible with $h$. Finally, $Q^*[t,\tau]$ is obtained from $Q'[t,\tau]$ by Lemma~\ref{lemma4.2}, and hence there exists a coloring $g^*\in Q^*[t,\tau]$ that is also compatible with $h$. This proves $(\star_{t,\tau})$.
\end{proof}

\paragraph{The root.}
In the end, the root $r$ is handled by the triangulation link $e_0=a_r b_r$ on the outer face. Since $P_r=V(G)$, we have $A_r=\emptyset$. Hence, for every coloring $\tau$ of $\{a_r,b_r\}$, the set $Q[r,\tau]$ is either empty or consists of the unique empty coloring on $A_r$. Moreover, $Q[r,\tau]\neq\emptyset$ if and only if there exists a proper $L$-coloring of $G$ extending $\tau$. Therefore, the input is a yes-instance if and only if there exists a coloring $\tau$ of $\{a_r,b_r\}$ such that $Q^*[r,\tau]\neq\emptyset$. This completes the proof.
\end{proof}

\noindent\fbox{
\begin{minipage}{0.97\textwidth}
\textsc{Precoloring Extension}\\
\textbf{Input:}  A graph $G = (V,E)$, a set $\mathcal{C}$, a subset $W \subseteq V$, and a function $f' \colon W \rightarrow \mathcal{C}$.\\
\textbf{Question:}  Is there a function $f \colon V \rightarrow \mathcal{C}$ such that for every $v \in W$, $f(v) = f'(v)$, and for every $uv \in E$, $f(u) \neq f(v)$?
\end{minipage}}

\begin{corollary}
\textsc{Precoloring Extension} is FPT on outer \(k\)-planar graphs parameterized by \(k\).
\end{corollary}

\begin{proof}
We use a standard reduction from \textsc{Precoloring Extension} to \textsc{List Coloring}. Given an instance \((G, \mathcal{C}, W, f')\) of the \textsc{Precoloring Extension} problem on an outer \(k\)-planar graph, we construct an instance of \textsc{List Coloring} on a graph \(H\) as follows. We first take a copy of \(G\) and assign to each vertex \(v\) in this copy a list \(L(v)\) of admissible colors. For every vertex \(v \notin W\), we define \(L(v) = \mathcal{C} \setminus \{ f'(u) : u \in N_G(v) \cap W \}\). For every vertex \(v \in W\), we set \(L(v) = \{ f'(v) \}\).

This construction can be carried out in polynomial time and does not modify the drawing of the graph. Hence the resulting graph \(H\) is still outer \(k\)-planar. Suppose that \(G\) admits a valid precoloring extension. Then assigning to each vertex in \(H\) the corresponding color from its list yields a proper list coloring of \(H\). Conversely, since each vertex \(v \in W\) in \(H\) has a singleton list \(L(v)=\{f'(v)\}\), any proper list coloring of \(H\) necessarily assigns color \(f'(v)\) to every vertex \(v \in W\). Therefore, such a coloring directly yields a valid precoloring extension of \(G\). Since \textsc{List Coloring} on outer \(k\)-planar graphs is FPT with respect to the parameter \(k\), it follows that \textsc{Precoloring Extension} is also FPT on outer \(k\)-planar graphs parameterized by \(k\).
\end{proof}

\subsection{Capacitated (Red-Blue) Dominating Set}\label{CDS}
In this section, we consider the well-studied capacitated variant of the \textsc{Dominating Set} problem, as well as its red-blue variant. Prior to the XNLP-hardness and XALP-hardness results for pathwidth and treewidth, \textsc{Capacitated Dominating Set} was already shown to be W[1]-hard when parameterized by treewidth in~\cite{CDS-tw}. It was also proved to be W[1]-hard on planar graphs when parameterized by the solution size~\cite{CDS-planar}. Furthermore, a lower bound argument in~\cite{CDS-red-blue-lb} implies that \textsc{Capacitated Red-Blue Dominating Set} is W[1]-hard when parameterized by feedback vertex set number. We first present an FPT algorithm for \textsc{Capacitated Dominating Set} on outer \(k\)-planar graphs parameterized by \(k\), and then show how a simple modification also yields an FPT algorithm for \textsc{Capacitated Red-Blue Dominating Set}.

\vspace{3mm}
\noindent\fbox{
\begin{minipage}{0.97\textwidth}
\textsc{Capacitated Dominating Set}\\
\textbf{Input:} Graph \(G = (V, E)\), with each vertex \(v \in V\) having a positive integer capacity \(c(v)\), and an integer \(m\). \\
\textbf{Question:} Is there a set \(S \subseteq V\) of at most \(m\) vertices, and an assignment \(f : V \setminus S \to S\) mapping each vertex not in \(S\) to a neighbor in \(S\), such that each vertex \(v \in S\) has at most \(c(v)\) neighbors assigned to it?
\end{minipage}}
\vspace{3mm}


Similar to the algorithm for \textsc{List Coloring}, we design a bottom-up dynamic programming algorithm running in FPT time on the triangulated graph shown in Figure~\ref{fig:subpolygon}. We use the same notation as in the proofs of Theorem~\ref{mim-proof} and Theorem~\ref{list-coloring-thm} above. Moreover, we define the set of \emph{boundary vertices} of the subpolygon $P_t$ as $B_t = \{a_t, b_t\} \cup A_t$, where \(A_t = \{ v \in P_t \setminus \{a_t, b_t\} \mid \exists\, vw \in E \text{ with } w \notin P_t \}\) is called the \emph{interface} of $G[P_t]$. By Lemma~\ref{crossing-lemma}, the triangulation link $\lambda_t$ is crossed by at most $k$ edges of \(G\), which implies $|A_t| \le k$ and hence $|B_t| \le k+2$. See Figure~\ref{fig:subpolygon} for an example illustrating some of these definitions.

\begin{figure}[ht!]
    \centering
    \includegraphics[width=0.9\linewidth]{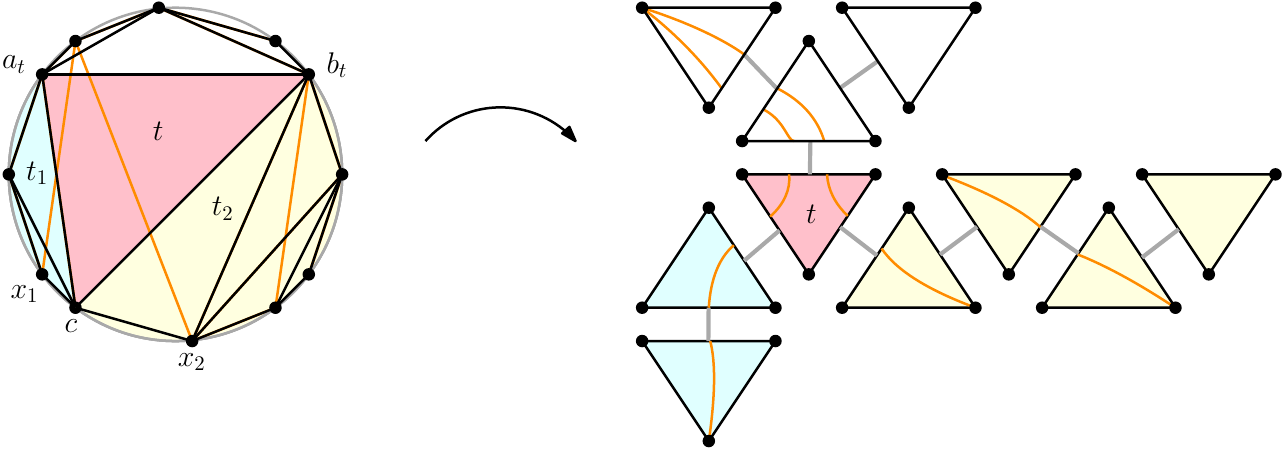}
    \caption{Each triangulation link \(\lambda_t = a_t b_t\) separates the triangulated graph into two subpolygons. Let \(P_{t_1}\) and \(P_{t_2}\) denote the two subproblems induced by the links \(a_t c\) and \(b_t c\), respectively. The boundary vertices of subpolygon \(P_{t_1}\) are \(B_{t_1} = \{a_t, x_1, c\}\), and \(A_{t_1} = \{x_1\}\), while the boundary vertices of \(P_{t_2}\) are \(B_{t_2} = \{b_t, x_2, c\}\), and \(A_{t_2} = \{x_2\}\). The boundary vertices of subpolygon \(P_t\) (obtained by merging the three colored regions) are \(B_t = \{a_t, b_t, x_1, x_2\}\), and \(A_t = \{x_1, x_2\}\).}
    \label{fig:subpolygon}
\end{figure}

For every boundary vertex \(x \in B_t\), we assign one of three possible states describing how \(x\) interacts with the partial solution inside the subproblem induced by \(P_t\). First, the state \(U\) indicates that \(x \notin S\) and that the vertex \(x\) has already been dominated by a vertex within the subpolygon, that is, \(f(x) \in S \cap P_t\). Second, the state \(O\) represents the case where \(x \notin S\) and its assignment is postponed to the exterior of the current subpolygon, meaning that eventually \(f(x) \in S \setminus P_t\). Finally, the state \(S(r_c)\) indicates that \(x \in S\), and that we record an integer \(r_c\) representing the remaining capacity of \(x\) that may be used by neighbors outside \(P_t\).

The range of the remaining capacity \(r_c\) depends on the role of the boundary vertex. If \(x \in A_t\), that is, \(x\) is an interface vertex, then the exterior can assign at most \(d_t^{\mathrm{out}}(x):=|N(x)\setminus P_t|\) vertices to \(x\) via edges crossing the link \(\lambda_t\). By Lemma~\ref{crossing-lemma}, we have \(d_t^{\mathrm{out}}(x)\le k\).  Hence it suffices to record \(r_c \in \{0,1,\ldots,d_t^{\mathrm{out}}(x), k^+\}\). Here \(k^+\) is a \emph{saturated} value and means that the true remaining capacity of \(x\) is greater than \(d_t^{\mathrm{out}}(x)\). Thus \(k^+\) does not denote one fixed integer; it only says that the remaining capacity is large enough for all neighbors outside \(P_t\) that may still be assigned to \(x\). On the other hand, if \(x \in \{a_t,b_t\}\), that is, \(x\) is one of the two endpoints of the link \(\lambda_t\), then vertices outside the subpolygon \(P_t\) may connect to \(x\) without the crossing restriction. Therefore the remaining capacity must be recorded exactly, and we store \(r_c \in \{0,1,\ldots,\min(c(x),n-1)\}\).

As discussed at the beginning of this section, unlike classical treewidth-based algorithms where each vertex in a bag may store at most $O(n)$ possible values for its remaining capacity, leading to a DP table of size $n^{O(k)}$ for each subproblem, outer $k$-planar graphs admit a stronger geometric restriction, which ensures that each vertex $x \in A_t$ has at most $k$ neighbors outside $P_t$, thus its number of possible values of $r_c$ can be bounded by \(O(k)\). This observation is the key ingredient that enables the algorithm to run in FPT time. To see why this is correct, consider a vertex \(v\in A_t\) and its remaining capacity recorded in a dynamic programming state. Suppose that the remaining capacities of all vertices in \(A_t\cup\{a_t,b_t\}\setminus\{v\}\) are fixed. Then all states in which the remaining capacity of \(v\) takes a value in \(\{d_t^{\mathrm{out}}(v)+1,d_t^{\mathrm{out}}(v)+2,\ldots,n-1\}\) are equivalent and can be represented by a single value \(k^+\). The reason is that, from the perspective of vertices outside \(P_t\), these states are indistinguishable. Our goal is merely to determine whether there exists a dominating set satisfying all capacity constraints, not to keep track of the exact amount of unused capacity left at each selected vertex. Since every interface vertex \(v\in A_t\) has at most \(d_t^{\mathrm{out}}(v) \le k\) neighbors outside \(P_t\), any remaining capacity greater than \(d_t^{\mathrm{out}}(v)\) can never be distinguished by these outside neighbors of \(v\) in future computations.

For example, suppose there exists a solution \(S\), in which for a DP subproblem the vertex \(v \in A_t\) is selected into \(S\) and uses \(k'\) units of its capacity inside \(P_t\). Whether the DP state records the remaining capacity of \(v\) as \(k+5\) or \(k+1\) is irrelevant. If a remaining capacity of \(k+5\) is sufficient to satisfy the demands of all neighbors of \(v\) outside \(P_t\), then a remaining capacity of \(k+1\) is also sufficient, because \(v\) has at most \(d_t^{\mathrm{out}}(v) \le k\) such neighbors. Consequently, all capacity values larger than \(d_t^{\mathrm{out}}(v)\) are equivalent and can be represented by a single value \(k^+\). This significantly reduces the number of states that must be maintained for boundary vertices.

Next, we present the full details of the algorithm for \textsc{Capacitated Dominating Set}. This includes the definition of the DP table, the handling of triangles corresponding to leaf, internal, and root nodes in the weak dual \(D\), the procedure for merging subpolygons and partial solutions along \(D\) in a bottom-up manner, the analysis of the running time, and the correctness proof.

\paragraph{Table definition.}
For each subpolygon \(P_t\), we define a signature \(\sigma\) for node \(t\) that assigns one of three possible states to every boundary vertex \(x \in B_t\), and the set \(\Sigma_t\) contains all signatures of node \(t\). We define the DP table entry as \(Q[t,\sigma] = \min |S \cap (P_t \setminus B_t)|\). For a signature \(\sigma\), the value \(Q[t,\sigma]\) represents the minimum number of selected vertices in \(P_{t} \setminus B_{t}\), such that every vertex of \(P_{t} \setminus B_{t}\) is assigned a dominator within \(P_{t}\), the boundary vertices in \(B_t\) behave according to the states specified by \(\sigma\), and all capacity constraints are satisfied. Counting only vertices from \(P_t \setminus B_t\) is crucial: if a boundary vertex is selected into \(S\), it is not counted at this stage, since it will later be `forgotten', i.e., become an internal non-boundary vertex, in a subpolygon corresponding to an ancestor node of \(t\), where it will be counted exactly once. This guarantees that every selected vertex contributes exactly one unit to the dominating set \(S\).

The final answer can be obtained from the table entries of the root node \(r\). In the weak dual \(D\) of the triangulation \(H\), the root node \(r\) corresponds to a triangle \(\triangle_r\) that contains an outer-face link \(e_0=a_r b_r\) which is not crossed by any other edge; see Figure \ref{fig:subpolygon} for an example. The interior of \(e_0\) (including its two endpoints) contains all vertices of the graph. Therefore, the boundary vertices \(B_r\) of the root \(r\) consist only of the two endpoints of \(e_0\). Consequently, the table entry \(Q[r, \sigma]\) represents the number of vertices selected into \(S\) among all vertices of the graph except the two endpoints of \(e_0\). By enumerating all possible signatures of \(B_r\), that is, all possible state combinations of \(a_r\) and \(b_r\), we obtain candidate solutions. The minimum value among these candidates gives the size of the minimum capacitated dominating set of \(G\).

At the beginning, all table entries are initialized to \(+\infty\). This indicates that for each node \(t\), only those signatures \(\sigma\) that can be obtained from valid combinations of subproblems will eventually receive a finite value in \(Q[t,\sigma]\). We now describe how to compute DP tables for leaf, root, and internal nodes of the weak dual. We start with the case where \(t\) is a leaf node.

\paragraph{Leaf nodes.}
In the weak dual \(D\) of the triangulation, every leaf node corresponds to a triangle that is an \emph{ear} of the polygon enclosed by the outer cycle. Let the parent link of this ear be \(\lambda_t=ab\), and let \(c\) be the third vertex of the triangle. The corresponding subpolygon \(P_t\) therefore contains exactly the three vertices \(\{a,b,c\}\). In this case the boundary vertices \(B_t = \{a,b\} \cup A_t\), where \(A_t \subseteq \{c\}\). Hence the boundary contains at most three vertices. To compute the table entry for a signature \(\sigma\) (which assigns one of the states \(U\), \(O\), or \(S(r_c)\) to each of \(a\), \(b\), and possibly \(c\)), we explicitly enumerate all subsets \(S \subseteq \{a,b,c\}\) inside the subpolygon. Since \(P_t\) contains only three vertices, this yields at most a constant number of possibilities.

For each such choice of \(S\), we determine whether the vertices not in \(S\) can be assigned to dominating vertices inside the triangle. Consider a vertex \(v \notin S\). If \(v \in B_t\) and the signature specifies \(\sigma(v)=O\), then the assignment of \(v\) is postponed to the exterior of \(P_t\), and therefore no assignment needs to be made within the current subproblem. Otherwise, \(v\) must be assigned to a neighbor \(u \in S \cap N(v)\) inside the triangle, which can be chosen among the vertices \(\{a,b,c\}\). For every vertex \(u \in S\), the number of vertices assigned to \(u\) within the triangle must not exceed its capacity \(c(u)\). After processing the assignments, we compute for each boundary vertex the remaining capacity that is still available for assignments coming from outside the subpolygon. The obtained value of the remaining capacity \(r_c\) must be consistent with the state prescribed by the signature \(\sigma\). If all these conditions are satisfied, the candidate partial solution is consistent with the signature. In this case we update the table entry \(Q[t,\sigma]\) using the value \(|S \cap (P_t \setminus B_t)|\).

\paragraph{Internal nodes.}
Let $t$ be an internal node of the weak dual whose corresponding triangle is $\triangle_t(a,b,c)$, where the parent link is $\lambda_t=ab$. Suppose the two children of $t$ are $t_1$ and $t_2$, whose parent links correspond to $\lambda_{t_1}=ac$ and $\lambda_{t_2}=bc$, respectively. We first explain how to compute $Q[t,\sigma]$ from the tables of the two children. The case where $t$ has only one child will be discussed afterwards. As before, the boundary vertices of the parent problem are \(B_t = \{a,b\} \cup A_t\), where \(A_t\) contains endpoints of edges crossing the link \(ab\) that lie within \(P_t\).

For each pair of child signatures \((\sigma_1,\sigma_2) \in \Sigma_{t_1} \times \Sigma_{t_2}\), we first check their compatibility and then construct candidate parent signatures. The first compatibility constraint concerns the shared vertex \(c\). If \(c \in S\), then both children must agree that \(c \in S\), and their internal loads assigned to \(c\) must be combined.  Let the remaining capacities reported by the two children be \(r_{c_1}\) and \(r_{c_2}\). Since each child already accounts for the load inside its own region, the combined remaining capacity for \(c\) after merging the two children is \(r_{c_0} = r_{c_1} + r_{c_2} - c(c)\). If \(r_{c_0} < 0\), then the total load assigned to \(c\) exceeds its capacity, and the pair \((\sigma_1,\sigma_2)\) is discarded. Note that the vertex \(c\) was an endpoint of \(\lambda_{t_i}\) for \(i\in\{1,2\}\), but is no longer an endpoint of \(\lambda_t\). If \(c \in S \cap B_t\), then at most \(d_t^{\mathrm{out}}(c)\) edges can cross the link \(ab\) and be incident to \(c\). Hence, if \(r_{c_0}\le d_t^{\mathrm{out}}(c)\), we store the exact value \(r_{c_0}\); otherwise we store the saturated value \(k^+\).


If \(c \notin S\), the pair \(\sigma_1(c) = U\) and \(\sigma_2(c) = U\) is not allowed, because the assignment function \(f(c)\) can only choose one dominator (it cannot assign it once on each side), therefore the allowed state combinations for shared vertex \(c\) are \((U,O), (O,U), (O,O)\). Note that the child state pair \((O,O)\) of \(c\) is compatible if and only if \(c\in B_t\), in which case the state of \(c\) in the parent problem is forced to be \(O\). If \(c\notin B_t\), then the pair \((O,O)\) is infeasible, because \(c\) will never be dominated. For the child state pair \((U, O)\) or \((O, U)\), the state of \(c\) in parent problem is marked as \(U\), meaning that the shared vertex \(c\) has already been dominated by a vertex within \(P_t\).

\begin{figure}[ht!]
    \centering
    \includegraphics[width=0.7\linewidth]{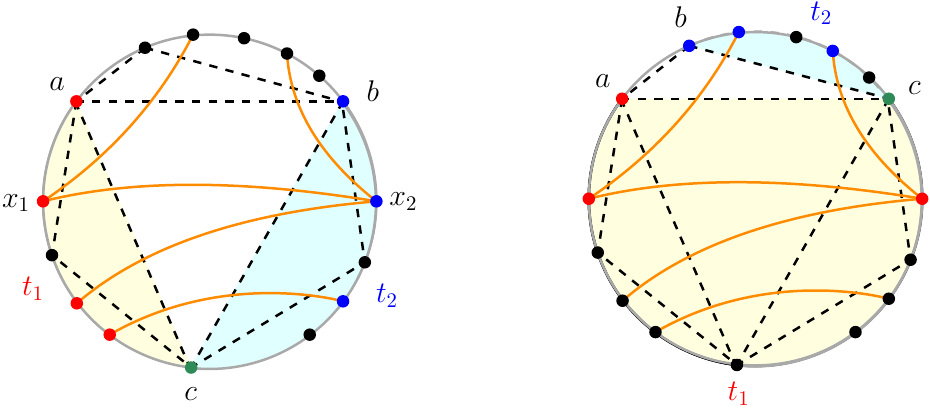}
    \caption{An example of merging internal nodes \(t_1\) and \(t_2\) to compute the DP table for node \(t\). Vertices in \(B_{t_1}\) on the outer cycle are marked in red, vertices in \(B_{t_2}\) are marked in blue, and the vertex \(c\) shared by the two subproblems is marked in green. The figure on the left shows the merging of \(t_1\) and \(t_2\), while the right one shows node \(t\) after the merge, used for the next step of further merging. Some triangulation links and edges are not drawn for better visualization.}
    \label{fig:DP-combine}
\end{figure}

The second compatibility condition concerns boundary vertices that are shared by a child and the parent.  For any vertex \(v \in B_{t_i} \cap B_t\) for \(i\in\{1,2\}\), the parent signature must inherit the states \(S(r_c)\) and \(U\) of \(v\) from the corresponding child, i.e., a vertex marked as \(U\) or \(S(r_c)\) in a child cannot be changed to another state in the parent (only selected status is inherited, but the remaining capacity is updated after the merge). However, if a vertex \(v \in B_{t_i} \cap B_t\) is marked as \(O\) in the child problems, there are two choices at this time: if it continues to remain \(O\), it means that \(v\) is not dominated within the current parent \(t\), but is postponed to the future; if it becomes \(U\), then it means that a vertex with state \(S(r_c)\) must be chosen from \(P_t \setminus P_{t_i}\) to dominate \(v\). Since \(B_{t_i}\) contains at most \(k + 2\) vertices, in this step we can enumerate all choices for the vertices in \(B_{t_i} \cap B_t\) whose state is \(O\). As shown in Figure \ref{fig:DP-combine}, vertex \(x_1\) belongs to \(B_{t_1} \cap B_t\). If its state is \(O\), then at this time we can continue to keep it as \(O\), or choose a vertex in \(P_t \setminus P_{t_1}\) whose state is \(S(r_c)\) to dominate it, for example choosing \(x_2\) if \(x_2\) is selected into \(S\).

Vertices that appear in the child boundaries but not in the parent boundary become internal at this step.  Formally we define \(W = (B_{t_1} \cup B_{t_2}) \setminus B_t\). If a vertex \(v \in W\) is selected into the dominating set in the combined solution, it becomes an internal vertex and its selection must be counted in the objective value at this step. Furthermore, if \(v \in W\) and the corresponding child signature of \(B_{t_i}\) marks \(v\) as \(O\), then \(v\) cannot remain unresolved, since it will disappear from the boundary. Therefore it must be assigned to a dominator from \(P_t\setminus P_{t_i}\) during this merge step. Let \(F_1 :=\bigcup_{i\in\{1,2\}} \{v \in W \cap B_{t_i} : \sigma_i(v) = O\}\) be the set of vertices that must be assigned a dominator in this step. In addition, some vertices in \(F_2 :=\bigcup_{i\in\{1,2\}} \{v \in B_t \cap B_{t_i} : \sigma_i(v)=O\}\) may optionally be resolved at this step.  For each subset \(R \subseteq F_2\), we interpret the vertices in \(R\) as being assigned during this merge step, while vertices in \(F_2 \setminus R\) remain in state \(O\) in the parent signature.  The set of demand vertices to be assigned at this step is therefore \(F = F_1 \cup R\).

To determine whether the demands in \(F\) can be satisfied, we enumerate, for each vertex \(v \in F\), possible dominators among vertices whose state is \(S(r_c)\). In particular, if a vertex \(v \in B_{t_i}\) has state \(O\) in the child signature, then its dominator must lie in \(P_t \setminus P_{t_i}\). This follows from the semantics of state \(O\), which indicates that the vertex is not dominated within the child problem \(t_i\) and therefore can only be dominated by a vertex outside the corresponding subpolygon \(P_{t_i}\). The set of vertices whose state is \(S(r_c)\) consists of two parts. The first part contains vertices in the set \(W\). These vertices become internal vertices at the current merge step, and therefore we no longer need to track their remaining capacity. It suffices to ensure that the number of vertices assigned to such a vertex does not exceed its available capacity. The second part consists of vertices that remain in the parent boundary \(B_t\). For each enumerated assignment of demands to supplies, we must record the remaining capacity of these boundary vertices, since this information is required for subsequent merge operations. Note that for the remaining capacity of any vertex \(v\in A_t\), if the child state stores \(k^+\), then the parent state also stores \(k^+\). Indeed, \(k^+\) means that \(v\) has enough capacity for all neighbors outside the child subpolygon \(P_{t_i}\) that may still be assigned to it; after assigning some of these vertices at the current merge, \(v\) still has enough capacity for all neighbors outside the parent subpolygon \(P_t\). For each enumerated demand--supply assignment, we verify that no vertex with state \(S(r_c)\) exceeds its available capacity. If this condition is satisfied, we record the resulting states of all vertices in \(B_t\). In this way, from the pair of child signatures \(\sigma_1\) and \(\sigma_2\) we obtain a candidate parent signature \(\sigma\), and we update the table entry \(Q[t,\sigma]\) using the value \(\min\left\{Q[t, \sigma], Q[t_1,\sigma_1] + Q[t_2,\sigma_2] + |\{v \in W : v \in S\}|\right\}\).

\paragraph{The one-child case.} We describe how to handle the special case where an internal node has only one child. As illustrated in Figure~\ref{fig:one-child}, suppose that node \(t\) has only one child \(t_2\), corresponding to the right subpolygon. In this case, we conceptually introduce a dummy left subproblem \(t_1\) that consists only of the two vertices \(a_t\) and \(c\). This dummy subproblem can be treated as a leaf node. Hence, the table \(Q[t_1,\sigma]\) can be computed using the same procedure as in the leaf initialization step. After computing \(Q[t_1,\sigma]\), we merge \(t_1\) and \(t_2\) using exactly the same merge procedure as for the general case with two children.

\begin{figure}[ht!]
    \centering
    \includegraphics[width=0.9\linewidth]{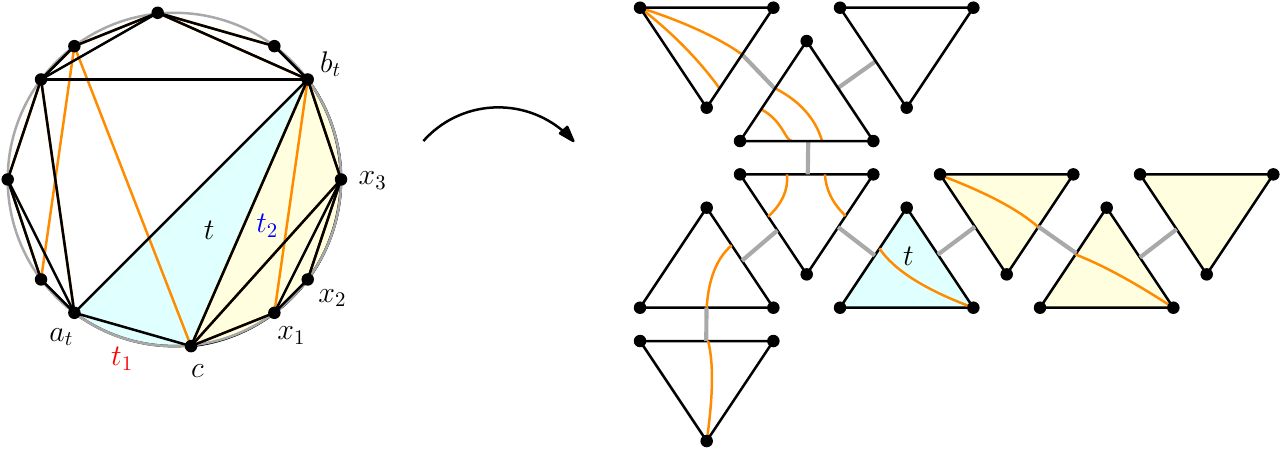}
    \caption{An illustration of how to handle an internal node with only one child problem.}
    \label{fig:one-child}
\end{figure}

\paragraph{The root.}
The root \(r\) corresponds to a triangle incident to an outer-face link \(e_0 = a_0 b_0\), and \(e_0\) is not crossed by any other edge, so the interface \(A_r\) is empty.  Consequently, the boundary vertices \(B_r\) of the root consists only of the two vertices \(\{a_0,b_0\}\). At the root node, the subpolygon \(P_r\) coincides with the entire vertex set \(V\) of the graph. Since the outside region is empty, any boundary vertex marked with state \(O\) would require a dominator outside the graph, which is impossible. Therefore admissible root signatures must satisfy \(\sigma(a_0) \neq O\) and \(\sigma(b_0) \neq O\).

The final solution is obtained by enumerating all signatures on the root boundary \(B_r\) satisfying \(\sigma(a_0), \sigma(b_0) \in \{U, S(\cdot)\}\). For each such signature \(\sigma\), we take the value \(Q[r,\sigma]\) and add one for each endpoint \(v \in \{a_0,b_0\}\) with \(\sigma(v)=S(\cdot)\), and return the minimum over all such signatures. By definition, the table value \(Q[t,\sigma]\) only counts selected vertices in \(P_t \setminus B_t\), and boundary vertices are counted later when they become internal. Since the root has no parent, vertices in \(B_r\) will never be forgotten in a higher-level merge, thus any selected endpoint of \(e_0\) must be added explicitly at the end. This also ensures that the two endpoints of \(e_0\) are fully dominated. If \(\sigma(a_0)=S(r_c)\), then \(a_0\) itself is selected into the dominating set and requires no assignment. If \(\sigma(a_0)=U\), then \(a_0\notin S\) but already has a dominator inside \(P_r\). The state \(O\) is forbidden because there are no outside vertices that could dominate \(a_0\).  The same interpretation applies to \(b_0\).

\begin{lemma}[correctness]\label{lem:cds-dp-correctness}
For every node \(t\) of the rooted weak dual and every signature \(\sigma\in\Sigma_t\), after the table of \(t\) has been computed, the value \(Q[t,\sigma]\) is equal to the minimum value of \(|S\cap (P_t\setminus B_t)|\) over all partial solutions inside \(P_t\) that realize the signature \(\sigma\). Consequently, the value returned at the root is the size of a minimum capacitated dominating set of \(G\).
\end{lemma}

\begin{proof}
We prove the invariant by induction over the rooted weak dual, from the leaves to the root. A partial solution for a state \((t,\sigma)\) consists of a set \(S_t\subseteq P_t\) and a partial assignment \(f_t\) such that every vertex of \(P_t\setminus B_t\) is either in \(S_t\) or is assigned to a neighbor in \(S_t\cap P_t\), every boundary vertex behaves according to its state in \(\sigma\), and no selected vertex receives more assigned vertices than allowed by its capacity. More precisely, if \(\sigma(v)=U\), then \(v\notin S_t\) and \(f_t(v)\in S_t\cap P_t\); if \(\sigma(v)=O\), then \(v\notin S_t\) and the value \(f_t(v)\) is not yet defined inside \(P_t\); and if \(\sigma(v)=S(r_c)\), then \(v\in S_t\) and the state records the amount of capacity of \(v\) still available for vertices outside \(P_t\). For an interface vertex \(u\in A_t\), if \(u\in S_t\) the saturated value \(k^+\) means that the true remaining capacity of \(u\) is larger than \(d_t^{\mathrm{out}}(u)=|N(u)\setminus P_t|\). This is sufficient because at most \(d_t^{\mathrm{out}}(u)\) neighbors of \(u\) outside the current subpolygon \(P_t\) can still be assigned to \(u\) in later merge steps.

For a leaf node, the subpolygon contains only the three vertices of an ear triangle. The algorithm enumerates all possible selected subsets inside this constant-size subpolygon and all assignments of unselected vertices that are not postponed by an \(O\)-state. It then checks adjacency, the states prescribed by the signature, and the capacity constraints. Hence a finite value is inserted exactly for those signatures that are realized by a valid partial solution on the leaf, and the stored table entry value is exactly the minimum possible value of \(|S\cap(P_t\setminus B_t)|\).

\paragraph{Soundness.}
Now consider an internal node \(t\) with children \(t_1\) and \(t_2\), where the corresponding triangle is \(\triangle(a,b,c)\), the parent link is \(ab\), and the child links are \(ac\) and \(bc\). Assume by induction that the invariant holds for both children. Suppose that the algorithm combines two finite child entries \(Q[t_1,\sigma_1]\) and \(Q[t_2,\sigma_2]\), chooses a set of postponed vertices to resolve at the current merge, and enumerates a feasible assignment of these demand vertices to selected supply vertices. By the induction hypothesis, there exist valid partial solutions realizing \(\sigma_1\) and \(\sigma_2\). The compatibility test for the shared vertex \(c\) ensures that these two partial solutions do not define two different assignments for \(c\). If \(c\) is selected, the loads reported by the two children are added and the pair is rejected whenever the capacity of \(c\) would be exceeded. If \(c\) is not selected, the pair \((U,U)\) is rejected, the pairs \((U,O)\) and \((O,U)\) make \(c\) resolved in the parent, and the pair \((O,O)\) is allowed only when \(c\in B_t\), in which case the parent state of \(c\) remains \(O\).

Every other vertex that disappears from the boundary at this merge is handled as follows. If it is already in state \(U\), its assignment has been fixed inside its child subpolygon. If it is selected, it is counted at this merge, because from now on it is internal. If it is in state \(O\), then it is included among the demand vertices that must be assigned now. Moreover, if such a demand vertex \(v\) comes from the child \(t_i\), the algorithm only allows dominators in \(P_t\setminus P_{t_i}\), which is exactly the meaning of the \(O\)-state in the child. For boundary vertices that remain in \(B_t\) and are in state \(O\) in a child, the algorithm branches over the choice of either keeping them unresolved in the parent or resolving them at the current merge. The enumerated demand--supply assignment fixes exactly the latter vertices. The capacity check and the update of the remaining capacities ensure that no selected vertex exceeds its available capacity. Therefore the two child partial solutions together with the newly chosen assignments form a valid partial solution for the parent signature \(\sigma\). The value used to update \(Q[t,\sigma]\) counts the selected vertices already counted in the two children and additionally counts exactly the selected vertices that disappear from the boundary at the current merge. Hence every finite value produced by the algorithm is sound.

\paragraph{Completeness.}
Let \((S_t,f_t)\) be any valid partial solution realizing a parent signature \(\sigma\). Restrict \((S_t,f_t)\) to \(P_{t_1}\) and \(P_{t_2}\). For each child boundary vertex \(v\), define its child state as follows. If \(v\in S_t\), the child state is \(S(r_c)\), where \(r_c\) is the capacity of \(v\) remaining after the assignments whose assigned vertices lie in that child subpolygon. If \(v\notin S_t\) and \(f_t(v)\) lies inside the child subpolygon, the child state is \(U\). If \(v\notin S_t\) and \(f_t(v)\) does not lie inside the child subpolygon, the child state is \(O\). These definitions give two child signatures \(\sigma_1\) and \(\sigma_2\), and the restrictions of \((S_t,f_t)\) are valid partial solutions for them. By the induction hypothesis, the child tables contain values no larger than the numbers of selected internal vertices used by these restrictions.

The parent partial solution also determines exactly which child \(O\)-vertices are resolved at the current merge: these are precisely the vertices whose assignment under \(f_t\) lies in \(P_t\) but outside the child subpolygon in which they were marked \(O\). The algorithm branches over all such choices, so it considers this one. The assignment \(f_t\) itself gives one of the enumerated demand--supply assignments. All compatibility conditions are satisfied, including the special rule for the shared vertex \(c\), because \(f_t\) is a function and because every vertex not kept on the parent boundary is already resolved in the parent partial solution. The capacity updates made by the algorithm are exactly those induced by \(f_t\). Therefore, the algorithm creates the parent signature \(\sigma\) with a partial solution of size at most \(|S_t\cap(P_t\setminus B_t)|\).

\paragraph{The root.}
The induction establishes the claimed invariant for every node. At the root \(r\), the subpolygon \(P_r\) is the whole vertex set \(V\), and the only boundary vertices are the endpoints \(a_0,b_0\) of the chosen outer-face link. Since there is no exterior to the root subpolygon, a root signature containing state \(O\) is inadmissible. If \(a_0\) or \(b_0\) is in state \(U\), then it already has a dominator inside \(G\); if it is in state \(S(\cdot)\), then it belongs to the dominating set and must be counted explicitly, because vertices \(a_0\) and \(b_0\) can never be forgotten by an ancestor. Hence taking the minimum of \(Q[r,\sigma]\) plus one for each outer-face link endpoint in state \(S(\cdot)\), over all admissible root signatures \(\sigma\), gives exactly the minimum size of a capacitated dominating set of \(G\).
\end{proof}

\begin{theorem}
\textsc{Capacitated Dominating Set} can be solved in time \(n^5k^{O(k)}\) on outer \(k\)-planar graphs. Therefore, the problem is FPT when parameterized by outer \(k\)-planarity.
\end{theorem}

\begin{proof}
Recall that the algorithm operates on the weak dual of the triangulation obtained from the given outer $k$-planar drawing.  Each node of this tree corresponds to a triangle and each subtree represents a subproblem defined on a subpolygon $P_t$ with boundary vertices $B_t$. For each interface vertex $x \in A_t$, the number of possible states is bounded by \(2 + (k + 2) \le k + 4\), since the vertex may be labeled $U$, $O$, or $S(r_c)$ with $r_c \in \{0,1,\ldots,d_t^{\mathrm{out}}(x), k^+\}$ (values larger than $d_t^{\mathrm{out}}(x) \le k$ are represented by \(k^+\)). For the endpoints $a_t$ and $b_t$ of a triangulation link \(\lambda_t\), the remaining capacity can range up to the maximum capacity of the vertex. Hence each endpoint has at most $n$ possible values of $r_c$, together with the two states $U$ and $O$. Combining these bounds, the number of signatures stored at each node $t$ satisfies \(|\Sigma_t| \le (n+2)^2 \cdot (k+4)^{k}\). 

A leaf node of the weak dual corresponds to an ear triangle with vertices $\{a,b,c\}$.  In this case the subpolygon $P_t$ contains only these three vertices, and the boundary is \(B_t = \{a,b\} \cup A_t\) where \(A_t \subseteq \{c\}\). To compute the table values of a leaf node, we enumerate all subsets $S \subseteq \{a,b,c\}$ indicating which vertices are selected into the dominating set.  For each such choice we assign every unselected vertex $v \notin S$ to one of its neighbors in $S$, provided that capacity constraints are satisfied. Since the number of vertices in the triangle is constant, all these checks take constant time. Therefore each signature $\sigma$ can be evaluated in constant time, and the total time required to compute the table for a leaf node is \(O(|\Sigma_t|) = O(n^2 \cdot k^{O(1)})\).

Consider an internal node $t$ with children $t_1$ and $t_2$. The table $Q[t,\sigma]$ is computed by combining signatures $\sigma_1 \in \Sigma_{t_1}$ and $\sigma_2 \in \Sigma_{t_2}$. The number of candidate pairs is bounded by \(|\Sigma_{t_1}| \cdot |\Sigma_{t_2}| = O(n^4 (k+4)^{2k})\). Let \(F_2\) be the set of vertices in \(B_t\) that are labeled \(O\) in the signature of their corresponding child problem. These vertices remain on the boundary after the merge, and the algorithm decides whether they should already be dominated at the current step or deferred to higher levels.  For this purpose we branch over a subset \(R \subseteq F_2\), which represents the vertices whose assignments are resolved during the current merge, and the states of vertices in \(R\) become \(U\). Since $|F_2| \le |B_t| \le k+2$, this branching introduces a factor \(2^{O(k)}\). Let \(F = F_1 \cup R\) be the set of vertices that must be assigned a dominator at this merge step.  Every vertex in $F \cap B_{t_i}$ must be dominated by a vertex in $S \cap (P_t \setminus P_{t_i})$, where $t_i$ denotes the child. Both the demand vertices and the potential supply vertices lie on the boundary, whose size is at most $k+2$.  Therefore the number of possible dominator assignments is bounded by \(k^{O(k)}\). For each assignment we check the compatibility constraints and verify that the capacity of every selected dominator is not exceeded. These checks can be performed in $k^{O(1)}$ time. Combining all factors, the running time for processing one internal node is bounded by \(O(n^4 \cdot (k+4)^{2k} \cdot 2^{O(k)} \cdot k^{O(k)})\).

The weak dual of the triangulation contains $O(n)$ nodes. Summing the cost of all nodes, the total running time of the algorithm is \(O(n^5 \cdot (k+4)^{2k} \cdot k^{O(k)})\). The correctness follows from Lemma~\ref{lem:cds-dp-correctness}, hence \textsc{Capacitated Dominating Set} is FPT parameterized by outer \(k\)-planarity.
\end{proof}

We extend the algorithm for \textsc{Capacitated Dominating Set} (CDS) on outer \(k\)-planar graphs to the \textsc{Capacitated Red-Blue Dominating Set} (CRBDS) problem. In CRBDS, the vertex set is partitioned into red vertices \(R\) (demands) and blue vertices \(B\) (facilities), and only blue vertices can be selected into the solution \(S\subseteq B\). Each red vertex must be assigned to a selected blue neighbor, subject to capacity constraints. The problem is defined as follows.

\vspace{3mm}
\noindent\fbox{
\begin{minipage}{0.97\textwidth}
\textsc{Capacitated Red-Blue Dominating Set}\\
\textbf{Input:} Bipartite graph \(G = (R \cup B, E)\), with each blue vertex \(v \in B\) having a positive integer capacity \(c(v)\), and an integer \(m\). \\
\textbf{Question:}  Is there a set \(S \subseteq B\) of at most \(m\) blue vertices, and an assignment \(f : R \to S\) mapping each red vertex to a (blue) neighbor in \(S\), such that each vertex \(v \in S\) has at most \(c(v)\) red neighbors assigned to it?
\end{minipage}}
\vspace{3mm}

The triangulation of the outer \(k\)-planar drawing and the corresponding weak dual remain unchanged from the CDS setting. In particular, each subproblem corresponds to a subpolygon \(P_t\) with interface \(A_t\) and boundary \(B_t\), where \(|A_t| \le k\) and \(|B_t|\le k+2\), and the interaction between subproblems is confined to these boundary vertices.

The main modification lies in the DP state definition. We distinguish between red and blue boundary vertices. For each red vertex \(r\in B_t\cap R\), we use two states: \(r\) is labeled \(U\) if it is already assigned to a blue vertex within \(P_t\), and \(O\) if it remains unassigned and must be handled outside \(P_t\). For each blue vertex \(b\in B_t\cap B\), we use two types of states: \(b\) is labeled \(F\) if it has not yet been selected into the solution within \(P_t\), and no red vertex in \(P_t\) is assigned to it; or \(S(\rho)\) if \(b\in S\) and has remaining capacity \(\rho\). As in the algorithm for CDS, for interface vertices we store a value \(\rho\in\{0,1,\dots,d_t^{\mathrm{out}}(b),k^+\}\). For the endpoints of the triangulation link \(\lambda_t=a_tb_t\), however, we store their exact remaining capacities, which may be as large as \(n-1\). Once two subpolygons are merged, two of these vertices cease to be endpoints of a triangulation link, and their remaining capacities are truncated to values in \(\{0,1,\dots,d_t^{\mathrm{out}}(b),k^+\}\).

For each node \(t\) and each boundary signature \(\sigma\), the table entry \(Q[t,\sigma]\) stores the minimum number of selected blue vertices in \(P_t\setminus B_t\) that realize \(\sigma\), such that each internal red vertex from \(P_t\setminus B_t\) is assigned to a blue dominator within \(P_t\), and all capacity constraints are satisfied.

\paragraph{Leaf nodes.}
The initialization at a leaf is modified in the natural way. Let \(t\) be a leaf node whose corresponding subpolygon consists of a single triangle with vertices \(\{a,b,c\}\), where \(ab\) is the parent link. Since the triangle contains only three vertices, we may enumerate all possibilities directly. More precisely, we enumerate which blue vertices of \(\{a,b,c\}\) are selected into the partial solution. For every red vertex in the triangle, we then decide whether it is already assigned inside the triangle or left unresolved for the exterior, provided that this is consistent with its boundary signature. If a red vertex is assigned internally, it must be assigned to a selected blue neighbor. Afterwards we verify that no selected blue vertex receives more assigned red vertices than its capacity. From this information we derive the boundary signature \(\sigma\). A boundary red vertex is labeled \(U\) if it is assigned within the triangle, and \(O\) otherwise. A boundary blue vertex is labeled \(S(\rho)\) if it is selected, where \(\rho\) is its remaining capacity after serving the red vertices, and it is labeled \(F\) if it is not selected and no red vertex is assigned to it. The table entry \(Q[t,\sigma]\) is then updated with the number of selected blue vertices in \(P_t\setminus B_t\). Since the triangle has constant size, this initialization takes constant time for each signature \(\sigma\).

\paragraph{Internal nodes.}
The merge procedure for internal nodes is almost the same as in the CDS algorithm, with modifications reflecting the different roles of red and blue vertices. Let \(t\) have two children \(t_1,t_2\), and define \(W=(B_{t_1}\cup B_{t_2})\setminus B_t\). Red vertices in \(W\) labeled \(O\) must be assigned at this step, since they disappear from the boundary of the current subpolygon \(P_t\). For red vertices from \(B_{t_1}\) and \(B_{t_2}\) that remain in \(B_t\), we branch on whether to keep them as \(O\) or to assign them immediately. Formally, we enumerate a subset \(R'\) of boundary red vertices to be resolved at the current node. Since \(|B_t|\le k+2\), this branching contributes a factor \(2^{O(k)}\).

When merging two subpolygons, the compatibility condition at the shared vertex \(c\) must distinguish whether \(c\) is red or blue. Suppose first that \(c\) is a red vertex. Then in each child signature the state of \(c\) is either \(U\) or \(O\). The combination \((U,U)\) is forbidden, since a red vertex can be assigned only once. The combinations \((U,O)\), \((O,U)\), and \((O,O)\) are allowed. In the first two cases, the parent interprets \(c\) as already assigned; in the third case, the parent keeps \(c\) unresolved if \(c \in B_t\) and delays the assignment to subsequent merge steps, otherwise if \(c \notin B_t\) this pair of child signatures is discarded, since no vertex can dominate \(c\) in this case.

Suppose now that the shared vertex \(c\) is a blue vertex. Then each child signature labels \(c\) either by \(F\) or by \(S(\rho)\). The combinations \((F,F)\), \((F,S(\rho))\), and \((S(\rho),F)\) are always compatible. If both children label \(c\) by \(S(\rho_1)\) and \(S(\rho_2)\), then the two partial loads assigned to \(c\) must be combined. Since \(c\) is an endpoint of both child links \(\lambda_{t_i}\), the values \(\rho_1\) and \(\rho_2\) are stored exactly. If the capacity of \(c\) is \(c(c)\), then after merging the two child solutions the remaining capacity of \(c\) is \(\rho_0 = \rho_1 + \rho_2 - c(c) \). If \( \rho_0 < 0 \), then the total load assigned to \(c\) exceeds its capacity, and the pair of child signatures is discarded. Otherwise the parent node \(t\) labels \(c\) by \(S(\rho_0)\) if \(c \in B_t\). Since \(c\) is not an endpoint of the parent link \(\lambda_t\), if \(c \in B_t\) then \(c\) becomes an interface vertex of the subpolygon \(P_t\), and its remaining capacity can be truncated to values in \(\{0,1,\dots,d_t^{\mathrm{out}}(c),k^+\}\).

Let \(R''\) denote the set of red vertices that will be assigned at a merge step (including all forced vertices and those selected in \(R'\)). We then determine whether there exists a valid assignment of vertices in \(R''\) to blue vertices with state \(S(\rho)\) or \(F\) in the current subproblem, subject to capacity constraints, and the restriction that a red vertex originating from one child cannot be assigned to a blue vertex entirely within that same child. This assignment involves only \(O(k)\) vertices and can be solved by enumeration in \(k^{O(k)}\) time. If a feasible assignment exists, we update the states of boundary vertices accordingly: assigned red vertices become \(U\), unassigned ones remain \(O\); blue vertices used in the assignment become \(S(\rho)\) with updated remaining capacity, while unused blue vertices remain \(F\). The table entry value is updated by combining the child entry values and adding the number of newly selected blue vertices from \((B_{t_1}\cup B_{t_2})\setminus B_t\).

\paragraph{The root.}
At the root whose corresponding triangle is incident to an outer-face link, the boundary contains only two vertices. All red vertices must be assigned, hence no boundary vertex can be labeled \(O\). We enumerate all valid boundary signatures satisfying this condition and take the minimum table value, adding the contribution of selected boundary blue vertices if necessary.

\begin{corollary}
\textsc{Capacitated Red-Blue Dominating Set} can be solved in time \(n^5\cdot k^{O(k)}\) on outer \(k\)-planar graphs. Therefore, the problem is FPT when parameterized by outer \(k\)-planarity.
\end{corollary}

\begin{proof}
The correctness of the algorithm follows the same argument as that for \textsc{Capacitated Dominating Set}. The number of states per subproblem is again bounded by \(O(n^2 \cdot k^{O(k)})\), and each merge step requires \(2^{O(k)}\) branching and a local assignment check on \(O(k)\) vertices. Thus the overall running time remains \(O(n^5 \cdot k^{O(k)})\). Therefore, \textsc{Capacitated Red-Blue Dominating Set} on outer \(k\)-planar graphs is FPT when parameterized by \(k\).
\end{proof}

\subsection{Capacitated Vertex Cover}\label{CVC}
Another well-studied capacitated variant of a classical graph problem is \textsc{Capacitated Vertex Cover}. In addition to the XNLP-hardness result when parameterized by pathwidth and the XALP-hardness result when parameterized by treewidth, it was previously shown by Dom et al.~\cite{CDS-tw} to be W[1]-hard when parameterized by treewidth alone, and FPT when parameterized by the combined parameter of treewidth and solution size. For recent progress on this problem, in terms of complexity results with respect to several other parameters, we refer to the work by Lampis and Vasilakis~\cite{CVC-recent}. In this section, we present an FPT algorithm for \textsc{Capacitated Vertex Cover} on outer $k$-planar graphs when parameterized by $k$.

\vspace{3mm}
\noindent\fbox{
\begin{minipage}{0.97\textwidth}
\textsc{Capacitated Vertex Cover}\\
\textbf{Input:} Graph \(G = (V, E)\), with each vertex \(v \in V\) having a positive integer capacity \(c(v)\), and an integer \(m\). \\
\textbf{Question:}  Is there a set \(S \subseteq V\) with \(|S| \le m\), and an assignment of each edge to an incident vertex in \(S\) such that no vertex \(v \in S\) has more than \(c(v)\) edges assigned to it?
\end{minipage}}
\vspace{3mm}

As before, we perform dynamic programming on the triangulated graph. We root its weak dual at an arbitrary node whose corresponding triangle is incident to an outer-face link $e_0 = a_0 b_0$, and we use the same notation as in the previous sections. For every non-root triangle $\triangle_t$, let $\lambda_t = a_t b_t$ be the triangulation link shared with its parent, and let $c_t$ be the third vertex of the triangle. For each link $\lambda_t = a_t b_t$, we define $A_t$ to be the set of vertices $x \in P_t \setminus \{a_t, b_t\}$ for which there exists an edge $xy \in E$ with $y \notin P_t$. Let $B_t = \{a_t, b_t\} \cup A_t$. Then $|A_t| \le k$ and $|B_t| \le k+2$. We also define the set of crossing edges as $X_t = \{xy \in E : x \in P_t \setminus \{a_t, b_t\},\ y \notin P_t\}$, and again $|X_t| \le k$. Edges incident to $a_t$ or $b_t$ that leave $P_t$ are not in $X_t$, since they do not cross $\lambda_t$.

The key observation is that, for a vertex \(x\in A_t\), every edge from \(x\) to the outside of \(P_t\) belongs to \(X_t\). Hence the entire future interaction of \(x\) with the outside is explicitly encoded by the states of the crossing edges incident with \(x\). In particular, there is no need to store an exact residual capacity for \(x\). It is enough to remember whether \(x\) has already been selected into the cover \(S\). By contrast, the link endpoints \(a_t\) and \(b_t\) may have many edges to vertices outside \(P_t\) that do not cross \(\lambda_t\), and therefore their exact current load must be stored explicitly.

\paragraph{Table definition.}
We now define the DP table. For each non-root node \(t\), a state consists of four parts. First, \(\sigma_a\) is the state of \(a_t\), and \(\sigma_b\) is the state of \(b_t\). Each of them is either \(F\), meaning that the current load is zero and the vertex need not yet be regarded as selected, or \(S(\ell)\), meaning that the vertex has already been selected and exactly \(\ell\) already processed edges are assigned to it. Here \(0\le \ell\le \min\{c(v),d_G(v)\}\le n-1\). Second, for every \(x\in A_t\), we store a bit \(\chi_t(x)\in\{0,1\}\), where \(\chi_t(x)=1\) means that \(x\) is selected into \(S\), and \(\chi_t(x)=0\) means that it is not selected. Third, for every crossing edge \(f=xy\in X_t\), where \(x\in P_t\setminus\{a_t,b_t\}\) is the inside endpoint and \(y\notin P_t\) is the outside endpoint, we store a value \(e_t(f)\in\{\mathrm{IN},\mathrm{OUT}\}\). The meaning is that \(\mathrm{IN}\) says that \(f\) is covered by \(x\), while \(\mathrm{OUT}\) says that \(f\) is covered by \(y\).

Let \(E_t=E(G[P_t])\cup X_t\). We define \(Q[t,\sigma_a,\sigma_b,\chi_t,e_t]\) to be the minimum value of \(|S\cap (P_t\setminus B_t)|\) over all pairs \((S,\varphi)\) such that \(S\subseteq P_t\), the mapping \(\varphi\) assigns every edge of \(E_t\) to one of its endpoints, every internal edge \(uv\in E(G[P_t])\) is assigned to a vertex in \(\{u,v\}\cap S\), every crossing edge \(xy\in X_t\) is assigned according to \(e_t\), the load of every vertex in \(S\cap P_t\) satisfies the capacity constraint, every \(x\in A_t\) satisfies \(x\in S\) if and only if \(\chi_t(x)=1\), and the states \(\sigma_a,\sigma_b\) describe exactly the current loads of \(a_t\) and \(b_t\). If no such pair \((S,\varphi)\) exists, then \(Q[t,\sigma_a,\sigma_b,\chi_t,e_t]=+\infty\). The table definition indicates vertices in \(B_t\) are not counted in the objective at node \(t\); they are counted later, at the first node where they cease to belong to the boundary vertices.

\paragraph{Leaf nodes.}
We next describe the initialization at a leaf. Let \(t\) be a leaf. Then the triangle \(\triangle_t\) is an ear triangle \((a,b,c)\), where \(\lambda_t=ab\) is the parent link, and \(P_t=\{a,b,c\}\). In this case \(A_t\subseteq\{c\}\), and \(X_t\) consists precisely of the edges from \(c\) to the outside that cross the link \(ab\). For every state \((\sigma_a,\sigma_b,\chi_t,e_t)\), we enumerate the assignments of the local edges among \(ab\), \(ac\), and \(bc\) that are present in \(G\). There are at most \(2^3\) such choices. The crossing edges are forced by \(e_t\). We compute the loads of \(a\), \(b\), and \(c\), and we check whether the endpoint states, the bit \(\chi_t(c)\) when \(c\in A_t\), and all capacity constraints are satisfied. If they are, then the corresponding table entry is set to 0 if \(c \in B_t\); if \(c \notin B_t\) and \(c \in S\) then set it to 1. Otherwise the value is set to \(+\infty\). Since the number of states at a leaf is \(O(n^2\cdot 2^{O(k)})\), the time cost for a leaf node is \(O(n^2\cdot 2^{O(k)})\).

\paragraph{Internal nodes.}
For an internal node \(t\), let its triangle \(\triangle_t\) be \((a,b,c)\), with parent link \(ab\), left child link \(ac\), and right child link \(cb\). We denote the children by \(t_L\) and \(t_R\). We merge a finite entry \(Q[t_L,\sigma^L_a,\sigma^L_c,\chi_L,\lambda_L]\) with a finite entry \(Q[t_R,\sigma^R_c,\sigma^R_b,\chi_R,\lambda_R]\). The first step is to check the compatibility of the shared vertex \(c\). If \(c\) carries loads \(\ell_c^L\) and \(\ell_c^R\) in the two child states, then the total child load is \(\ell_c^L+\ell_c^R\). If this exceeds \(c(c)\), the pair is infeasible. This is the reason why \(c\), being a link endpoint in both children, must carry exact load information in the child states.

Next we classify the edges that are first resolved at node \(t\). There are edges between the two child regions, edges from the left child region to \(b\), edges from the right child region to \(a\), and edges from \(c\) to vertices outside \(P_t\). The last type is crucial: these edges do not appear in the child states, because they share the endpoint \(c\) with the child links, and hence they first appear only at node \(t\) and cross the parent link \(ab\). Consider first an edge \(uv\) with \(u\) in the left child region and \(v\) in the right child region. Such an edge appears as a crossing edge in both children. The compatibility condition is that exactly one child labels it \(\mathrm{IN}\). Thus \((\mathrm{IN},\mathrm{OUT})\) and \((\mathrm{OUT},\mathrm{IN})\) are legal, while \((\mathrm{IN},\mathrm{IN})\) and \((\mathrm{OUT},\mathrm{OUT})\) are illegal. This decides uniquely which endpoint covers the edge. For every edge from the left child region to \(b\), if the left child labels it \(\mathrm{IN}\), then it is already covered by the left endpoint and nothing further happens at the parent. If the left child labels it \(\mathrm{OUT}\), then the edge is forced to be covered by \(b\). We let \(\Delta_b^{\mathrm{ch}}\) denote the number of such forced assignments to \(b\). Symmetrically we define \(\Delta_a^{\mathrm{ch}}\) from the right child.

We then enumerate the newly appearing crossing edges incident with \(c\). Let \(Y_c=\{cx\in E : x\notin P_t\}\). Since every such edge crosses the parent link \(ab\), \(|Y_c|\le k\). We enumerate a map \(\mu_c:Y_c\to\{\mathrm{IN},\mathrm{OUT}\}\), where \(\mathrm{IN}\) means that the edge is covered by \(c\), and \(\mathrm{OUT}\) means that it is left to the outside endpoint. Each \(\mathrm{IN}\)-edge in \(Y_c\) contributes one additional unit to the load of \(c\). Finally, we enumerate the assignments of the \emph{local} edges among \(ab\), \(ac\), and \(bc\) that first appear at node \(t\). Again this introduces only a constant factor. From these choices we derive the parent state. If \(\ell_a^L\) is the load of \(a\) inherited from the left child, then the parent load of \(a\) is \(\ell_a=\ell_a^L+\Delta_a^{\mathrm{ch}}+\Delta_a^{\mathrm{loc}}\), where \(\Delta_a^{\mathrm{loc}}\) counts the local edges assigned to \(a\). Similarly \(\ell_b=\ell_b^R+\Delta_b^{\mathrm{ch}}+\Delta_b^{\mathrm{loc}}\). For the shared vertex \(c\), the parent load is \(\ell_c=\ell_c^L+\ell_c^R+\Delta_c^{\mathrm{new}}+\Delta_c^{\mathrm{loc}}\), where \(\Delta_c^{\mathrm{new}}\) is the number of edges in \(Y_c\) labeled \(\mathrm{IN}\), and \(\Delta_c^{\mathrm{loc}}\) is the number of local edges assigned to \(c\). If any of these loads exceeds the corresponding capacity, the combination is infeasible. Otherwise the endpoint states of the parent are \(F\) when the load is zero and \(S(\ell)\) when the load is positive. For \(c\), if \(c\in A_t\), then for the parent \(\chi_t(c)=1\) if \(\ell_c>0\) and \(\chi_t(c)=0\) otherwise. If \(c\notin A_t\), then \(c\) is forgotten at this node, and if \(\ell_c>0\), then \(c\) contributes one unit to \(S\) at this node. For all other vertices of \(A_t\), their \(\chi\) values are simply inherited from the corresponding child.

The parent crossing-edge state \(e_t\) is obtained by inheriting all child crossing edges whose outside endpoint is still outside \(P_t\), and by adding the new labels from \(\mu_c\) on the edges of \(Y_c\). This yields a unique candidate parent state. We then update the table entry by taking the minimum of the current value and the sum of the two child values plus the number of selected vertices that are forgotten at node \(t\). This completes the recurrence. If one of the two children is absent and the corresponding child link lies on the outer face, we simply omit the missing child table and treat the corresponding local edge \(ac\) or \(bc\), if present in \(G\), as one of the local edges first introduced at node \(t\). The recurrence remains unchanged.

\paragraph{The root.}
At the root we use the outer-face link \(e_0=a_0b_0\). Then \(P_{e_0}=V\), hence \(X_{e_0}=\emptyset\) and \(A_{e_0}=\emptyset\). Therefore the table of the root has only states of the two endpoints \(a_0\) and \(b_0\). The final answer is obtained by taking the minimum of \(Q[r,\sigma_{a_0},\sigma_{b_0},\emptyset,\emptyset]\) over all states of the root and then adding \(1\) for each of \(a_0\) and \(b_0\) that is selected. The input is a yes-instance if and only if this minimum value is at most \(m\).

\paragraph{Correctness.}
We prove the correctness by induction on the rooted weak-dual. The induction invariant states that each table entry is exactly the minimum number of selected internal vertices among all partial solutions realizing the corresponding state. The base case is immediate from the constant-size brute-force initialization at a leaf. For the inductive step, take any feasible parent partial solution. Its restrictions to the two child regions yield feasible child states. Every edge between the two child regions is necessarily encoded by one \(\mathrm{IN}\) and one \(\mathrm{OUT}\). Every child edge to the opposite parent endpoint \(a\) or \(b\) contributes to \(\Delta_a^{\mathrm{ch}}\) or \(\Delta_b^{\mathrm{ch}}\), every new edge from \(c\) to the outside contributes to a choice of \(\mu_c\), and every local edge has a definite assignment. Hence every feasible parent solution is generated by the recurrence. Conversely, every compatible pair of child solutions together with valid choices for the new crossing edges and local edges clearly combine into a feasible parent solution, because every newly introduced edge is assigned to one of its endpoints and all capacities are checked explicitly. This proves the induction invariant.

\begin{theorem}
\textsc{Capacitated Vertex Cover} can be solved in time \(n^5 \cdot 2^{O(k)}\) on outer \(k\)-planar graphs. Therefore, the problem is FPT when parameterized by outer \(k\)-planarity.
\end{theorem}

\begin{proof}
For a fixed node \(t\), the two endpoint states of link \(\lambda_t\) contribute a factor \(O(n^2)\). The \(\chi_t\) values on \(A_t\) contribute at most \(2^{|A_t|}\le 2^k\). The crossing-edge labels contribute at most \(2^{|X_t|}\le 2^k\). Thus each table has size \(O(n^2\cdot 4^k)\). At an internal node, we enumerate pairs of child table entries, hence \(O(n^4\cdot 16^k)\) combinations, and for each such pair we enumerate the labels on the at most \(k\) new edges of \(Y_c\) and a constant number of assignments of local edges. The work needed to check compatibility and build the parent state is \(O(k)\). Therefore the running time per internal node is \(O(n^4\cdot 2^{O(k)})\). Since the weak dual has \(O(n)\) nodes, the total running time is \(O(n^5\cdot 2^{O(k)})\). Therefore, \textsc{Capacitated Vertex Cover} is FPT when parameterized by outer \(k\)-planarity.
\end{proof}

\subsection{Outdegree Orientations and Flows}\label{TOO}
In this section, we study several orientation problems, including \textsc{Target Outdegree Orientation}, \textsc{Maximum Outdegree Orientation}, \textsc{Minimum Outdegree Orientation}, and \textsc{Circulating Orientation}. The \textsc{Circulating Orientation} problem is also called \textsc{Flow Orientation}. Closely related to these are flow problems such as \textsc{All-or-Nothing Flow}.

We first present an FPT algorithm for \textsc{Target Outdegree Orientation} on outer $k$-planar graphs, parameterized by $k$. Then we describe how to modify this algorithm slightly to obtain FPT algorithms for \textsc{Maximum/Minimum Outdegree Orientation}. Finally, we give reductions from  \textsc{All-or-Nothing Flow} and \textsc{Circulating Orientation} to \textsc{Target Outdegree Orientation}, and show that these two problems are FPT as well.

Given a directed graph $G = (V,E)$ with edge weights $w(e)$, the weighted outdegree of a vertex $v$ is defined as the sum of the weights of all arcs leaving $v$, that is, $\sum_{vx \in E} w(vx)$.

\vspace{3mm}
\noindent\fbox{
\begin{minipage}{0.97\textwidth}
\textsc{Target Outdegree Orientation}\\
\textbf{Input:} Undirected graph \(G = (V, E)\), a positive integer weight \(w(e)\) for each edge \(e \in E\), and for each vertex \(v \in V\), a positive integer target \(t(v)\).\\
\textbf{Question:}  Is there an orientation of \(G\), such that for each \(v\), the weighted outdegree of \(v\) equals \(t(v)\)?
\end{minipage}}
\vspace{3mm}

The \textsc{Minimum Outdegree Orientation} and \textsc{Maximum Outdegree Orientation} problems are defined analogously, except that the weighted outdegree of each vertex $v$ is required to be at least $t(v)$, respectively at most $t(v)$. In addition to being XNLP-complete when parameterized by pathwidth and XALP-complete when parameterized by treewidth, all of these three problems are known to be FPT when parameterized by tree-partition width~\cite{xnlp-flows}.

As before, we perform dynamic programming on the triangulated graph. We root the weak dual of the triangulation at a node whose corresponding triangle is incident to an outer-face link $e_0 = a_0 b_0$. We use the same notation as in the previous sections. For every non-root triangle \(\triangle_t\), let \(\lambda_t=a_tb_t\) be the triangulation link shared with its parent, and let \(c_t\) be the third vertex of the triangle. As before, we denote by \(P_t\) the set of vertices on the side of \(\lambda_t\) corresponding to the subtree rooted at \(t\), i.e. the vertex set of the subpolygon cut off by \(\lambda_t\), including \(a_t\) and \(b_t\).

For each link \(\lambda_t=a_tb_t\), we define the set of crossing edges by \(X_t=\{xy\in E : x\in P_t\setminus\{a_t,b_t\},\ y\notin P_t\}\). Since every such edge crosses the link \(\lambda_t\), we have \(|X_t|\le k\). We define the set of interface vertices by \(A_t=\{x\in P_t\setminus\{a_t,b_t\} : \exists xy\in E,\ y\notin P_t\}\). Clearly \(|A_t|\le |X_t|\le k\). Note that an edge incident with \(a_t\) or \(b_t\) and leaving \(P_t\) cannot cross the link \(\lambda_t\).

The key observation is that, for every vertex \(x\in A_t\), all edges from \(x\) to the outside of \(P_t\) belong to \(X_t\). Hence, once the directions of the edges in \(X_t\) are fixed, the entire contribution of the outside to the weighted outdegree of \(x\) is already fixed as well. Therefore, no additional residual value needs to be stored for vertices of \(A_t\). By contrast, the endpoints \(a_t\) and \(b_t\) may have many edges to vertices outside \(P_t\) that do not belong to \(X_t\), and therefore exact residual target values must be stored for them explicitly.

\paragraph{Table definition.}
For every non-root triangle \(\triangle_t\), we define a table \(Q[t,\alpha,\beta,e]\). Here \(\alpha\in \{0,1,\dots,t(a_t)\}\) denotes how much weighted outdegree vertex \(a_t\) still needs to obtain from edges outside the current subpolygon \(P_t\), and \(\beta\in \{0,1,\dots,t(b_t)\}\) is defined analogously for \(b_t\). Furthermore, \(e\) is a function on \(X_t\) with values in \(\{0,1\}\). For an edge \(xy\in X_t\), where \(x\in P_t\setminus\{a_t,b_t\}\) is the inside endpoint and \(y\notin P_t\) is the outside endpoint, the value \(e(xy)=1\) means that the edge is oriented as \(x\to y\), while \(e(xy)=0\) means that it is oriented as \(y\to x\).

We define \(Q[t,\alpha,\beta,e]=1\) if and only if there exists an orientation of the edge set \(E(G[P_t])\cup X_t\) such that every vertex \(v\in P_t\setminus\{a_t,b_t\}\) has weighted outdegree exactly \(t(v)\), vertex \(a_t\) has weighted outdegree exactly \(t(a_t)-\alpha\), vertex \(b_t\) has weighted outdegree exactly \(t(b_t)-\beta\), and every edge in \(X_t\) is oriented according to \(e\). Otherwise we set \(Q[t,\alpha,\beta,e]=0\). Since every edge incident with a vertex in \(P_t\setminus\{a_t,b_t\}\) already belongs to \(E(G[P_t])\cup X_t\), the definition implies the weighted outdegree of every vertex in \(P_t\setminus\{a_t,b_t\}\) is fully determined and must already equal its target.

\paragraph{Leaf nodes.}
We first describe the initialization for a leaf. Let \(t\) be a leaf of the rooted weak-dual. Then the triangle \(\triangle_t\) is an ear triangle with vertices \(a\), \(b\), and \(c\), where \(\lambda_t=ab\) is the parent link, and \(P_t=\{a,b,c\}\). In this case \(X_t\) consists precisely of the edges from \(c\) to vertices outside \(P_t\) that cross the link \(ab\), and \(|X_t|\le k\). To determine whether \(Q[t,\alpha,\beta,e]=1\), we only need to enumerate the directions of the \emph{local} edges among \(ab\), \(ac\), and \(bc\) that are present in \(G\). There are at most \(2^3\) such choices. The directions of edges in \(X_t\) are already forced by \(e\). For every orientation of the local edges, we compute the weighted outdegrees of \(a\), \(b\), and \(c\), and we check whether they are equal to \(t(a)-\alpha\), \(t(b)-\beta\), and \(t(c)\), respectively. If at least one local orientation satisfies these equalities, then \(Q[t,\alpha,\beta,e]=1\), and otherwise \(Q[t,\alpha,\beta,e]=0\).

\paragraph{Internal nodes.}
For an internal triangle \(\triangle_t\), let its parent link be \(\lambda_t=ab\), and let the third vertex of the triangle be \(c\). Suppose first that both children exist. Then the left child corresponds to the link \(\lambda_L=ac\), and the right child corresponds to the link \(\lambda_R=cb\). We will derive entries of the form \(Q[t,\alpha,\beta,e]\) from entries \(Q[t_L,\alpha_L,\gamma_L,e_L]\) and \(Q[t_R,\gamma_R,\beta_R,e_R]\), where \(\gamma_L\) and \(\gamma_R\) are the residual targets of the shared vertex \(c\) in the two child tables.

The edges that are first handled at node \(t\) fall into five groups. The first group consists of edges whose one endpoint lies in the strict interior of the left child region and whose other endpoint lies in the strict interior of the right child region. Every such edge appears in both \(X_{t_L}\) and \(X_{t_R}\). The second group consists of edges from the strict interior of the left child region to \(b\). The third group consists of edges from the strict interior of the right child region to \(a\). The fourth group consists of the edges \(Y_c=\{cz\in E : z\notin P_t\}\). These edges do not appear in the child tables because they share the endpoint \(c\) with the child links \(ac\) and \(cb\), and they therefore cross only the parent link \(ab\). The fifth group consists of local triangle edges that are first introduced at \(\triangle_t\), namely \(ab\), and also \(ac\) or \(bc\) when the corresponding child is absent.

We first impose consistency on the edges between the two child regions. Consider such an edge \(uv\), where \(u\) lies in the strict interior of the left child region and \(v\) lies in the strict interior of the right child region. In the left child table, the edge is viewed as an edge from an inside endpoint to an outside endpoint, while in the right child table the same edge is viewed in the opposite way. Thus the two child labels must be complementary. If we encode by \(1\) that the inside endpoint is the tail, then we require \(e_L(uv)+e_R(uv)=1\). If this condition fails for some such edge, then the two child entries are incompatible.

We next determine the residual values for the endpoints \(a\) and \(b\). The left child already realizes weighted outdegree \(t(a)-\alpha_L\) for \(a\). At node \(t\), additional outgoing contribution to \(a\) may arise from edges of the third group that are oriented as \(a\to y\), and from local edges first introduced at \(t\) that are oriented away from \(a\). Let \(\Delta_a\) denote the total weight of all such new outgoing edges of \(a\). Then the parent residual for \(a\) must be \(\alpha=\alpha_L-\Delta_a\). If \(\alpha<0\), then the contribution to \(a\) already exceeds its target, and the combination is infeasible. Symmetrically, if \(\Delta_b\) denotes the total weight of all newly introduced outgoing edges of \(b\), then the parent residual for \(b\) must be \(\beta=\beta_R-\Delta_b\), and again \(\beta\) must be non-negative.

The shared vertex \(c\) is handled by a balancing condition. In the left child table, the edges already represented there contribute weighted outdegree \(t(c)-\gamma_L\) to \(c\). In the right child table, the edges already represented there contribute weighted outdegree \(t(c)-\gamma_R\) to \(c\). Hence, after combining the child subproblems, the already realized contribution to \(c\) is \(2t(c)-\gamma_L-\gamma_R\). In order for the final weighted outdegree of \(c\) to be exactly \(t(c)\), the total weight of all newly introduced edges oriented away from \(c\) at node \(t\) must therefore equal \(\gamma_L+\gamma_R-t(c)\). These newly introduced outgoing edges of \(c\) come from the set \(Y_c\) and from the local edges first introduced at \(t\). If \(\Delta_c\) denotes the total weight of those new edges that are oriented away from \(c\), then we require \(\Delta_c=\gamma_L+\gamma_R-t(c)\). If the right-hand side is negative, then the two children already force too much weighted outdegree on \(c\), and the combination is infeasible.

The crossing-edge state \(e\) for the parent problem is obtained by inheriting from \(e_L\) and \(e_R\), where all child crossing edges whose outside endpoint remains outside \(P_t\), and by adding the new directions chosen for the edges of \(Y_c\). The edges between the two child regions do not belong to the parent crossing set \(X_t\), because both of their endpoints are already inside \(P_t\). The recurrence is now straightforward. We enumerate all finite pairs of child entries \(Q[t_L,\alpha_L,\gamma_L,e_L]\) and \(Q[t_R,\gamma_R,\beta_R,e_R]\). For each pair, we check the consistency of the edges between the child regions. We then enumerate all orientations of the edges in \(Y_c\), and all orientations of the local edges first introduced at \(t\). For each such enumeration, we compute \(\Delta_a\), \(\Delta_b\), and \(\Delta_c\), and we test whether \(\alpha=\alpha_L-\Delta_a\) and \(\beta=\beta_R-\Delta_b\) are non-negative and whether \(\Delta_c=\gamma_L+\gamma_R-t(c)\). If all conditions are satisfied, we construct the parent function \(e\) as described above and set the table entry \(Q[t,\alpha,\beta,e]=1\). If one of the two children is absent, then the same recurrence applies with the missing child omitted. In that case the missing side contributes no previously realized outdegree. Thus, for the endpoint on that side, the starting residual is simply its full target, and the corresponding local edge, if present in \(G\), is treated as one of the local edges first introduced at node \(t\). The leaf case is actually the special case in which both children are absent.

\paragraph{The root.}
At the root we use the fixed outer-face link \(e_0=a_0b_0\). Since no edge can cross an outer-face link, we have \(X_r=\emptyset\) for the root triangle \(\triangle_r\). Furthermore, \(P_r=V\), and hence the entire graph is represented in the root table. The input instance is a yes-instance if and only if \(Q[r,0,0,\emptyset]=1\). The two zero residuals express that the endpoints \(a_0\) and \(b_0\) need no further outgoing contribution from outside, because at the root there is no outside.

\paragraph{Correctness.}
We prove the correctness by induction on the rooted weak-dual. The induction invariant states that \(Q[t,\alpha,\beta,e]=1\) if and only if there exists an orientation of \(E(G[P_t])\cup X_t\) realizing the state \((\alpha,\beta,e)\) exactly as specified in the definition of the table. The base case holds because at a leaf the subproblem contains only the ear triangle and its crossing edges, and all local edge orientations are checked exhaustively. For the inductive step, consider first a feasible parent solution. Restricting it to the left and right child regions yields feasible child states. Every edge between the child regions must induce complementary labels in the two child tables. Every edge from a child region to the opposite endpoint contributes to \(\Delta_a\) or \(\Delta_b\). Every edge from \(c\) to the outside of \(P_t\) appears for the first time at the parent and contributes to \(\Delta_c\). The equalities for \(\alpha\), \(\beta\), and \(\Delta_c\) are therefore necessary. Hence every feasible parent solution is generated by the recurrence. Conversely, if the recurrence accepts a pair of child states and a set of new edge directions, then the child orientations can be combined because the shared edges between the child regions are consistent, the newly introduced edges are oriented explicitly, and the balancing equations guarantee that the weighted outdegrees of \(a\), \(b\), and \(c\) are correct. Thus every accepted combination yields a feasible parent solution. This proves the induction invariant.

\begin{theorem}
    \textsc{Target Outdegree Orientation} can be solved in time \(2^{O(k)}\cdot n^{O(1)}\) on outer \(k\)-planar graphs, assuming that all edge weights and targets are encoded in unary.
\end{theorem}

\begin{proof}
For a fixed node \(t\), the values of \(\alpha\) and \(\beta\) range over \(0,\dots,t(a_t)\) and \(0,\dots,t(b_t)\), respectively, and \(e\) ranges over at most \(2^{|X_t|}\le 2^k\) possibilities. Hence the table size at \(t\) is \(O((t(a_t)+1)\cdot(t(b_t)+1)\cdot 2^k)\). Since the targets are given in unary, this is polynomial in the input size times \(2^k\). At an internal node, the recurrence enumerates pairs of child table entries, and for each pair it enumerates the at most \(2^{|Y_c|}\le 2^k\) directions of the edges in \(Y_c\) and a constant number of local edge orientations. All compatibility checks and all computations of \(\Delta_a\), \(\Delta_b\), and \(\Delta_c\) take \(O(k)\) time. Therefore each internal node can be processed in time \(2^{O(k)}\cdot n^{O(1)}\). Since the weak dual of the triangulation has \(O(n)\) nodes, the total running time is \(2^{O(k)}\cdot n^{O(1)}\). Therefore, \textsc{Target Outdegree Orientation} is FPT on outer \(k\)-planar graphs when parameterized by \(k\).
\end{proof}

\begin{corollary}
\textsc{Minimum Outdegree Orientation} is FPT on outer \(k\)-planar graphs parameterized by \(k\), assuming that all edge weights and target lower bounds are encoded in unary.
\end{corollary}

\begin{proof}
We simply modify the algorithm for \textsc{Target Outdegree Orientation}. Recall that, for a triangulation link \(\lambda=a_tb_t\), the algorithm stores exact residual values for the two endpoints \(a_t\) and \(b_t\), and stores the orientations of all crossing edges whose inner endpoint is not one of \(a_t,b_t\). For \textsc{Minimum Outdegree Orientation}, the same triangulated graph, the same subpolygons \(P_t\), and the same crossing-edge states can be used. The only change is the interpretation of the residual value of a link endpoint. For a vertex \(v\in\{a_t,b_t\}\), let \(p_{\lambda}(v)\) denote the total weight of edges already oriented out of \(v\) inside the current subpolygon \(P_t\). We store the remaining deficit \(\rho_{\lambda}(v)=\max\{0,t(v)-p_{\lambda}(v)\}\), thus \(\rho_{\lambda}(v)\) is the amount of outgoing weight that \(v\) still needs from the part outside the current subpolygon in order to satisfy the lower bound.

When two child subproblems are merged, suppose that the newly processed edges contribute total outgoing weight \(\Delta(v)\) to a link endpoint \(v\). The updated residual is then \(\rho'(v)=\max\{0,\rho(v)-\Delta(v)\}\). A transition is feasible whenever all local orientations are compatible with the crossing-edge states and the resulting residual values are in their allowed ranges. For every other boundary vertex \(x\in B_t\setminus\{a_t,b_t\}\), all edges from \(x\) to the outside of the current subpolygon are among the crossing edges of the link, and their orientations are already fixed by the crossing-edge state. Hence the full weighted outdegree of \(x\) is known at the moment when \(x\) is no longer a link endpoint, and we only need to check whether it is at least \(t(x)\).

At the root, there is no outside part left. Therefore a table entry is accepting if the final residual value of every remaining outer-face link endpoint is \(0\). The DP table size is unchanged, and the number of crossing-edge states is still \(2^{O(k)}\). Hence the running time remains \(2^{O(k)}\cdot n^{O(1)}\).
\end{proof}

\begin{corollary}
\textsc{Maximum Outdegree Orientation} is FPT on outer \(k\)-planar graphs parameterized by \(k\), assuming that all edge weights and target upper bounds are encoded in unary.
\end{corollary}

\begin{proof}
There are two ways to obtain the result. First, one can modify the algorithm for \textsc{Target Outdegree Orientation} directly. For a triangulation link \(\lambda=a_tb_t\), the crossing-edge state is kept exactly as before. For each link endpoint \(v\in\{a_t,b_t\}\), instead of storing a deficit, we store a remaining allowance. If \(p_{\lambda}(v)\) is the total weight of edges already oriented out of \(v\) inside the current subpolygon \(P_t\), then the stored value is \(\rho_{\lambda}(v)=t(v)-p_{\lambda}(v)\). A transition is feasible only if this value never becomes negative. When a merge introduces additional outgoing weight \(\Delta(v)\) at \(v\), the residual allowance is updated as \(\rho'(v)=\rho(v)-\Delta(v)\). For every boundary vertex \(x \in B_t \setminus \{a_t,b_t\}\), all its edges to the outside of the current subpolygon are represented among the crossing edges, so its final weighted outdegree is already known when it is no longer a link endpoint. At that moment we check that the weighted outdegree of \(x\) is at most \(t(x)\). At the root, it is enough that the remaining allowances of the link endpoints \(a_0\) and \(b_0\) are nonnegative. This gives the same type of \(2^{O(k)}\cdot n^{O(1)}\) algorithm as for \textsc{Target Outdegree Orientation}.

Second, one can also reduce \textsc{Maximum Outdegree Orientation} to \textsc{Minimum Outdegree Orientation} on the same graph. For each vertex \(v\), let \(W(v)=\sum_{e\in E,\, v\in e} w(e)\), i.e. the total weight of all edges incident with \(v\). Let \(\operatorname{out}_w(v)\) be the weighted outdegree of \(v\). An orientation satisfies \(\operatorname{out}_w(v)\le t(v)\) if and only if the reversed orientation satisfies \(\operatorname{out}_w(v)\ge W(v)-t(v)\). Hence we define a new lower bound \(t'(v)=\max\{0,W(v)-t(v)\}\) for every vertex \(v\). The graph, the edge weights, and the given outer \(k\)-planar drawing are unchanged. Therefore the FPT algorithm for \textsc{Minimum Outdegree Orientation} applies directly.
\end{proof}

Next, we show that the two problems \textsc{All-or-Nothing Flow} and \textsc{Circulating Orientation} are also FPT on outer $k$-planar graphs when parameterized by $k$.

A \emph{flow network} is a tuple $(G,s,t,c)$, where $G=(V,E)$ is a directed graph, $s,t \in V$ are two designated vertices, and $c : E \to \mathbb{Z}^+$ is a capacity function assigning a positive integer capacity to each arc. A \emph{flow} is a function $f : E \to \mathbb{N}$ such that $0 \le f(e) \le c(e)$ for every arc $e \in E$, and for every vertex $v \in V \setminus \{s,t\}$, the total incoming flow equals the total outgoing flow. The value of a flow is defined as the total flow leaving $s$ minus the total flow entering $s$. A flow is called \emph{all-or-nothing} if for every arc $e \in E$, we have $f(e) \in \{0, c(e)\}$.

We assume that there are no parallel arcs, and that for any pair of vertices $x,y \in V$, at most one of the arcs $xy$ and $yx$ is present in $E$. The \textsc{All-or-Nothing Flow} problem asks whether there exists a flow of a given value such that every edge carries either zero flow or its full capacity.

\vspace{3mm}
\noindent\fbox{
\begin{minipage}{0.97\textwidth}
\textsc{All-or-Nothing Flow}\\
\textbf{Input:} A flow network $(G,c,s,t)$ and an integer $r$. \(c(e)\) and \(r\) are encoded in unary. \\
\textbf{Question:} Is there an all-or-nothing flow from $s$ to $t$ of value exactly $r$ in $G$?
\end{minipage}}
\vspace{3mm}

We remark that the variant asking whether there exists a flow of value at least $r$ is equivalent in difficulty, since one can introduce a new source $s'$ and add an arc from $s'$ to $s$ with capacity $r$.

\begin{theorem}
\textsc{All-or-Nothing Flow} is FPT on outer \(k\)-planar graphs parameterized by \(k\).
\end{theorem}

\begin{proof}
We use the same reduction from \textsc{All-or-Nothing Flow} to \textsc{Target Outdegree Orientation} as in~\cite{outerkplanar}. Let $(G,s,t,c)$ be a flow network, where $G=(V,E)$ is a directed outer \(k\)-planar graph, $c : E \to \mathbb{Z}^+$ is a capacity function, and $s,t \in V$ are the source and sink.

We construct an undirected graph $H=(V,F)$ by ignoring the directions of the arcs in $E$. For each edge $e \in F$, we assign a weight $w(e)$ equal to the capacity of the corresponding arc in $G$. For each vertex $v \in V \setminus \{s,t\}$, define \(t(v) = \sum_{xv \in E} c(xv)\). For the source and sink, set \(t(s) = \sum_{xs \in E} c(xs) + r\) and \(t(t) = \sum_{xt \in E} c(xt) - r\).

We show that there exists an all-or-nothing flow of value exactly $r$ in $G$ if and only if $H$ admits an orientation such that every vertex $v$ has weighted outdegree exactly $t(v)$.

\textbf{($\Rightarrow$)} Suppose there exists an all-or-nothing flow $f$ of value $r$ in $G$. We define an orientation of $H$ as follows. For each edge $\{x,y\} \in F$, if the corresponding arc $xy$ in $G$ carries flow $c(xy)$, we orient the edge from $x$ to $y$. Otherwise, we orient it from $y$ to $x$.

We verify that the weighted outdegree condition holds. Consider a vertex $v \in V \setminus \{s,t\}$. Let $\alpha$ be the total flow entering $v$. The total capacity of incoming arcs is $t(v)$, so the incoming arcs with zero flow contribute $t(v) - \alpha$ to the weighted outdegree, as they are oriented outward. The outgoing arcs carrying flow contribute $\alpha$, since they are oriented according to the direction of the flow. Hence the total weighted outdegree of $v$ equals $(t(v) - \alpha) + \alpha = t(v)$.

The cases for $s$ and $t$ are analogous, taking into account the definition of the flow value. It follows that the constructed orientation satisfies the required outdegree constraints.

\textbf{($\Leftarrow$)} Conversely, suppose $H$ has an orientation in which every vertex $v$ has weighted outdegree exactly $t(v)$. We construct a flow $f$ in $G$ by sending flow $c(xy)$ along an arc $xy$ if and only if the corresponding edge $\{x,y\}$ is oriented from $x$ to $y$ in $H$, and sending zero flow otherwise. For any vertex $v \in V \setminus \{s,t\}$, let $\beta$ be the total weight of arcs \(xv\) whose edges are directed as \(xv\). Then, the weight of the edges with an incoming arc to \(v\) which are directed out of \(v\) is \(t(v) - \beta\), so the weight of the edges \(vy\) with an outgoing arc from \(v\) that are directed as \(vy\) is \(\beta\). Therefore, both the inflow and outflow of \(v\) equal \(\beta\), so the flow conservation is satisfied. The analyses for \(s\) and \(t\) are similar. Finally, the definitions of $t(s)$ and $t(t)$ ensure that the value of the flow is exactly $r$. Therefore, $f$ is an all-or-nothing flow of value $r$ in $G$. Note that the reduction does not change the drawing of \(G\), therefore \(H\) is still outer \(k\)-planar. This completes the proof. 
\end{proof}

The \textsc{Circulating Orientation} problem is defined as follows. In addition to being XNLP-complete when parameterized by pathwidth and XALP-complete when parameterized by treewidth, earlier results show that it is NP-hard even on planar graphs~\cite{DIDIMO201972}. It was also shown in~\cite{circulating-orientation} that the problem is W[1]-hard when parameterized by treewidth.

\vspace{3mm}
\noindent\fbox{
\begin{minipage}{0.97\textwidth}
\textsc{Circulating Orientation}\\
\textbf{Input:} A graph $G = (V, E)$ together with a positive integer weight $w(e)$ for each edge $e \in E$.\\
\textbf{Question:} Does there exist an orientation of $G$ such that, for every vertex $v \in V$, the weighted outdegree of $v$ equals its weighted indegree?
\end{minipage}}

\begin{theorem}
\textsc{Circulating Orientation} is FPT on outer \(k\)-planar graphs parameterized by \(k\), provided that the edge weights are encoded in unary.
\end{theorem}

\begin{proof}
We reduce the problem to \textsc{Target Outdegree Orientation}. Let \((G,w)\) be an instance of \textsc{Circulating Orientation}, where \(G=(V,E)\) is an undirected outer \(k\)-planar graph and \(w:E\to \mathbb N\) is the edge-weight function. For every vertex \(v\in V\), we define \(W(v)=\sum_{e\in E,\, v\in e} w(e)\), that is, \(W(v)\) is the total weight of all edges incident with \(v\).

Suppose first that \(G\) has a circulating orientation. Then for every vertex \(v\), the weighted outdegree of \(v\) equals the weighted indegree of \(v\). Since the sum of the weighted outdegree and the weighted indegree of \(v\) is exactly \(W(v)\), it follows that both of them must be equal to \(W(v)/2\). Hence \(W(v)\) must be even for every vertex \(v\). Therefore, if there exists a vertex \(v\) such that \(W(v)\) is odd, we immediately reject the instance. Otherwise, for every vertex \(v\in V\), define its target value \(t(v)=W(v)/2\). We construct an instance of \textsc{Target Outdegree Orientation} on the same graph \(G\), with the same edge-weight function \(w\), and with target function \(t\).

We claim that the two instances are equivalent. Let \(D\) be an orientation of \(G\). If \(D\) is a circulating orientation, then for every vertex \(v\), its weighted outdegree is \(W(v)/2=t(v)\). Hence \(D\) is a feasible solution for the constructed \textsc{Target Outdegree Orientation} instance. Conversely, suppose \(D\) is an orientation such that every vertex \(v\) has weighted outdegree \(t(v)=W(v)/2\). Since the weighted indegree of \(v\) is \(W(v)-t(v)=W(v)/2\), the weighted outdegree and weighted indegree of \(v\) are equal. Thus \(D\) is a circulating orientation as well.

The reduction does not modify the drawing of the graph. Hence if the input graph is outer \(k\)-planar, then the constructed \textsc{Target Outdegree Orientation} instance is also outer \(k\)-planar. Moreover, if the weights are encoded in unary, then the values \(W(v)\) and \(t(v)\) are also computable and representable in polynomial time in the input size. By the FPT algorithm for \textsc{Target Outdegree Orientation} on outer \(k\)-planar graphs, the corresponding circulating orientation instance can be solved in FPT time as well. Therefore, \textsc{Circulating Orientation} is FPT on outer \(k\)-planar graphs parameterized by \(k\). This completes the proof.
\end{proof}

\subsection{$f$-Domination and $q$-Domination}\label{f-k-dom}
In this section, we study another variant of the dominating set problem, namely \textsc{$f$-Dominating Set}. In contrast to the capacitated variants of dominating set that we considered in Section~\ref{CDS}, where each selected vertex can dominate only a bounded number of neighbors, the \textsc{$f$-Dominating Set} problem requires that vertices not in the solution must be dominated by multiple neighbors. Formally, the problem is defined as follows.

\vspace{3mm}
\noindent\fbox{
\begin{minipage}{0.97\textwidth}
\textsc{$f$-Dominating Set}\\
\textbf{Input:} A graph $G=(V,E)$, a demand function $f: V \to \mathbb{N}$, and an integer $m$.\\
\textbf{Question:} Does there exist a set $S \subseteq V$ with $|S| \le m$ such that for every vertex $v \in V$, either $v \in S$ or $|N(v) \cap S| \ge f(v)$?
\end{minipage}}
\vspace{3mm}

A special case of \textsc{$f$-Dominating Set} is the \textsc{$q$-Dominating Set} problem, where the demand is uniform, that is, $f(v)=q$ for all vertices $v \in V$. The problem is defined as follows.

\vspace{3mm}
\noindent\fbox{
\begin{minipage}{0.97\textwidth}
\textsc{$q$-Dominating Set}\\
\textbf{Input:} A graph $G=(V,E)$ and integers $\ell, q$.\\
\textbf{Question:} Does there exist a set $S \subseteq V$ with $|S| \le \ell$ such that for every vertex $v \in V$, either $v \in S$ or $|N(v) \cap S| \ge q$?
\end{minipage}}
\vspace{3mm}

The notion of a $q$-dominating set was introduced by Fink and Jacobson~\cite{n-dom} in 1985 and has been extensively studied in graph theory. For fixed $q$, the problem can be solved in time $O(q^{O(\mathrm{tw})} n)$ using dynamic programming on tree decompositions~\cite{k-dom-tw}. The more general notion of $f$-domination was later introduced by Chen and Zhou~\cite{f-dom}. It can be seen that \textsc{$f$-Dominating Set} is also FPT with respect to the combined parameter of treewidth and the maximum demand value $\max_{v \in V} f(v)$, again via dynamic programming on tree decompositions in~\cite{k-dom-tw}.

We present an FPT algorithm for \textsc{\(f\)-Dominating Set} on outer \(k\)-planar graphs, parameterized by \(k\). Since \textsc{\(q\)-Dominating Set} can be seen as a special case of \textsc{\(f\)-Dominating Set}, both problems are FPT on outer \(k\)-planar graphs.  We use the same triangulated graph and the same rooted weak dual as before, and also keep the notation introduced earlier. In particular, for a non-root triangle \(\triangle_t\), the parent link is denoted by \(\lambda_t=a_tb_t\), the corresponding subpolygon is denoted by \(P_t\), the interface vertices are denoted by \(A_t\), and the boundary is \(B_t=\{a_t,b_t\}\cup A_t\). For a vertex \(v\in B_t\), let \(d_t^{\mathrm{out}}(v)=|N(v)\setminus P_t|\). Note that \(d_t^{\mathrm{out}}(v)\leq k\) for every \(v\in A_t\), while for \(v\in\{a_t,b_t\}\) this number can be as large as \(n-1\). A vertex in the solution set \(S\) is called \emph{selected}.

\paragraph{Table definition.}
For a node \(t\) in the weak dual, a state \(\sigma\) assigns to every vertex \(v\in B_t\) one of the following types. The first type is \(\mathrm{Sel}\), meaning that \(v\) is selected. The second type is \(\mathrm{Need}(r)\), meaning that \(v\) is not selected and still needs \(r\) selected neighbors from outside \(P_t\). The value \(r\) is only allowed if \(0\leq r\leq d_t^{\mathrm{out}}(v)\). This restriction is essential: if \(r>d_t^{\mathrm{out}}(v)\), then the outside of \(P_t\) cannot provide enough selected neighbors for \(v\).

For a state \(\sigma\), we define \(Q[t,\sigma]\) to be the minimum value of \(|S_t\cap (P_t\setminus B_t)|\) over all sets \(S_t\subseteq P_t\) satisfying the following conditions. If \(\sigma(v)=\mathrm{Sel}\), then \(v\in S_t\). If \(\sigma(v)=\mathrm{Need}(r)\), then \(v\notin S_t\) and \(r=\max\{0,f(v)-|N(v)\cap S_t\cap P_t|\}\). Moreover, every vertex \(u\in P_t\setminus B_t\) must already be satisfied inside \(P_t\), that is, either \(u\in S_t\), or \(|N(u)\cap S_t\cap P_t|\geq f(u)\). If no such set \(S_t\) exists, then \(Q[t,\sigma]=+\infty\). The table entry value \(Q[t,\sigma]\) counts only selected vertices that are strictly inside \(P_t\). Selected boundary vertices are counted later, when they disappear from the boundary, or at the root if they remain on the final outer-face link \(e_0=a_0b_0\).

\paragraph{Leaf nodes.}
Let \(t\) be a leaf of the weak dual. Then the triangle corresponding to \(t\) is an ear triangle with vertices \(a_t,b_t,c_t\), and \(P_t=\{a_t,b_t,c_t\}\). For every possible state \(\sigma\) on \(B_t\), we compute \(Q[t,\sigma]\) by brute force. We enumerate all subsets \(S_t\subseteq P_t\) that agree with the state indicated by \(\sigma\). For each such set, we compute, for every \(v\in B_t\) with \(v\notin S_t\), the value \(r_v=\max\{0,f(v)-|N(v)\cap S_t\cap P_t|\}\). The set \(S_t\) realizes \(\sigma\) if \(\sigma(v)=\mathrm{Need}(r_v)\) for every non-selected boundary vertex \(v\), and if \(r_v\leq d_t^{\mathrm{out}}(v)\) for all such vertices. In addition, every vertex in \(P_t\setminus B_t\) must be satisfied inside \(P_t\). For every feasible set \(S_t\), we update \(Q[t,\sigma]\) with the value \(|S_t\cap(P_t\setminus B_t)|\). Since \(P_t\) contains only three vertices, this takes constant time for each state \(\sigma\).

\paragraph{Internal nodes.}
Let \(t\) be an internal node whose triangle is \(\triangle(a,b,c)\), and suppose that the parent link of \(t\) is \(\lambda_t=ab\). Let \(t_1\) and \(t_2\) be the two children corresponding to the links \(ac\) and \(cb\), respectively. We write \(P_i=P_{t_i}\) and \(B_i=B_{t_i}\) for \(i\in\{1,2\}\).

We combine two child states \(\sigma_1\) and \(\sigma_2\) with finite table values. The selected status must be compatible on common boundary vertices. In particular, for the shared vertex \(c\), either both child states mark \(c\) as \(\mathrm{Sel}\), or both child states mark \(c\) as a \(\mathrm{Need}\)-state. A pair in which one child marks \(c\) as selected and the other marks it as non-selected is discarded. Let \(U=B_1\cup B_2\). The set of selected vertices in \(U\) is determined by the two child states. For a vertex \(v\in B_i\) with \(\sigma_i(v)=\mathrm{Need}(r_i)\), define \(m_i(v)=|N(v)\cap U\cap (P_t\setminus P_i)\cap S_U|\), where \(S_U\) denotes the selected vertices of \(U\). Therefore, \(m_i(v)\) is the number of available selected neighbors of \(v\) that are outside the child subpolygon \(P_i\) but inside the parent subpolygon \(P_t\). This number can be computed from the two child states, because every edge from \(B_i\) to \(P_t\setminus P_i\) is represented within \(U\).

First consider a vertex \(v\in B_i\setminus\{c\}\) with \(\sigma_i(v)=\mathrm{Need}(r_i)\). Its remaining need after the merge is \(r=\max\{0,r_i-m_i(v)\}\). If \(v\in B_t\), then the parent state must assign \(\mathrm{Need}(r)\) to \(v\), and this is valid only if \(r\leq d_t^{\mathrm{out}}(v)\). If \(v\notin B_t\), then \(v\) disappears from the boundary at \(t\), and therefore we require \(r=0\). If this condition fails, the pair of child states cannot be merged. If \(v\in B_i\setminus\{c\}\) is marked as \(\mathrm{Sel}\), then it remains marked as \(\mathrm{Sel}\) in the parent state if \(v\in B_t\). If \(v\notin B_t\), then \(v\) is forgotten at \(t\), and we add one to the cost of the merge, i.e. \(v\) should be counted now.

It remains to handle the shared vertex \(c\). If both children mark \(c\) as \(\mathrm{Sel}\), then \(c\) is selected in the merged partial solution. If \(c\in B_t\), the parent state marks \(c\) as \(\mathrm{Sel}\). If \(c\notin B_t\), then \(c\) is forgotten at \(t\), and we add one to the cost of the merge. If both children mark \(c\) as non-selected, say \(\sigma_1(c)=\mathrm{Need}(r_1)\) and \(\sigma_2(c)=\mathrm{Need}(r_2)\), then the remaining need of \(c\) after merging the two child subpolygons is \(r_c=\max\{0,r_1+r_2-f(c)\}\). Indeed, if \(s_i\) is the number of selected neighbors of \(c\) inside \(P_i\), then \(r_i=\max\{0,f(c)-s_i\}\), and the number of selected neighbors of \(c\) inside \(P_t\) is \(s_1+s_2\). Hence the remaining need inside the parent is \(\max\{0,f(c)-s_1-s_2\}\), which is equal to \(\max\{0,r_1+r_2-f(c)\}\). If \(c\in B_t\), the parent state assigns \(\mathrm{Need}(r_c)\) to \(c\), and this is valid only if \(r_c\leq d_t^{\mathrm{out}}(c)\). If \(c\notin B_t\), then we require \(r_c=0\); otherwise the merge is invalid.

Whenever all these conditions are satisfied, the two child states generate a parent state \(\sigma\). We update \(Q[t,\sigma]\) with the value \(Q[t_1,\sigma_1]+Q[t_2,\sigma_2]+\mu\), where \(\mu\) is the number of newly selected vertices in \((B_1\cup B_2)\setminus B_t\). If several pairs of child states generate the same parent state, we keep the minimum value.

\paragraph{The one-child case.}
Suppose that a node \(t\) has only one child \(t_1\), and assume that the child link is \(ac\). The current triangle is \(\triangle(a,b,c)\), the parent link is \(\lambda_t=ab\), and the side \(bc\) lies on the outer face. In this case \(P_t=P_{t_1}\cup\{b\}\). The vertex \(c\) is already contained in the child boundary, since it is an endpoint of the child link \(ac\). Hence the only new vertex outside the child subpolygon is \(b\). Let \(\sigma_1\) be a finite child state. We extend \(\sigma_1\) by deciding whether \(b\) is selected. If \(b\) is selected, then the parent state marks \(b\) as \(\mathrm{Sel}\). Since \(b\in B_t\), it is not counted at this node.

Now consider a vertex \(v\in B_{t_1}\). If \(\sigma_1(v)=\mathrm{Sel}\), then \(v\) remains selected in the parent state whenever \(v\in B_t\). If \(v\notin B_t\), then \(v\) is forgotten at \(t\), and we add one to the merge cost. If \(\sigma_1(v)=\mathrm{Need}(r)\), then the only newly available selected neighbor of \(v\) in \(P_t\setminus P_{t_1}\) can be vertex \(b\). Let \(\eta_b(v)=1\) if \(b\) is selected and \(vb\in E(G)\), and let \(\eta_b(v)=0\) otherwise. The remaining need of \(v\) after passing from \(P_{t_1}\) to \(P_t\) is \(r'=\max\{0,r-\eta_b(v)\}\). If \(v\in B_t\), then the parent state assigns \(\mathrm{Need}(r')\) to \(v\). If \(v\in A_t\), we additionally require \(r'\le d_t^{\mathrm{out}}(v)\). If \(v\notin B_t\), then \(v\) is forgotten at \(t\), and therefore we require \(r'=0\); otherwise the merge of \(B_{t_1}\) and \(\{b\}\) is invalid.

It remains to determine the state of the new vertex \(b\) when \(b\) is not selected. Since \(b\) is an endpoint of the parent link, its remaining need is stored exactly. All selected neighbors of \(b\) inside \(P_t\) that can already be seen at this node are determined by the child state. Indeed, any edge from \(b\) to a vertex of \(P_{t_1}\setminus\{a,c\}\) crosses the child link \(ac\), and hence its endpoint in \(P_{t_1}\) lies in \(B_{t_1}\). Let \(s_b\) be the number of selected vertices in \(B_{t_1}\) adjacent to \(b\). If \(b\) is not selected, then the parent state assigns \(\mathrm{Need}(r_b) \) to \(b\), where \(r_b=\max\{0,f(b)-s_b\}\).

The value of the parent table \(Q[t,\sigma]\) is obtained by adding the number of newly selected vertices in \(B_{t_1}\setminus B_t\) to the child table value \(Q[t_1,\sigma_1]\). The new vertex \(b\) is not counted here, because \(b\in B_t\). The case in which the unique child corresponds to the link \(bc\) is symmetric, with the roles of \(a\) and \(b\) exchanged.

\paragraph{The root.}
Let \(\triangle_r\) be the root triangle incident with the outer-face link \(e_0=a_0b_0\). We treat \(e_0\) as the parent link of \(\triangle_r\). Then \(P_r=V(G)\) and \(B_r=\{a_0,b_0\}\). Since the outside of \(P_r\) is empty, a non-selected root-boundary vertex can only appear as \(\mathrm{Need}(0)\). The instance is accepted if there exists a root state \(\sigma\) such that every non-selected vertex of \(\{a_0,b_0\}\) is in state \(\mathrm{Need}(0)\), and \(Q[r,\sigma]+\nu\leq m\), where \(\nu\) is the number of vertices among \(a_0,b_0\) that are marked as \(\mathrm{Sel}\) by \(\sigma\).

\paragraph{Correctness.}
Let \(t\) be a node and let \(\sigma\) be a state on \(B_t\). We say a set \(S_t\subseteq P_t\) \emph{realizes} \(\sigma\) if the following conditions hold. For every \(v\in B_t\), \(\sigma(v)=\mathrm{Sel}\) if and only if \(v\in S_t\). If \(\sigma(v)=\mathrm{Need}(r)\), then \(v\notin S_t\) and \(r=\max\{0,f(v)-|N(v)\cap S_t\cap P_t|\}\). If additionally \(v\in A_t\), then \(r\le d_t^{\mathrm{out}}(v)\). Moreover, every vertex \(u\in P_t\setminus B_t\) is satisfied inside \(P_t\), that is, either \(u\in S_t\), or \(|N(u)\cap S_t\cap P_t|\ge f(u)\). Next, we prove the table entry \(Q[t,\sigma]\) obtained by our algorithm indeed stores the minimum value of \(|S_t\cap(P_t\setminus B_t)|\) over all sets \(S_t\) realizing the state \(\sigma\).

\begin{lemma}
\label{lem:fdom-table-correctness}
For every node \(t\) and every state \(\sigma\) on \(B_t\), after the table entry \(Q[t,\sigma]\) has been computed, its value is exactly the minimum of \(|S_t\cap(P_t\setminus B_t)|\) over all sets \(S_t\subseteq P_t\) realizing \(\sigma\).
\end{lemma}

\begin{proof}
We prove the statement by induction on the rooted weak dual. If \(t\) is a leaf, then \(P_t\) consists of the three vertices of an ear triangle. The algorithm checks all possible choices of selected vertices in this constant-size set. For each such choice it computes exactly, for every boundary vertex, whether the vertex is selected and, if not, how many selected neighbors it still needs from outside \(P_t\). It also checks whether every non-boundary vertex is already satisfied inside \(P_t\). Therefore the table value for a leaf is exactly the desired minimum.

Now let \(t\) be an internal node with two children \(t_1\) and \(t_2\). Write \(P_i=P_{t_i}\) and \(B_i=B_{t_i}\) for \(i\in\{1,2\}\). Let the triangle of \(t\) be \(\triangle(a,b,c)\), where \(ab\) is the parent link. Suppose that the algorithm takes two child states \(\sigma_1,\sigma_2\), both with finite table values, and generates a parent state \(\sigma\). By the induction hypothesis, there are sets \(S_i\subseteq P_i\) realizing \(\sigma_i\), with \(|S_i\cap(P_i\setminus B_i)|=Q[t_i,\sigma_i]\), for \(i\in\{1,2\}\). The compatibility check guarantees that the shared vertex \(c\) is selected in both child states or in neither. Hence \(S_t=S_1\cup S_2\) is well-defined as a subset of \(P_t\).

Every vertex in \(P_i\setminus B_i\) is satisfied inside \(P_i\), and therefore also inside \(P_t\). It remains only to consider vertices that belong to a child boundary. Let \(v\in B_i\) and suppose first that \(\sigma_i(v)=\mathrm{Need}(r_i)\). The algorithm subtracts from \(r_i\) exactly the number of selected neighbors of \(v\) that lie in \(P_t\setminus P_i\). Let this number be \(m_i(v)\). Thus the remaining need of \(v\) after the merge is \(r'=\max\{0,r_i-m_i(v)\}\). If \(v\in B_t\), then the parent state records \(\mathrm{Need}(r')\). If \(v\in A_t\), the algorithm also requires \(r'\le d_t^{\mathrm{out}}(v)\). If \(v\notin B_t\), then \(v\) becomes an internal vertex of \(P_t\), and the algorithm requires \(r'=0\). Hence every such forgotten non-selected vertex is satisfied inside \(P_t\).

If \(v\in B_i\) is selected, then it remains selected in the parent state whenever \(v\in B_t\). If \(v\notin B_t\), then \(v\) becomes internal at \(t\), and the algorithm adds one to the cost. Therefore selected vertices are counted exactly when they cease to be boundary vertices.

For the shared vertex \(c\), if \(c\) is selected, the same argument applies. If \(c\) is not selected, let \(\sigma_1(c)=\mathrm{Need}(r_1)\) and \(\sigma_2(c)=\mathrm{Need}(r_2)\). Let \(s_i=|N(c)\cap S_i\cap P_i|\). Since \(P_1\cap P_2=\{c\}\) and \(c\notin S_t\), the selected neighbors of \(c\) inside \(P_t\) are exactly the selected neighbors contributed by the two child subpolygons. Hence the remaining need of \(c\) in \(P_t\) is \(\max\{0,f(c)-s_1-s_2\}\). Since \(r_i=\max\{0,f(c)-s_i\}\), this value is \(\max\{0,r_1+r_2-f(c)\}\). This is exactly the value used by the algorithm. If \(c\in B_t\), this value is recorded in the parent state and is required to be at most \(d_t^{\mathrm{out}}(c)\). If \(c\notin B_t\), the algorithm requires it to be zero.

It follows that \(S_t\) realizes the parent state \(\sigma\). The cost assigned by the algorithm is \(Q[t_1,\sigma_1]+Q[t_2,\sigma_2]\) plus the number of selected vertices in \((B_1\cup B_2)\setminus B_t\). This is precisely \(|S_t\cap(P_t\setminus B_t)|\). Thus every value written into \(Q[t,\sigma]\) corresponds to a feasible set realizing \(\sigma\). This proves the soundness. Next, we prove the completeness of the algorithm.

Conversely, let \(S_t\subseteq P_t\) be any set realizing a parent state \(\sigma\). We show that the algorithm considers child states whose combination gives a value at most \(|S_t\cap(P_t\setminus B_t)|\). For \(i\in\{1,2\}\), let \(S_i=S_t\cap P_i\). Define a state \(\sigma_i\) on \(B_i\) as follows. If \(v\in B_i\cap S_t\), set \(\sigma_i(v)=\mathrm{Sel}\). If \(v\in B_i\setminus S_t\), set \(\sigma_i(v)=\mathrm{Need}(r_i)\), where \(r_i=\max\{0,f(v)-|N(v)\cap S_t\cap P_i|\}\). This is a valid child state. Indeed, for interface vertices in \(A_{t_i}\) the value \(r_i\) cannot exceed the number of neighbors outside \(P_i\), hence \(r_i\le d_{t_i}^{\mathrm{out}}(v)\). If \(v\) is one of the two endpoints of the child link, then \(r_i\) is stored exactly.

The set \(S_i\) realizes \(\sigma_i\). Therefore, by the induction hypothesis, \(Q[t_i,\sigma_i]\le |S_i\cap(P_i\setminus B_i)|\). The two child states \(\sigma_1\) and \(\sigma_2\) are compatible, because both are induced by the same set \(S_t\). For every non-selected boundary vertex \(v\in B_i\), the algorithm subtracts exactly the selected neighbors of \(v\) that lie in \(P_t\setminus P_i\), which are precisely the selected neighbors that were not counted inside the child subpolygon. Hence the remaining need computed by the algorithm is the same as the remaining need of \(v\) in the parent state. If \(v\notin B_t\), then \(v\) is an internal vertex of \(P_t\), and since \(S_t\) realizes \(\sigma\), it is satisfied inside \(P_t\). Hence the algorithm accepts the forgetting of \(v\). The same calculation applies to the shared vertex \(c\), where the remaining need is \(\max\{0,r_1+r_2-f(c)\}\).

Therefore the algorithm generates the parent state \(\sigma\) from the induced child states \(\sigma_1\) and \(\sigma_2\). The table entry value generated is 
\( Q[t_1,\sigma_1]+Q[t_2,\sigma_2]+\mu \),
where \(\mu\) is the number of selected vertices in \((B_1\cup B_2)\setminus B_t\). By the inequalities above, this is at most
\( |S_1\cap(P_1\setminus B_1)|+|S_2\cap(P_2\setminus B_2)|+\mu \).
These three terms count exactly the selected vertices of \(S_t\cap(P_t\setminus B_t)\): selected vertices strictly inside the two child subpolygons are counted by the first two terms, and selected child-boundary vertices that become internal at \(t\) are counted by \(\mu\). Hence the algorithm computes a value for \(Q[t,\sigma]\) at most
\( |S_t\cap(P_t\setminus B_t)| \). Since \(S_t\) is an arbitrary set realizing \(\sigma\), the table value \(Q[t,\sigma]\) cannot be larger than the desired minimum. Together with the first direction, this proves the claim for nodes with two children. The one-child case is proved in the same way. If the unique child corresponds to the link \(ac\), then \(P_t=P_{t_1}\cup\{b\}\). The algorithm decides whether \(b\) is selected, updates the remaining needs of the child boundary vertices by using the possible selected neighbor \(b\), and computes the state of \(b\) itself if it is not selected. The case where the unique child corresponds to \(bc\) is symmetric. By induction, the statement holds for all nodes.
\end{proof}

\begin{theorem}
\label{thm:fdom-outer-k-planar}
\(f\)-\textsc{Dominating Set} can be solved in time \(n^5k^{O(k)}\) on outer \(k\)-planar graphs.
\end{theorem}

\begin{proof}
For every node \(t\), the boundary \(B_t\) has size at most \(k+2\). For each vertex of \(B_t\setminus\{a_t,b_t\}\), there are at most \(k+2\) possible states, namely \(\mathrm{Sel}\) and \(\mathrm{Need}(r)\) with \(0\leq r\leq d_t^{\mathrm{out}}(v)\leq k\). For each parent link endpoint \(a_t,b_t\), there are at most \(n+1\) possible states. Hence the number of states at \(t\) is bounded by \(n^2k^{O(k)}\). For a node with two children, we enumerate pairs of child states. This gives at most \(n^4k^{O(k)}\) pairs. For each pair, all compatibility checks and residual updates are performed by inspecting the local set \(B_{t_1}\cup B_{t_2}\), whose size is \(O(k)\) and therefore takes \(k^{O(1)}\) time. Since the weak dual has \(O(n)\) nodes, the total running time is \(n^5k^{O(k)}\). By Lemma~\ref{lem:fdom-table-correctness}, the algorithm accepts exactly the yes-instances.
\end{proof}

\begin{corollary}
\label{cor:qdom-outer-k-planar}
For every positive integer \(q\), \(q\)-\textsc{Dominating Set} can be solved in time \(n^5k^{O(k)}\) on outer \(k\)-planar graphs.
\end{corollary}

\begin{proof}
Given an instance of \(q\)-\textsc{Dominating Set} with budget \(\ell\), define an instance of \(f\)-\textsc{Dominating Set} on the same graph by setting \(f(v)=q\) for every vertex \(v\), and by setting the solution budget to \(\ell\). A set \(S\) is feasible for the resulting \(f\)-domination instance if and only if \(S\) is a \(q\)-dominating set of the original instance. The graph and the drawing are unchanged, and hence the parameter \(k\) is unchanged. The result follows directly from Theorem~\ref{thm:fdom-outer-k-planar}.
\end{proof}

\subsection{Target Set Selection}\label{TSS}
The \textsc{Target Set Selection} problem is an NP-hard problem in graph theory and social network analysis that seeks a smallest subset of vertices, called a \emph{target set} to initially activate in a network such that, through a cascading process, all other vertices become active based on given thresholds. It models influence spreading, such as finding the best influencers for marketing.

Formally, let $G = (V,E)$ be a graph, let $t : V \to \mathbb{N}$ be a threshold function, and let $S \subseteq V$ be an initial set of active vertices called \emph{seeds}. The activation process starting from $S$ is defined as a sequence of sets $\mathrm{Active}[0] \subseteq \mathrm{Active}[1] \subseteq \cdots$, where $\mathrm{Active}[0] = S$, and for each $i \ge 1$, a vertex $v$ belongs to $\mathrm{Active}[i]$ if either $v \in \mathrm{Active}[i-1]$ or $v$ has at least $t(v)$ neighbors in $\mathrm{Active}[i-1]$. The process continues until a step $j$ is reached such that either $\mathrm{Active}[j] = V$ or $\mathrm{Active}[j] = \mathrm{Active}[j-1]$. We define $\mathrm{Active}(S) = \mathrm{Active}[j]$. The problem is defined as follows.

\vspace{3mm}
\noindent\fbox{
\begin{minipage}{0.97\textwidth}
\textsc{Target Set Selection}\\
\textbf{Input:} A graph $G = (V,E)$, a threshold function $t : V \to \mathbb{N}$, and an integer $m$.\\
\textbf{Question:} Does there exist a set $S \subseteq V$ of size at most \(m\) such that $\mathrm{Active}(S) = V$?
\end{minipage}}
\vspace{3mm}

Ben-Zwi et al.~\cite{BENZWI201187} gave an XP algorithm for this problem when parameterized by treewidth, together with a lower bound of $n^{\Omega(\sqrt{\mathrm{tw}})}$. Their result also implies W[1]-hardness.

In this section, we present an FPT algorithm for \textsc{Target Set Selection} on outer $k$-planar graphs, parameterized by $k$. Our algorithm is inspired by the dynamic programming framework of Ben-Zwi et al.~\cite{BENZWI201187} for bounded-treewidth graphs. In their algorithm, one stores for every subproblem, i.e. a subtree of the tree decomposition, both a threshold vector and an activation order for boundary vertices. The threshold information records how much help a boundary vertex still needs from outside the subtree, while the activation order prevents deadlocks at join bags.

The crucial difference in our setting is that the weak dual of the triangulation is not a nice tree decomposition, and therefore the join formula from the treewidth setting cannot be transferred verbatim. Instead, we use the same conceptual ingredients, but the merge step is replaced by a local verification procedure on a set of vertices of size at most $O(k)$.

Before describing the algorithm, we first prove the following Lemma~\ref{lem:witness-order}, which forms the theoretical foundation underlying both the algorithm of Ben-Zwi et al.~\cite{BENZWI201187} and our algorithm. A \emph{weak order} on a vertex set \(X\) is a partition of \(X\) into ordered levels. If a vertex \(u\) appears in a strictly earlier level than a vertex \(v\), we write \(u\prec_\rho v\). Given a seed set \(S\), a weak order \(\rho\) on \(V(G)\) is a \emph{witness order} for \(S\) if every vertex of \(S\) lies in level \(0\), every vertex of \(V(G)\setminus S\) lies in a later level, and every vertex \(v\in V(G)\setminus S\) has at least \(t(v)\) neighbors \(u\) with \(u\prec_\rho v\).

Given an initial seed set $S$, we define the \emph{standard activation process} as the activation sequence induced by repeatedly propagating activations from $S$. This sequence forms a weak order on all vertices according to their earliest possible activation times.

\begin{lemma}
\label{lem:witness-order}
Let \(G\) be a graph with threshold function \(t\), and let \(S\subseteq V(G)\). Then the following are equivalent.
\begin{enumerate}
    \item Starting from \(S\), the standard activation process activates all vertices of \(G\).
    \item There exists a witness order for \(S\).
\end{enumerate}
\end{lemma}

\begin{proof}
Assume first that the standard activation process from \(S\) activates all vertices. For every vertex \(v\), let \(\rho(v)\) be the first round in which \(v\) becomes active. Then every vertex of \(S\) lies in round \(0\), and every non-seed vertex \(v\) has at least \(t(v)\) neighbors that became active in earlier rounds. Thus \(\rho\) is a witness order for \(S\).

Conversely, assume that there exists a witness order \(\rho\) for \(S\). We prove by induction over the levels of \(\rho\) that every vertex eventually becomes active in the standard activation process. All vertices of \(S\) lie in level \(0\), and these vertices are active initially. Suppose that all vertices in earlier levels have already become active. Let \(v\) be a non-seed vertex in the next level. Since \(\rho\) is a witness order, \(v\) has at least \(t(v)\) neighbors in earlier levels, and these neighbors are active by the induction hypothesis. Therefore \(v\) eventually becomes active. Repeating this argument for all levels proves that the standard activation process can activate all vertices.
\end{proof}

Lemma~\ref{lem:witness-order} shows that the dynamic program does not need to reconstruct the unique earliest synchronous activation sequence, that is, the exact sequence in which vertices become active as early as possible starting from $S$. Instead, it is sufficient to construct some witness order $\rho$ certifying that the chosen seed set $S$ eventually activates the entire graph. In other words, the dynamic programming only needs to find a weak order $\rho$ together with a corresponding seed set $S$ such that every vertex \(v \in V\) has at least $t(v)$ neighbors appearing before it in $\rho$. The weak order $\rho$ does not need to coincide with the earliest possible activation sequence produced by the standard activation process. Moreover, the activation sequence induced by repeatedly propagating activations from $S$ is merely one possible weak order, but not the only one.

As an example, consider the graph in Figure~\ref{fig:tss}, where the threshold $t(v)$ of each vertex is indicated, and the vertex \(a\) is chosen as the only seed. One can verify that both activation orders \(a,b,(c,d,e),f\) and \(a,b,(d,e),f,c\) are valid, since in both cases every vertex \(v\) has at least $t(v)$ neighbors activated before it. Thus, it suffices to find any such valid weak order witnessing that the initial seed set $S=\{a\}$ activates all vertices of the graph. In fact, from any valid weak order, one can rearrange vertices according to their earliest possible activation times and recover the standard activation sequence induced by $S$. For the second order, we can move $c$ to its earliest possible activation time, namely after $b$, thereby recovering the standard activation sequence.

\begin{figure}[ht!]
    \centering
    \includegraphics[width=0.38\linewidth]{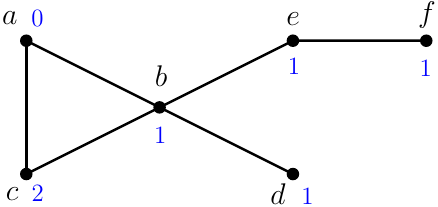}
    \caption{The standard activation order \(a,b,(c,d,e),f\) induced by repeatedly propagating activations from $S=\{a\}$ is not the only valid activation order. Another order, \(a,b,(d,e),f,c\), is also valid. Moreover, the latter can be reordered so that each vertex appears at its earliest possible activation time, thereby recovering the former sequence. Both orders certify that the seed set $S=\{a\}$ is a feasible solution. The threshold value of each vertex is marked in blue.}
    \label{fig:tss}
\end{figure}

This observation is crucial for our algorithm. For a given seed set $S$, it is enough for the algorithm to verify the existence of some valid weak order in which every vertex \(v\) has at least $t(v)$ preceding active neighbors. Such a weak order can certify that $S$ is a candidate solution.

Next, we present the full details of the algorithm. As before, we perform dynamic programming on the triangulated graph and use the same notation as in the previous sections. For each node $t$ in the weak dual of the triangulation, we store a table indexed by triples $(Z_t,O_t,\delta_t)$ with the following meaning. The set $Z_t \subseteq B_t$ is the set of boundary vertices that are already chosen as seeds in \(S\). The symbol $O_t$ denotes a weak order on the set $B_t \setminus Z_t$. We think of the vertices in $Z_t$ as lying in level $0$, before all non-seed boundary vertices. The deficit function \(\delta_t \colon B_t \setminus Z_t \to \mathbb N\) records the remaining amount of help that each non-seed boundary vertex still needs from outside $P_t$. More precisely, for every $v \in B_t \setminus Z_t$, the value $\delta_t(v)$ is the number of neighbors outside $P_t$ that must be activated strictly earlier than $v$ in order to complete a witness order for $P_t$.

For every \(v\in A_t\), we only consider values \(\delta_t(v)\in\{0,\ldots,d_t^{\mathrm{out}}(v)\}\), where \(d_t^{\mathrm{out}}(v)=|N(v)\setminus P_t|\). Since every edge from \(v\in A_t\) to \(V(G)\setminus P_t\) crosses the triangulation link \(a_tb_t\), we have \(d_t^{\mathrm{out}}(v)\le k\). For the two link endpoints $a_t$ and $b_t$, the values must be recorded exactly, and thus may range over $\{0,\dots,n\}$, where $n = |V(G)|$.

\paragraph{Table definition.}
For a state $(Z_t,O_t,\delta_t)$, let the table entry $Q_t[Z_t,O_t,\delta_t]$ denote the minimum number of seeds in $P_t \setminus B_t$ such that there exists a set \(S_t \subseteq P_t\) with $S_t \cap B_t = Z_t$ and a witness weak order $\rho_t$ on $P_t$ satisfying the following three conditions.

\begin{enumerate}
\item The restriction of $\rho_t$ to $B_t \setminus Z_t$ is $O_t$, while every vertex of $Z_t$ lies in the first level of \(\rho_t\).

\item For every vertex $u \in P_t \setminus B_t$, either $u \in S_t$, or $u$ has at least $t(u)$ neighbors in $P_t$ that lie in levels strictly earlier than the level of $u$ in $\rho_t$.

\item For every vertex $v \in B_t \setminus Z_t$, let \(p_t(v)\) be the number of neighbors of \(v\) in \(P_t\) that appear strictly before \(v\) in \(\rho_t\). We require \(\delta_t(v)=\max\{0,t(v)-p_t(v)\}\). Moreover, if \(v\in A_t \setminus Z_t\), we require \(\delta_t(v)\le d_t^{\mathrm{out}}(v)\). This means after using all earlier neighbors available inside $P_t$, the vertex $v$ still needs exactly $\delta_t(v)$ earlier neighbors from outside $P_t$ to become active.
\end{enumerate}
We say that a pair \((S_t,\rho_t)\) \emph{realizes} a state \((Z_t,O_t,\delta_t)\) at node \(t\) if it satisfies the conditions in the definition of \(Q_t[Z_t,O_t,\delta_t]\). In particular, \(\rho_t\) is not necessarily a witness order that activates all of \(P_t\) by itself. Boundary vertices may still have positive deficit. Rather, \(\rho_t\) is a local certificate showing that all internal vertices of \(P_t\) are activated, and that each boundary vertex \(v\in B_t\setminus Z_t\) needs exactly \(\delta_t(v)\) earlier neighbors from outside \(P_t\). If no such pair $(S_t,\rho_t)$ exists, then we set $Q_t[Z_t,O_t,\delta_t]=+\infty$. The value stored by the dynamic program counts only seeds in $P_t \setminus B_t$. Boundary seeds are not counted at this stage. They are counted when they disappear from the boundary, and the two root-boundary vertices, i.e. two endpoints of the outer-face link \(e_0=a_0b_0\), are added explicitly at the end. This ensures that each seed is counted exactly once in \(S\).

\paragraph{Leaf nodes.}
Suppose that $t$ is a leaf of the weak dual. Then the triangle corresponding to $t$ is an ear triangle, and $P_t$ consists exactly of the three vertices of that triangle, say $a_t$, $b_t$, and $c_t$. In this case the table can be filled by brute force. We enumerate all subsets $S_t \subseteq \{a_t,b_t,c_t\}$ and all weak orders on $\{a_t,b_t,c_t\}\setminus S_t$. Since there are only three vertices, this takes constant time. For every choice we read off the induced boundary seed set $Z_t$, the induced weak order $O_t$ on $B_t \setminus Z_t$, and the deficit values $\delta_t(v)$ for the non-seed boundary vertices. If \(c_t \notin S_t\) and \(c_t \in B_t\), we require \(\delta_t(c_t) \le d_t^{\mathrm{out}}(c_t)\), which ensures that \(c_t\) can be activated. If every non-boundary vertex can be activated within $P_t$, then we update the table entry $Q_t[Z_t,O_t,\delta_t]$ by the value $|S_t \setminus B_t|$.

\paragraph{Internal nodes.}
Let $t$ be an internal node of the weak dual. Write the corresponding triangle of $t$ as $\triangle(a,b,c)$, and suppose that its parent link is $ab$. Then the two other sides of the triangle are $ac$ and $bc$. Each of these links either corresponds to a child of $t$ or lies on the outer face. We denote the existing children by $t_1$ and $t_2$. Their boundary vertices are denoted by \(B_{t_1}\) and \(B_{t_2}\). If only one child exists, we treat the missing child differently, which will be discussed later. Note that \(B_{t_1} \cap B_{t_2} = \{c\}\), which means the two child subpolygons are disjoint except for the shared vertex $c$, and every edge that creates a new dependency between the two child subproblems has at least one endpoint on one of the child boundaries. Therefore all new interactions created at the parent level are confined to the set \(U_t := B_{t_1} \cup B_{t_2}\), whose size is bounded by $2k+3$. This is the place where our algorithm departs from the join operations of Ben-Zwi et al.~\cite{BENZWI201187}. Since the two child boundaries are not equal, we do not combine the two child states by a closed formula. Instead we explicitly verify the merge on the small set $U_t$.

Fix two child states \((Z_1,O_1,\delta_1)\) and \((Z_2,O_2,\delta_2)\) with finite table values. We first require the seed status of the shared vertex \(c\) to be consistent: $c$ must either belong to both $Z_1$ and $Z_2$, or to neither. We then enumerate a weak order $\Pi$ on
\(U_t \setminus (Z_1 \cup Z_2)\) whose restrictions to $B_{t_1}\setminus Z_1$ and $B_{t_2}\setminus Z_2$ are exactly $O_1$ and $O_2$, respectively. More precisely, the weak orders \(\Pi\) are obtained by enumerating all interleavings of the orders \(O_1\) and \(O_2\), considering only the vertices in \(U_t \setminus (Z_1 \cup Z_2)\), since vertices in \(Z_1 \cup Z_2\) are already fixed to appear in the first level. Once such an interleaving $\Pi$ is fixed, we run a local validation process on the set $U_t$.

We now compute the state induced at the parent node. Consider first a non-seed vertex \(v\in B_{t_1}\setminus\{c\}\). Let \(H_1(v)\) be the set of vertices \(u\in U_t\setminus B_{t_1}\) such that \(uv\in E(G)\), and such that either \(u\in Z_1\cup Z_2\) or \(u\prec_{\Pi} v\). Put \(h_1(v)=|H_1(v)|\). This is the number of earlier neighbors of \(v\) that become available at the parent level but lie outside the left child subpolygon. If \(v\in B_t\), then the parent deficit of \(v\) is \(\delta_t(v)=\max\{0,\delta_1(v)-h_1(v)\}\). If \(v\in A_t\), we additionally require \(\delta_t(v)\le d_t^{\mathrm{out}}(v)\); otherwise the merge of \((Z_1,O_1,\delta_1)\) and \((Z_2,O_2,\delta_2)\) is invalid. If \(v\notin B_t\), then we require \(\delta_1(v)\le h_1(v)\); otherwise the merge is also invalid since \(v\) no longer belongs to the boundary but cannot be activated. The handling for vertices \(v\in B_{t_2}\setminus\{c\}\) is symmetric.

It remains to handle the shared vertex \(c\). If \(c\) is a seed, then it is treated as a seed in both children. If \(c\notin B_t\), it will be forgotten from the boundary and should be counted in the target set \(S\) now; if \(c\in B_t\), it remains in the parent seed set. Suppose now that \(c\) is not a seed. Then both child states contain a deficit value \(\delta\) for \(c\). Since the left and right child subpolygons are disjoint except for \(c\), the numbers of earlier neighbors of \(c\) supplied by the two children add up. Hence the remaining deficit of \(c\) at the parent is \(\delta_t(c)=\max\{0,\delta_1(c)+\delta_2(c)-t(c)\}\). If \(c\notin B_t\), this value must be zero; otherwise \(c\) cannot be activated and the merge is invalid. If \(c\in B_t\), then \(c\in A_t\), and we store this value as the parent deficit of \(c\), provided that it is at most \(d_t^{\mathrm{out}}(c)\). If the deficit value is larger than \(d_t^{\mathrm{out}}(c)\), then \(c\) cannot be activated and the merge is invalid.

The parent seed set is \(Z_t=(Z_1\cup Z_2)\cap B_t\), and the parent activation order \(O_t\) is the restriction of \(\Pi\) to \(B_t\setminus Z_t\). After merging the two child states, we obtain a candidate table entry with value
\(
Q_{t_1}[Z_1,O_1,\delta_1]+
Q_{t_2}[Z_2,O_2,\delta_2]+
|(Z_1\cup Z_2)\setminus B_t|
\).
The vertices in $(Z_1 \cup Z_2)\setminus B_t$ are exactly the seeds that disappear from the boundary at the parent node and must therefore be counted in \(S\) at this level. We update the corresponding parent table entry \(Q_t[Z_t,O_t,\delta_t]\) with the minimum value obtained over all compatible child states and all compatible interleavings $\Pi$.

\paragraph{The one-child case.}
Suppose that exactly one child exists. Without loss of generality, assume that the child corresponds to the link \(ac\), and that the side \(bc\) lies on the outer face. The other case is symmetric. Let this child be \(t_1\). The vertex \(b\) is the new triangle vertex that is not contained in the child subpolygon \(P_{t_1}\). We set \(U_t=B_{t_1}\cup\{b\}\). We enumerate whether \(b\) is a seed. If \(b\) is chosen as a seed, then \(b\) is included in the parent seed set. If \(b\) is not a seed, we enumerate its position in a weak order \(\Pi\) on \(U_t\setminus Z_1\), extending the child order \(O_1\).

For vertices in \(B_{t_1}\), all cases, including whether a vertex is a seed and whether it belongs to \(B_t\) or \(A_t\), are handled in the same way as in the two-children case described above. For every non-seed vertex \(v\in B_{t_1}\), we define \(H_1(v)\) as above, using \(U_t\) in place of \(B_{t_1}\cup B_{t_2}\), and update or check its deficit in the same way. Put \(h_1(v)=|H_1(v)|\). If \(v\in B_t\), then we set \(\delta_t(v)=\max\{0,\delta_1(v)-h_1(v)\}\). If \(v \in A_t\), we require \(\delta_t(v) \le d_t^{\mathrm{out}}(v)\). If \(v\notin B_t\), then we require \(\delta_1(v)\le h_1(v)\); otherwise the merge is invalid. If \(b\) is not a seed, its parent deficit is computed directly from the earlier neighbors of \(b\) inside \(U_t\). More precisely, let \(p(b)\) be the number of neighbors \(u\in U_t\setminus\{b\}\) such that \(u\prec_\Pi b\). We set \(\delta_t(b)=\max\{0,t(b)-p(b)\}\). The resulting state is then used to update the parent table. Since only one new vertex is introduced in addition to the child boundary, this case has the same asymptotic running time as the two-child case.

\paragraph{The root.}
Let \(\triangle_r\) be the root triangle, incident with the chosen outer-face link \(e_0=a_0b_0\). Then \(P_r=V(G)\) and \(B_r=\{a_0,b_0\}\). Since there is no outside of \(P_r\), a root state \((Z_r,O_r,\delta_r)\) is accepting if \(\delta_r(v)=0\) for every \(v\in B_r\setminus Z_r\). The optimum is the minimum value of \(Q_r[Z_r,O_r,\delta_r]+|Z_r|\) over all accepting root states. The input is a yes-instance if and only if this value is at most \(m\).

\paragraph{Correctness.}
Next, we prove the correctness of the algorithm. We show that the dynamic program, which processes the subpolygons bottom-up by merging the triangles corresponding to the nodes of the weak dual of the triangulation, is both sound and complete. Consequently, the algorithm correctly finds a solution. After that, we analyze the running time of the algorithm, thereby showing that \textsc{Target Set Selection} is FPT on outer \(k\)-planar graphs.

\begin{lemma}[soundness of the merge step]
\label{lem:soundness-merge}
Suppose that a parent state \((Z_t,O_t,\delta_t)\) is produced from child states \((Z_1,O_1,\delta_1)\) and \((Z_2,O_2,\delta_2)\) by a compatible interleaving \(\Pi\) and by the merge rules described above. If the two child states are feasible, then the parent state is also feasible.
\end{lemma}

\begin{proof}
Let \(\rho_1\) and \(\rho_2\) be witness weak orders realizing the two child states. We construct a weak order \(\rho_t\) on \(P_t\) by taking an extension of the two child weak orders and the interleaving \(\Pi\) on \(U_t=B_{t_1}\cup B_{t_2}\). Such an extension exists because \(\Pi\) restricts to \(O_1\) and \(O_2\) on the two child boundaries. Inside each child subpolygon, the relative order of vertices is kept as in the corresponding child witness order.

Consider first a vertex \(u\in P_t\setminus (B_{t_1}\cup B_{t_2})\). Such a vertex is strictly internal to one child subpolygon, say \(P_{t_i}\). If \(u\) is a seed, then it is already active in the constructed witness order. Otherwise, since the child state is realized by \(\rho_i\), the vertex \(u\) has at least \(t(u)\) neighbors in \(P_{t_i}\) that appear strictly before \(u\) in \(\rho_i\). The construction of \(\rho_t\) preserves the relative order of vertices inside \(P_{t_i}\). Hence these same neighbors still appear strictly before \(u\) in \(\rho_t\). Therefore \(u\) satisfies its activation condition in the parent witness order \(\rho_t\).

Now consider a non-seed vertex \(v\in B_{t_1}\setminus\{c\}\). The child state guarantees that, inside \(P_{t_1}\), the vertex \(v\) already has enough earlier neighbors except for a deficit of \(\delta_1(v)\). By definition, \(h_1(v)\) counts the earlier neighbors of \(v\) that become available at the parent level but lie outside \(P_{t_1}\). If \(v\notin B_t\), the merge rule requires \(\delta_1(v)\le h_1(v)\), and therefore \(v\) has enough earlier neighbors inside \(P_t\). If \(v\in B_t\), the parent stores the remaining deficit \(\delta_t(v)=\max\{0,\delta_1(v)-h_1(v)\}\). If \(v \in A_t\), the merge rule additionally requires \(\delta_t(v) \le d_t^{\mathrm{out}}(v)\). Thus \(v\) satisfies the boundary condition of the parent state. The same argument applies symmetrically to vertices in \(B_{t_2}\setminus\{c\}\).

It remains to consider the shared vertex \(c\). If \(c\) is a seed, then its seed status is consistent in both children, and it is handled by the parent seed set or is forgotten from the boundary and counted in the target set \(S\). Suppose that \(c\) is not a seed. In the left child, \(c\) has at least \(t(c)-\delta_1(c)\) earlier neighbors inside \(P_{t_1}\). In the right child, it has at least \(t(c)-\delta_2(c)\) earlier neighbors inside \(P_{t_2}\). These two sets of neighbors are disjoint, since the two child subpolygons intersect only in \(c\). Hence inside \(P_t\), the vertex \(c\) has at least \(2t(c)-\delta_1(c)-\delta_2(c)\) earlier neighbors. The parent deficit \(\delta_t(c)=\max\{0,\delta_1(c)+\delta_2(c)-t(c)\}\) is exactly the remaining amount of help that \(c\) may still need from outside \(P_t\). If \(c \in B_t\), the range check \(\delta_t(c) \le d_t^{\mathrm{out}}(c)\) guarantees that the deficit can indeed be supplied by vertices outside \(P_t\). If \(c\) is forgotten at the parent, the merge rule requires this deficit value to be zero. Therefore \(c\) is also handled correctly. Thus every vertex that becomes internal at the parent satisfies its threshold inside \(P_t\), and every remaining boundary vertex has exactly the deficit recorded by \(\delta_t\). Hence \(\rho_t\) witnesses the parent state.
\end{proof}

\begin{lemma}[completeness of the merge step]
\label{lem:completeness-merge}
Let \((Z_t,O_t,\delta_t)\) be a feasible parent state of an internal node \(t\). Then there exist feasible child states \((Z_1,O_1,\delta_1)\) and \((Z_2,O_2,\delta_2)\), together with a compatible interleaving \(\Pi\), such that the merge procedure reconstructs the parent state \((Z_t,O_t,\delta_t)\).
\end{lemma}

\begin{proof}
Let \((S_t,\rho_t)\) be a pair realizing the parent state. Restrict \(S_t\) and \(\rho_t\) to the two child subpolygons \(P_{t_1}\) and \(P_{t_2}\). These restrictions define child seed sets, child boundary orders, and child deficit functions. Denote the resulting child states by \((Z_1,O_1,\delta_1)\) and \((Z_2,O_2,\delta_2)\). By construction, both child states are feasible.

Let \(\Pi\) be the restriction of \(\rho_t\) to \(U_t=B_{t_1}\cup B_{t_2}\). Then \(\Pi\) is a compatible interleaving of \(O_1\) and \(O_2\). For every ordinary vertex \(v\in B_{t_i}\setminus\{c\}\), the number \(h_i(v)\) computed by the merge procedure is exactly the number of earlier neighbors of \(v\) that lie in \(P_t\setminus P_{t_i}\). Therefore, if \(v\) is forgotten at the parent, the inequality \(\delta_i(v)\le h_i(v)\) holds; and if \(v\) remains in the parent boundary, the remaining deficit is exactly \(\max\{0,\delta_i(v)-h_i(v)\}\).

For the shared vertex \(c\), the earlier neighbors supplied by the two child subpolygons add up, because the two subpolygons intersect only in \(c\). If \(c\) is not a seed, the deficit induced at the parent is therefore exactly \(\max\{0,\delta_1(c)+\delta_2(c)-t(c)\}\). If \(c\) is forgotten at the parent, this value must be zero, since the parent state is feasible. Therefore the merge procedure applied to these child states and to \(\Pi\) reconstructs the parent state \((Z_t,O_t,\delta_t)\).
\end{proof}

\begin{theorem}
\label{thm:tss-outer-k}
\textsc{Target Set Selection} can be solved in time \(n^{5}\cdot k^{O(k)}\) on outer \(k\)-planar graphs. Therefore, the problem is FPT when parameterized by outer $k$-planarity.
\end{theorem}

\begin{proof}
Correctness follows from Lemmas~\ref{lem:witness-order}, \ref{lem:soundness-merge}, and \ref{lem:completeness-merge}, together with the leaf initialization and the root evaluation described above. It remains to bound the running time. For every node $t$, the boundary \(B_t\) has size at most $k+2$. The number of possible seed sets $Z$ is therefore at most $2^{k+2}$. The number of weak orders on at most $k+2$ elements is at most $(k+2)^{k+2}$. For the deficit function, the two endpoints contribute at most $(n+1)^2$ possibilities, while each of the at most $k$ remaining boundary vertices contributes at most $k+1$ possibilities. Hence the number of states of a single node is bounded by \((n+1)^2 \cdot 2^{k+2}\cdot (k+2)^{k+2}\cdot (k+1)^{k}=n^2\cdot k^{O(k)}\).

Now consider one internal node $t$. We enumerate two child states, which contributes at most \(n^4\cdot k^{O(k)}\). For each such pair we enumerate all compatible interleavings $\Pi$ on the set $U_t$, whose size is $O(k)$, and this contributes another factor $k^{O(k)}$. For each interleaving, the merge procedure only checks adjacencies among vertices of \(U_t\), and \(|U_t|=O(k)\). Hence this step takes \(k^{O(1)}\) time. Thus one internal node is processed in time \(n^4\cdot k^{O(k)}\). The weak dual tree of the triangulation has $O(n)$ nodes. Therefore the total running time is \(n\cdot n^4\cdot k^{O(k)}=n^5\cdot k^{O(k)}\).
\end{proof}

\section{Outer \(k\)-Planarity in the Graph Parameter Hierarchy}\label{sec:hierarchy}
In Section~\ref{mimw-bound}, we proved that every outer \(k\)-planar graph has mim-width at most \(k+2\). In this section, we further explore the relationship between various graph parameters and outer \(k\)-planarity, i.e.~convex local crossing number \(\operatorname{lcr}^\circ(G)=k\). For two graph parameters \(a\) and \(b\), we say that bounded \(a\) implies bounded \(b\) if and only if every graph with parameter \(a\) at most \(k\) also satisfies \(b \le f(k)\) for some function \(f\). For example, since the treewidth of outer \(k\)-planar graphs satisfies \(\operatorname{tw}(G)\le 1.5k+2\)~\cite{firman_et_al}, we obtain that bounded outer \(k\)-planarity implies bounded treewidth. If neither bounded \(a\) implies bounded \(b\) nor bounded \(b\) implies bounded \(a\), then we say that the two parameters \(a\) and \(b\) are \emph{incomparable}.

We show that several parameters are unbounded on outer \(k\)-planar graphs by providing counterexamples. Since bounded outer \(k\)-planarity implies bounded treewidth, many other bounded graph parameters follow automatically from the parameter hierarchy. Consequently, bounded outer \(k\)-planarity also implies that numerous other parameters are bounded by some function \(f(k)\), such as clique-width, stack number, tree-independence number, and many others. Based on these boundedness and unboundedness results, we obtain several relative relationships between these parameters and outer \(k\)-planarity within the hierarchy of graph parameters\footnote{One can consult  \url{https://graphparameterhierarchy.com/} for a visualization of the hierarchy among many known graph parameters.}.

\subsection{Unbounded Parameters}
It is easy to see that many graph parameters are not bounded by any function of \(k\) on outer \(k\)-planar graphs. For example, pathwidth is unbounded, since complete binary trees have pathwidth \(O(\log n)\), while every tree admits an outerplanar drawing, that is, every tree is outer \(0\)-planar. Similarly, since a path \(P_n\) can be embedded directly on the outer face, parameters such as treedepth and vertex cover number are also unbounded.

\begin{figure}[ht!]
    \centering
    \includegraphics[width=0.55\linewidth]{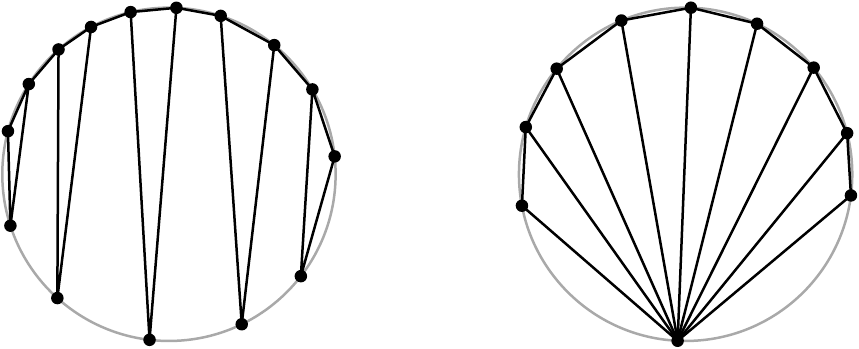}
    \caption{A sequence of triangles can be embedded as an outerplanar graph, which implies that the feedback vertex set number is unbounded (left). Arbitrarily large fan graphs \(F_n\) can also be embedded as outerplanar graphs, leading to several other parameters being unbounded (right).}
    \label{fig:fvs-fan}
\end{figure}

For the feedback vertex set number, consider the construction in Figure~\ref{fig:fvs-fan} consisting of a sequence of pairwise disjoint triangles. Since at least one vertex from each triangle must be selected in order to destroy all cycles, the feedback vertex set number of an outer \(k\)-planar graph can be as large as \(n/3\), and therefore is not bounded either. Table~\ref{tab:unbounded-parameters} summarizes more unbounded parameters together with their corresponding counterexamples. All of these constructions admit outer \(k\)-planar drawings, while the associated parameters remain unbounded.

\begin{table}[ht!]
\centering
\renewcommand{\arraystretch}{1.2}
\begin{tabular}{|p{0.29\textwidth}|p{0.65\textwidth}|}
\hline
Graph Parameter & Counterexample \\
\hline
tree-depth~\cite{Sparsity-book} & \(P_n\) \\
\hline
vertex cover number & \(P_n\) \\
\hline
feedback vertex set number & Figure~\ref{fig:fvs-fan} \\
\hline
tree-bandwidth~\cite{treebw} & \(F_n\) \\
\hline
pathwidth & complete binary tree \\
\hline
tree-partition-width~\cite{tpw-1,tpw-2} & \(F_n\); see e.g.~\cite{Distel2024} \\
\hline
cut-width~\cite{cutwidth-1, cutwidth-2} & star graph \(S_n\), \(F_n\), or via unbounded pathwidth \\ 
\hline
bandwidth~\cite{Saxe-bw,xnlp1} & complete \(k\)-ary tree~\cite{bw-k-ary}, or via unbounded cut-width \\
\hline
tree-cut width~\cite{tree-cutwidth} & \(F_n\), or via unbounded tree-partition-width~\cite[Theorem 5.6]{compute-tpw} \\
\hline
edge-cut width~\cite{edge-cut-width} &  \(F_n\), or via unbounded tree-cut width; see~\cite[Proposition 1]{edge-cut-width} \\
\hline
\(\alpha\)-edge-crossing width~\cite{CHANG2026492} & \(F_n\), or via unbounded tree-partition-width; see~\cite[Section 5]{CHANG2026492} \\
\hline
feedback edge set number & via unbounded feedback vertex set or edge-cut width~\cite{edge-cut-width} \\
\hline
distance to outerplanar & Section~\ref{hierarchy} and Figure~\ref{fig:distance-to-outerplanar} \\
\hline
\end{tabular}
\caption{Some unbounded graph parameters on outer \(k\)-planar graphs. The corresponding counterexamples or reasons showing unboundedness are indicated. Definitions of these parameters can be found in the cited references. Many results follow directly from the parameter hierarchy.}
\label{tab:unbounded-parameters}
\end{table}

Among these unbounded parameters, many are in fact incomparable with the outer \(k\)-planarity. For example, neither bounded vertex cover number nor bounded feedback vertex set number implies bounded outer \(k\)-planarity. This means even if a graph has vertex cover number or feedback vertex set number at most \(k\), it does not necessarily admit an outer \(f(k)\)-planar drawing for any function \(f\). A counterexample is the graph \(K_{2,n}\). Its vertex cover number is only \(2\), and its feedback vertex set number is only \(1\). However, as discussed earlier, in every circular drawing of \(K_{2,n}\), some edge must be crossed \(\Theta(n)\) times. The same conclusion can be extended to several other parameters, such as pathwidth, treedepth, tree-partition width, tree-cut width, and tree-bandwidth. One can easily verify that \(K_{2,n}\) has pathwidth \(2\), treedepth \(3\), and tree-partition width, tree-cut width, and tree-bandwidth all at most \(2\).

\subsection{Cut-Width and Cyclomatic Number}
To determine more precisely the position of outer \(k\)-planarity within the parameter hierarchy, we also need to identify its parent parameters in the hierarchy. In other words, we need to determine which graph parameters satisfy the property that every graph with parameter value at most \(k\) admits an outer \(f(k)\)-planar drawing. We prove that cut-width, bandwidth, cyclomatic number (feedback edge set number), and all parameters above them, i.e., their ancestors, in the parameter hierarchy satisfy this property. Moreover, among these parameters, cut-width and cyclomatic number are two particularly natural choices of direct parent parameters that are structurally close to outer \(k\)-planarity within the hierarchy.

\begin{theorem}\label{cut-width-2k}
If \(\operatorname{cutw}(G)\le k\), then \(G\) admits an outer \(2k\)-planar drawing.
\end{theorem}

\begin{proof}
Let \(\pi=(v_1,\ldots,v_n)\) be a linear ordering of \(V(G)\) of cut-width at most \(k\). We place the vertices of \(G\) on the outer face in the cyclic order \(v_1,\ldots,v_n\), and draw every edge as the straight chord between its two endpoints.

We show that in this drawing every edge is crossed by at most \(2k\) other edges. Fix an edge \(e=v_i v_j\), where \(i<j\). An edge \(f\) crosses \(e\) in this circular straight-line drawing if and only if the endpoints of \(f\) alternate with the endpoints of \(e\) in the circular order. Equivalently, exactly one endpoint of \(f\) lies in the interval \(\{v_{i+1},\ldots,v_{j-1}\}\), and the other endpoint lies outside the interval \(\{v_i,\ldots,v_j\}\). Edges sharing an endpoint with \(e\) are not counted as crossings.

We split the edges crossing \(e\) into two classes. First consider the edges with one endpoint in \(\{v_{i+1},\ldots,v_{j-1}\}\) and the other endpoint in \(\{v_1,\ldots,v_{i-1}\}\). Every such edge crosses the cut between \(\{v_1,\ldots,v_i\}\) and \(\{v_{i+1},\ldots,v_n\}\). Since the cut-width of \(\pi\) is at most \(k\), there are at most \(k\) such edges. Next consider the edges with one endpoint in \(\{v_{i+1},\ldots,v_{j-1}\}\) and the other endpoint in \(\{v_{j+1},\ldots,v_n\}\). Every such edge crosses the cut between \(\{v_1,\ldots,v_{j-1}\}\) and \(\{v_j,\ldots,v_n\}\). Again, since the cut-width of \(\pi\) is at most \(k\), there are at most \(k\) such edges.

These two classes contain all edges crossing \(e\). Hence \(e\) is crossed by at most \(k+k=2k\) edges. Since \(e\) is arbitrary, the drawing is outer \(2k\)-planar.
\end{proof}

\begin{corollary}\label{bw-2k2}
If \(\operatorname{bw}(G)\le k\), then \(G\) admits an outer \(2k^2\)-planar drawing.
\end{corollary}

\begin{proof}
Let \(\pi=(v_1,\ldots,v_n)\) be a linear ordering of \(V(G)\) of bandwidth at most \(k\). We first show that the cut-width of this ordering is at most \(k^2\).

Fix a cut between \(\{v_1,\ldots,v_i\}\) and \(\{v_{i+1},\ldots,v_n\}\). Let \(v_pv_q\) be an edge crossing this cut, with \(p\le i<q\). Since the bandwidth of \(\pi\) is at most \(k\), we have \(q-p\le k\). Therefore \(p\ge i-k+1\), so the left endpoint \(v_p\) belongs to the set \(\{v_{i-k+1},\ldots,v_i\}\), which has at most \(k\) vertices. Similarly, \(q\le i+k\), so the right endpoint \(v_q\) belongs to the set \(\{v_{i+1},\ldots,v_{i+k}\}\), which also has at most \(k\) vertices. Thus every edge crossing this cut has one endpoint in a set of at most \(k\) vertices on the left side and one endpoint in a set of at most \(k\) vertices on the right side. Since \(G\) is simple, there are at most \(k^2\) such edges.

As this holds for every cut of \(\pi\), the cut-width of \(\pi\) is at most \(k^2\). Hence \(\operatorname{cutw}(G)\le k^2\). By Theorem~\ref{cut-width-2k} above, \(G\) admits an outer \(2k^2\)-planar drawing.
\end{proof}

\begin{theorem}\label{fes-6k}
Let \(G\) be a graph whose feedback edge set number is at most \(k\). Then \(G\) admits an outer \(6k\)-planar drawing.
\end{theorem}

\begin{proof}
We first prove the statement for connected graphs. Let \(G\) be connected and let the feedback edge set number be \(r:=|E(G)|-|V(G)|+1\). Thus \(r\le k\). If \(r=0\), then \(G\) is a tree, and hence has an outerplanar drawing. If \(r=1\), then \(G\) is a unicyclic graph with trees attached to the unique cycle. Such a graph is outerplanar as well: place the vertices of the cycle in cyclic order on the outer face, and draw each attached tree in a small interval near its attachment vertex. Thus the claim is immediate for \(r\le 1\). In the rest of the proof we assume \(r\ge 2\).

Let \(C\) be the \(2\)-core of \(G\), obtained by repeatedly deleting vertices of degree \(0\) or \(1\). The vertices and edges removed in this process form trees attached to vertices of \(C\). These trees do not contribute to the cyclomatic number. We can first draw the core and then insert all attached trees locally. Now suppress all maximal paths of \(C\) whose internal vertices have degree \(2\) in \(C\). This produces a multigraph \(H\). The \emph{skeleton} graph \(H\) may have parallel edges and loops. A loop in \(H\) corresponds to a cycle in \(C\) that starts and ends at the same branch vertex. We use the standard  convention that a loop contributes \(2\) to the degree of its incident vertex.

Suppressing degree-\(2\) paths does not change the cyclomatic number. Therefore, if \(n_H:=|V(H)|\) and \(m_H:=|E(H)|\), then \(m_H-n_H+1=r\). Since \(r\ge 2\), the core is not just a single cycle. Hence every vertex of \(H\) has degree at least \(3\), where loops count twice. By the handshaking identity for multigraphs, \(2m_H=\sum_{v\in V(H)}\deg_H(v)\ge 3n_H\). Combining this with \(m_H=n_H+r-1\), we obtain \(2(n_H+r-1)\ge 3n_H\), and hence \(n_H\le 2r-2\). It follows that \(m_H=n_H+r-1\le 3r-3\). Thus the graph \(H\) has at most \(3r-3\) edges.

Each edge \(e\) of \(H\) corresponds to a path in the core \(C\). If \(e\) has distinct endpoints \(a\) and \(b\), then this path has the form \(a,x_1,\ldots,x_q,b\), where the internal vertices \(x_1,\ldots,x_q\) have degree \(2\) in \(C\). If \(q=0\), the skeleton edge corresponds to the original edge \(ab\). If \(e\) is a loop at a vertex \(a\), then the corresponding path has the form \(a,x_1,\ldots,x_q,a\), with \(q\ge 1\).

We now define a circular order of the vertices of \(C\). First place the vertices of \(H\) on the outer face in an arbitrary circular order. For each skeleton edge \(e\), choose one endpoint of \(e\) as its anchor. If \(e\) is a loop, its unique endpoint is chosen as the anchor. For every vertex \(a\in V(H)\), reserve a small interval immediately after \(a\) in the circular order. This interval is divided into smaller consecutive blocks, one for each skeleton edge anchored at \(a\). If an edge \(e\) anchored at \(a\) corresponds to the path \(a,x_1,\ldots,x_q,b\), we put \(x_1,\ldots,x_q\) consecutively into the block of \(e\), in this order. If \(e\) is a loop at \(a\), we similarly put the vertices \(x_1,\ldots,x_q\) of the corresponding closed path \(a,x_1,\ldots,x_q,a\) consecutively into the block of \(e\). See Figure~\ref{fig:fes-2} for an illustration of how the skeleton edges in \(H\) are expanded back into the degree-2 paths of the \(2\)-core graph \(C\).

\begin{figure}[ht!]
    \centering
    \includegraphics[width=0.7\linewidth]{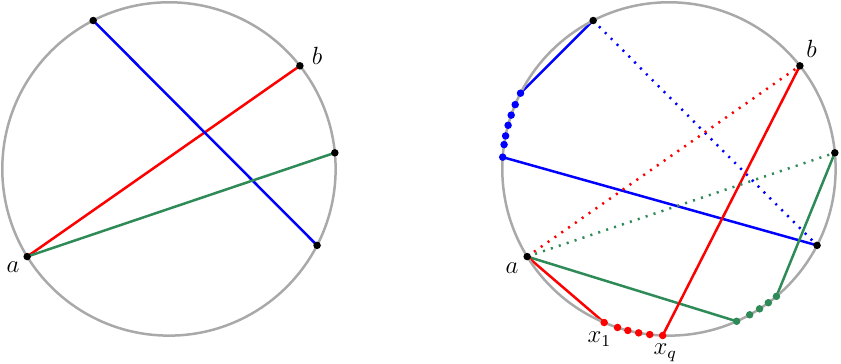}
    \caption{An example illustrating how three skeleton edges in \(H\) (left) are expanded back into paths in the graph \(C\) (right). The vertices of each degree-2 path are placed consecutively along an interval of the outer face, so that the internal edges \(x_i x_{i+1}\) of the path do not cross any other edge. In \(H\), at most \(3r-3\) crossing-relevant edges participate in crossings, while after expanding the paths in \(C\), at most \(6r-6\) crossing-relevant edges participate in crossings. The dotted edges are included only to indicate the positions of the edges in \(H\); they are not edges of \(C\).}
    \label{fig:fes-2}
\end{figure}

We call an edge of \(C\) \emph{crossing-relevant} if it is one of the first or last edges of such an expanded skeleton path. Thus, for a non-loop path \(a,x_1,\ldots,x_q,b\), the crossing-relevant edges are \(a x_1\) and \(x_q b\) if \(q\ge 1\), and just \(ab\) if \(q=0\). For a loop path \(a,x_1,\ldots,x_q,a\), the crossing-relevant edges are \(a x_1\) and \(x_q a\). Every skeleton edge contributes at most two crossing-relevant edges. Hence the total number of crossing-relevant edges is at most \(2m_H\le 6r-6\). All other edges of \(C\) lie inside one of the consecutive blocks. Indeed, these edges are of the form \(x_i x_{i+1}\) along an expanded skeleton path. Since the vertices \(x_1,\ldots,x_q\) were placed consecutively in the block of that path, such edges are local inside a single block. Therefore they do not alternate, in the circular order, with any edge whose endpoints are outside that block. They also do not cross edges inside the same block, because they are consecutive along a path. Hence the only edges of the core that may participate in crossings are the crossing-relevant edges. It follows that every edge of \(C\) is crossed by at most \(2m_H-1\le 6r-7\) other edges of \(C\). Non-crossing-relevant edges are crossed by no edge at all, while a crossing-relevant edge can cross only other crossing-relevant edges.

\begin{figure}[ht!]
    \centering
    \includegraphics[width=0.75\linewidth]{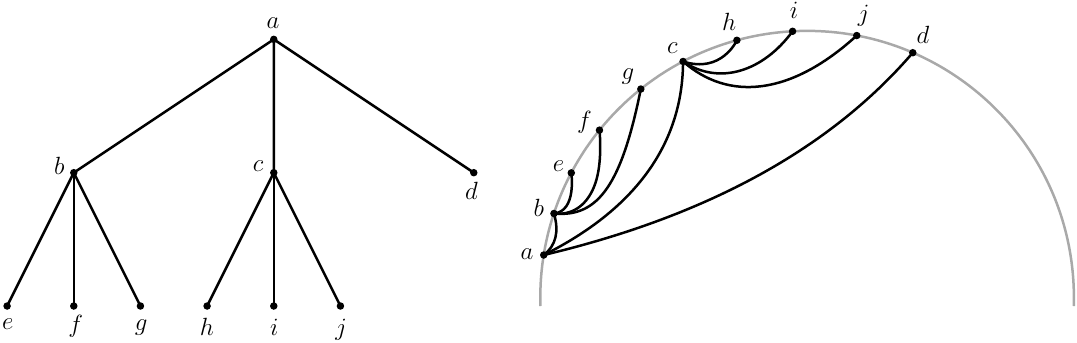}
    \caption{Every tree attached to a vertex \(a \in V(C)\) can have its vertices placed consecutively in an interval of the outer face adjacent to \(a\) according to a pre-order traversal. The edges of these trees do not cross any other edges of the graph. Hence, after reconstructing \(G\) from the 2-core graph \(C\), every edge still has at most \(6r-7\) crossings.}
    \label{fig:fes-tree}
\end{figure}

It remains to insert the trees that were deleted when passing from \(G\) to the \(2\)-core. Each such tree is attached to a unique vertex \(v\) of the core \(C\). Around every core vertex \(v \in V(C)\), we reserve additional small intervals next to \(v\). For each tree attached at \(v\), root the tree at \(v\) and place its vertices in one of these intervals according to a standard depth-first, or pre-order, traversal of the rooted tree. This gives a non-crossing drawing of the attached tree inside that interval. Since all vertices of the attached tree lie in a consecutive interval adjacent to \(v\), every edge of the tree has both endpoints in that interval, except for edges incident with the root \(v\). Thus no edge of an attached tree alternates with an edge whose endpoints lie outside the interval; see Figure~\ref{fig:fes-tree} for an example. Consequently, inserting all attached trees creates no new crossings. We have therefore constructed a circular ordering of all vertices of the connected graph \(G\) such that every edge is crossed at most \(6r-7\) times if \(r\ge 2\), and no edge is crossed if \(r\le 1\). Thus, every connected graph with feedback edge set number \(r\) has an outer \(6r\)-planar drawing.

Now let \(G\) be disconnected, with connected components \(G_1,\ldots,G_s\). Let \(r_i\) be the feedback edge set number of \(G_i\). Then \(\sum_{i=1}^s r_i=r\le k\). Apply the connected construction to each component \(G_i\). Place the circular orders of the components in pairwise disjoint consecutive intervals on the common outer face. Since there are no edges between distinct components, no crossings are created between different components. Every edge of \(G_i\) is crossed at most \(6r_i\le 6k\) times. Hence the whole graph \(G\) admits an outer \(6k\)-planar drawing.
\end{proof}

\subsection{Outer \(k\)-Planarity in the Hierarchy}\label{hierarchy}
Another parameter that appears to be related to outer \(k\)-planarity is distance to outerplanar. The distance to outerplanar of a graph \(G=(V,E)\) is the minimum size of a vertex subset \(X\subseteq V\) such that \(G[V\setminus X]\) is outerplanar. Note that graphs with distance to outerplanar at most \(k\) have treewidth at most \(k+2\), since outerplanar graphs have treewidth at most 2~\cite{BOD1998}. Moreover, bounded feedback vertex set number implies bounded distance to outerplanar, since every forest is outerplanar. We show that distance to outerplanar and outer \(k\)-planarity are incomparable.

\begin{figure}[ht!]
    \centering
    \includegraphics[width=0.55\linewidth]{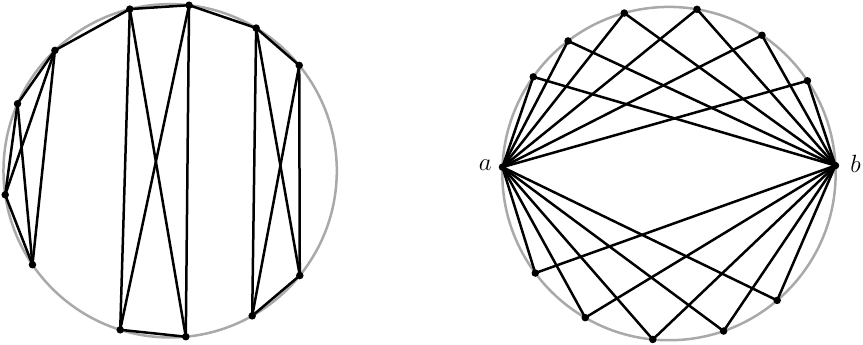}
    \caption{A sequence of connected disjoint copies of \(K_4\) is outer \(1\)-planar, while its distance to outerplanar can be \(n/4\) (left). \(K_{2,n}\) has distance to outerplanar equal to \(1\), since deleting either vertex \(a\) or \(b\) yields a star graph, but \(K_{2,n}\) is not outer \(f(k)\)-planar for any function \(f\) (right).}
    \label{fig:distance-to-outerplanar}
\end{figure}

First, outer \(k\)-planarity does not imply that the distance to outerplanar is bounded by a function of \(k\). This already fails for \(k=1\). Consider a graph \(G_t\) obtained from \(t\) vertex-disjoint copies of \(K_4\) by connecting consecutive copies with bridge edges; see Figure~\ref{fig:distance-to-outerplanar} for an example. Each copy of \(K_4\) has an outer \(1\)-planar drawing: place its four vertices on the outer face and draw the two diagonals so that they cross once. Placing the \(t\) copies in pairwise disjoint consecutive intervals on the outer face and drawing the bridge edges without crossings gives an outer \(1\)-planar drawing of \(G_t\). Hence \(G_t\) is outer \(1\)-planar for every \(t\).

On the other hand, the distance to outerplanar of \(G_t\) grows with \(t\). Since outerplanar graphs are closed under taking subgraphs and \(K_4\) is not outerplanar, any vertex set whose deletion makes \(G_t\) outerplanar must contain at least one vertex from each \(K_4\) copy. Therefore at least \(t\) vertices must be deleted. Conversely, deleting one vertex from each \(K_4\) copy turns every copy into a triangle, and the resulting graph is a tree of outerplanar blocks connected by bridges, hence outerplanar. Thus the distance to outerplanar of \(G_t\) is exactly \(t\). Since \(t\) is arbitrary, there is no function \(f\) such that every outer \(k\)-planar graph has distance to outerplanar at most \(f(k)\).

The converse direction also fails. A graph of bounded distance to outerplanar need not admit an outer \(f(k)\)-planar drawing for any function \(f\) depending only on the distance. An example is \(K_{2,n}\) (see Figure~\ref{fig:distance-to-outerplanar}). Let the part of size two be \(\{a,b\}\) and the other part be \(\{x_1,\ldots,x_n\}\). Deleting \(a\) leaves a star centered at \(b\), which is outerplanar. Hence the distance to outerplanar of \(K_{2,n}\) is at most \(1\). However, \(K_{2,n}\) does not have outer \(f(1)\)-planar drawings, since as \(n\) grows, every circular drawing of \(K_{2,n}\) contains an edge crossed by \(\Theta(n)\) other edges. Therefore bounded distance to outerplanar does not imply bounded outer \(k\)-planarity.

Based on all boundedness and unboundedness results established in this section, a natural position of outer \(k\)-planarity within the graph parameter hierarchy is as a parent of treewidth, with an arc pointing to treewidth directly, and as a child of cut-width and cyclomatic number. Figure~\ref{fig:para-hierarchy} illustrates these boundedness relationships among the corresponding graph parameters.

\begin{figure}[ht!]
    \centering
    \includegraphics[width=0.95\linewidth]{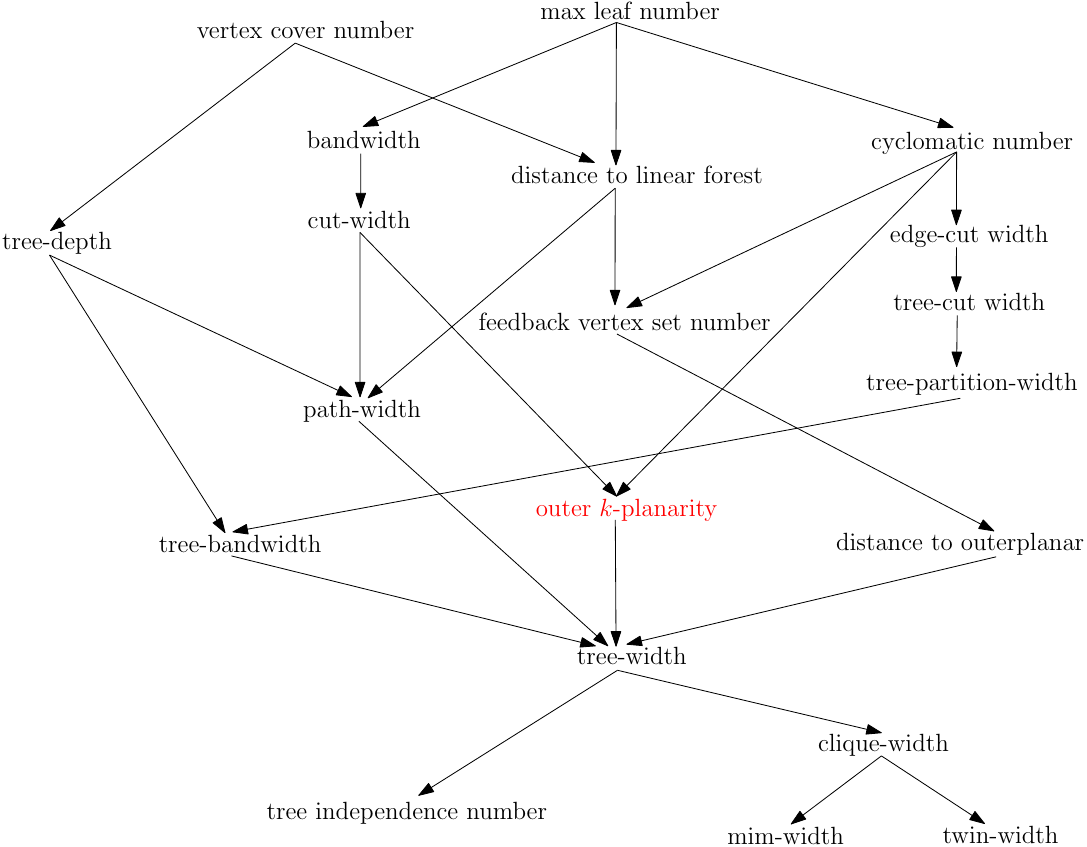}
    \caption{The position of outer \(k\)-planarity within the graph parameter hierarchy, as implied by the currently known boundedness and unboundedness results. An arc from parameter \(a\) to parameter \(b\) indicates that bounded \(a\) implies bounded \(b\). Only a selection of commonly studied and related parameters, as well as their relationships are shown.}
    \label{fig:para-hierarchy}
\end{figure}

\section{Conclusion and Open Problems}\label{sec:conclusion}

We initiated a systematic study of the parameterized complexity of classical graph problems on outer \(k\)-planar graphs, with the crossing bound \(k\), equivalently the convex local crossing number \(\operatorname{lcr}^{\circ}\), as the parameter. Since outer \(k\)-planar graphs have treewidth bounded by a linear function of \(k\), one might initially expect the complexity of problems on this graph class to closely follow their complexity when parameterized by treewidth. Our results show that the situation is more subtle. Some problems, such as \textsc{Binary CSP} and \textsc{Scattered Set}, remain XALP-complete on outer \(k\)-planar graphs. On the other hand, a broad collection of problems that are XALP-complete for treewidth, such as \textsc{List Coloring}, \textsc{Capacitated (Red-Blue) Dominating Set}, \textsc{Capacitated Vertex Cover}, \textsc{Target Outdegree Orientation}, \(f\)-\textsc{Dominating Set}, and \textsc{Target Set Selection}, become FPT when an outer \(k\)-planar drawing is given.

The main algorithmic tool behind these positive results is the triangulation method of Firman et al.~\cite{firman_et_al}. Although this method was originally used to prove treewidth and separator bounds, we show that it provides a more powerful and general framework. After triangulating the polygon enclosed by the outer cycle, the weak dual gives a binary tree structure along which dynamic programming can be performed. The crucial property is that every triangulation link is crossed by at most \(k\) original graph edges. Hence, each subproblem in the dynamic program corresponds to a subpolygon that induces a separator of the original graph of size at most \(k+2\). Moreover, apart from the two endpoints of the corresponding triangulation link, every boundary vertex has at most \(k\) neighbors outside the subpolygon. This geometric restriction makes it possible, for many problems that are hard with respect to treewidth, to replace the \(n^{O(k)}\) size table arising in standard dynamic programming on tree decompositions by table of size \(f(k)n^{O(1)}\).

We also studied the position of outer \(k\)-planarity within the graph-parameter hierarchy. We proved that every outer \(k\)-planar graph has mim-width at most \(k+2\), improving over the bound obtained indirectly from treewidth. We also showed that bounded cut-width, bandwidth, and feedback edge set number implies bounded outer \(k\)-planarity. In the opposite direction, we showed that many classical graph parameters are incomparable with outer \(k\)-planarity.

Several directions remain open. First, one can study the parameterized complexity of the classical graph problems considered in this paper on other beyond-planar graph classes, again using the maximum number of crossings per edge as the parameter. A central question is whether the geometric structure of such graph classes can be exploited in a way similar to the triangulation framework used here, so that problems that are hard for classical parameters become tractable. Closely related to outer \(k\)-planar graphs are, for example, outer \(k\)-quasi-planar graphs (see~\cite{beyond-survery}), and outer min-\(k\)-planar graphs (see~\cite{min-k-1, min-k-2}). It would be interesting to explore how the geometric properties of these graph classes interact with the design of parameterized algorithms.

A second direction is to further investigate the relationship between outer \(k\)-planarity and other graph parameters. In this paper, we have seen that bounded outer \(k\)-planarity automatically implies that many other parameters are bounded by some function of \(k\). One may try to obtain explicit upper bounds for these parameters on outer \(k\)-planar graphs, or improve upper bounds that follow indirectly from known parameter hierarchy. For instance, it is known that bounded treewidth implies bounded clique-width~\cite{tw-cdw-1, tw-cdw-2}. Courcelle~\cite{cdw-crossing} studied clique-width for graphs with edge crossings, and it would be interesting to see whether these results can be used to obtain a non-trivial upper bound, for example a subexponential bound, on the clique-width of outer \(k\)-planar graphs. Another promising parameter is stretch-width, introduced by Bonnet and Duron~\cite{stretch-width} as a variant of twin-width for ordered graphs. Since an outer \(k\)-planar drawing comes with a natural circular ordering of the vertices, and since outer \(k\)-planar graphs do not contain large complete bipartite graphs as subgraphs, \cite[Lemma~7]{stretch-width} can be used to directly derive an \(O(k^3)\) upper bound on the stretch-width of outer \(k\)-planar graphs. Edge-cut width is a recently developed parameter based on edge cuts, introduced by Brand et al.~\cite{edge-cut-width}. In the parameter hierarchy studied here, bounded cyclomatic number implies bounded edge-cut width. It remains open whether bounded edge-cut width implies bounded outer \(k\)-planarity. If this is true, then in Figure~\ref{fig:para-hierarchy}, a more appropriate direct parent of outer \(k\)-planarity would be edge-cut width.

Finally, the complexity of recognizing outer \(k\)-planar graphs remains a major open problem. Kobayashi et al.~\cite{outerkrecog} gave an XP algorithm and proved that testing outer \(k\)-planarity is XNLP-hard. They also conjectured that the problem is XALP-complete. For proving XALP-hardness, one possible route is to reduce from \textsc{Tree-Partition Width}, whose computation is known to be XALP-complete~\cite{compute-tpw}. For proving membership in XALP, one may try to modify the XP algorithm of Kobayashi et al.~\cite{outerkrecog} into a nondeterministic algorithm that runs in FPT time and logarithmic space with an additional stack. Alternatively, one could start from another equivalent definition of XALP and attempt to simulate the XP algorithm directly by an alternating Turing machine; see~\cite{xalp} for several equivalent characterizations of the class XALP.

\bibliography{refs}

@book{DF13,
	author	= {Rodney G. Downey and Michael R. Fellows},
	publisher	= {Springer},
	series	= {Texts in Computer Science},
	title		= {Fundamentals of Parameterized Complexity},
	year		= {2013},
	isbn      = {978-1-4471-5558-4},
	doi       = {10.1007/978-1-4471-5559-1},
}

@book{DF99,
  title={Parameterized Complexity},
  author={Rodney G. Downey and Michael R. Fellows},
  isbn={9780387948836},
  lccn={97022882},
  series={Monographs in Computer Science},
  doi={10.1007/978-1-4612-0515-9},
  year={1999},
  publisher={Springer}
}

@book{flum06,
  title={Parameterized Complexity Theory},
  author={Jörg Flum and Martin Grohe},
  isbn={9783540299530},
  lccn={2005938662},
  series={Texts in Theoretical Computer Science. An EATCS Series},
  doi={10.1007/3-540-29953-X},
  year={2006},
  publisher={Springer}
}

@article{wh2,
title = {Fixed-parameter tractability and completeness {II}: On completeness for {W[1]}},
journal = {Theoretical Computer Science},
volume = {141},
number = {1},
pages = {109-131},
year = {1995},
issn = {0304-3975},
doi = {https://doi.org/10.1016/0304-3975(94)00097-3},
author = {Rod G. Downey and Michael R. Fellows},
}

@article{xnlp1,
author       = {Hans L. Bodlaender and
                  Carla Groenland and
                  Jesper Nederlof and
                  C{\'{e}}line M. F. Swennenhuis},
  title        = {Parameterized problems complete for nondeterministic {FPT} time and
                  logarithmic space},
  journal      = {Information and Compututation},
  volume       = {300},
  pages        = {105195},
  year         = {2024},
  nourl          = {https://doi.org/10.1016/j.ic.2024.105195},
  doi          = {10.1016/J.IC.2024.105195},
  bibsource    = {dblp computer science bibliography, https://dblp.org}
}

@InProceedings{xalp,
  author =	{Bodlaender, Hans L. and Groenland, Carla and Jacob, Hugo and Pilipczuk, Marcin and Pilipczuk, Micha{\l}},
  title =	{On the Complexity of Problems on Tree-Structured Graphs},
  booktitle =	{17th International Symposium on Parameterized and Exact Computation (IPEC 2022)},
  pages =	{6:1--6:17},
  series =	{Leibniz International Proceedings in Informatics (LIPIcs)},
  ISBN =	{978-3-95977-260-0},
  ISSN =	{1868-8969},
  year =	{2022},
  volume =	{249},
  publisher =	{Schloss Dagstuhl -- Leibniz-Zentrum f{\"u}r Informatik},
  noaddress =	{Dagstuhl, Germany},
  URN =		{urn:nbn:de:0030-drops-173626},
  doi =		{10.4230/LIPIcs.IPEC.2022.6}
}

@article{outerkplanar,
title = {{XALP}-completeness of parameterized problems on planar graphs},
journal = {Discrete Applied Mathematics},
volume = {386},
pages = {156-174},
year = {2026},
issn = {0166-218X},
doi = {https://doi.org/10.1016/j.dam.2026.01.021},
author = {Hans L. Bodlaender and Krisztina Szilágyi},
}

@book{paralg-book,
  author       = {Marek Cygan and
                  Fedor V. Fomin and
                  Lukasz Kowalik and
                  Daniel Lokshtanov and
                  D{\'{a}}niel Marx and
                  Marcin Pilipczuk and
                  Michal Pilipczuk and
                  Saket Saurabh},
  title        = {Parameterized Algorithms},
  publisher    = {Springer},
  year         = {2015},
  doi          = {10.1007/978-3-319-21275-3},
  isbn         = {978-3-319-21274-6}
}

@InProceedings{outerkrecog,
  author =	{Kobayashi, Yasuaki and Okada, Yuto and Wolff, Alexander},
  title =	{Recognizing 2-Layer and Outer $k$-Planar Graphs},
  booktitle =	{41st International Symposium on Computational Geometry (SoCG 2025)},
  pages =	{65:1--65:16},
  series =	{Leibniz International Proceedings in Informatics (LIPIcs)},
  ISBN =	{978-3-95977-370-6},
  ISSN =	{1868-8969},
  year =	{2025},
  volume =	{332},
  publisher =	{Schloss Dagstuhl -- Leibniz-Zentrum f{\"u}r Informatik},
  noaddress =	{Dagstuhl, Germany},
  doi =		{10.4230/LIPIcs.SoCG.2025.65}
}

@InProceedings{beyond-survery,
author="Chaplick, Steven
and Kryven, Myroslav
and Liotta, Giuseppe
and L{\"o}ffler, Andre
and Wolff, Alexander",
title="Beyond Outerplanarity",
booktitle="25th International Symposium on Graph Drawing and Network Visualization (GD 2018)",
year="2018",
series = {Lecture Notes in Computer Science},
volume = 10692,
publisher="Springer",
noaddress="Cham",
pages="546--559",
isbn="978-3-319-73915-1",
doi="10.1007/978-3-319-73915-1_42"
}

@InProceedings{pyzik,
  author =	{Pyzik, Rafa{\l}},
  title =	{Treewidth of Outer $k$-Planar Graphs},
  booktitle =	{33rd International Symposium on Graph Drawing and Network Visualization (GD 2025)},
  pages =	{28:1--28:16},
  series =	{Leibniz International Proceedings in Informatics (LIPIcs)},
  ISBN =	{978-3-95977-403-1},
  ISSN =	{1868-8969},
  year =	{2025},
  volume =	{357},
  publisher =	{Schloss Dagstuhl -- Leibniz-Zentrum f{\"u}r Informatik},
  noaddress =	{Dagstuhl, Germany},
  doi =		{10.4230/LIPIcs.GD.2025.28}
}

@InProceedings{firman_et_al,
  author =	{Firman, Oksana and Gutowski, Grzegorz and Kryven, Myroslav and Okada, Yuto and Wolff, Alexander},
  title =	{Bounding the Treewidth of Outer $k$-Planar Graphs via Triangulations},
  booktitle =	{32nd International Symposium on Graph Drawing and Network Visualization (GD 2024)},
  pages =	{14:1--14:17},
  series =	{Leibniz International Proceedings in Informatics (LIPIcs)},
  ISBN =	{978-3-95977-343-0},
  ISSN =	{1868-8969},
  year =	{2024},
  volume =	{320},
  publisher =	{Schloss Dagstuhl -- Leibniz-Zentrum f{\"u}r Informatik},
  noaddress =	{Dagstuhl, Germany},
  doi =		{10.4230/LIPIcs.GD.2024.14},
}

@article{outer-1,
author = {Auer, Christopher and Bachmaier, Christian and Brandenburg, Franz J. and Glei\ss{}ner, Andreas and Hanauer, Kathrin and Neuwirth, Daniel and Reislhuber, Josef},
title = {Outer 1-Planar Graphs},
year = {2016},
issue_date = {Apr 2016},
publisher = {Springer},
noaddress = {Berlin, Heidelberg},
volume = {74},
number = {4},
issn = {0178-4617},
url = {https://doi.org/10.1007/s00453-015-0002-1},
doi = {10.1007/s00453-015-0002-1},
journal = {Algorithmica},
nomonth = apr,
pages = {1293–1320},
numpages = {28}
}

@InProceedings{xnlp-flows,
author="Bodlaender, Hans L.
and Cornelissen, Gunther
and van der Wegen, Marieke",
title="Problems Hard for Treewidth but Easy for Stable Gonality",
booktitle="48th International Workshop on Graph-Theoretic Concepts  in Computer Science (WG 2022)",
year="2022",
publisher="Springer",
series = {Lecture Notes in Computer Science},
volume = 13453,
noaddress="Cham",
pages="84--97",
isbn="978-3-031-15914-5",
doi="10.1007/978-3-031-15914-5_7"
}

@article{some-colorful,
title = {On the complexity of some colorful problems parameterized by treewidth},
journal = {Information and Computation},
volume = {209},
number = {2},
pages = {143-153},
year = {2011},
issn = {0890-5401},
doi = {10.1016/j.ic.2010.11.026},
author = {Michael R. Fellows and Fedor V. Fomin and Daniel Lokshtanov and Frances Rosamond and Saket Saurabh and Stefan Szeider and Carsten Thomassen},
}

@article{spsc,
author = {Pilipczuk, Micha\l{} and Wrochna, Marcin},
title = {On Space Efficiency of Algorithms Working on Structural Decompositions of Graphs},
year = {2018},
issue_date = {December 2017},
publisher = {Association for Computing Machinery},
address = {New York, NY, USA},
volume = {9},
number = {4},
issn = {1942-3454},
doi = {10.1145/3154856},
journal = {ACM Transactions on Computation Theory},
nomonth = jan,
articleno = {18},
numpages = {36},
}

@article{local-crossing-3,
author = {Didimo, Walter and Liotta, Giuseppe and Montecchiani, Fabrizio},
title = {A Survey on Graph Drawing Beyond Planarity},
year = {2019},
issue_date = {January 2020},
publisher = {Association for Computing Machinery},
noaddress = {New York, NY, USA},
volume = {52},
number = {1},
issn = {0360-0300},
doi = {10.1145/3301281},
journal = {ACM Computing Surveys},
nomonth = feb,
articleno = {4},
numpages = {37},
}

@Article{local-crossing-4,
  author =	{Dujmovi\'{c}, Vida and Hong, Seok-Hee and Kaufmann, Michael and Pach, J\'{a}nos and F\"{o}rster, Henry},
  title =	{{Beyond-Planar Graphs: Models, Structures and Geometric Representations (Dagstuhl Seminar 24062)}},
  pages =	{71--94},
  journal =	{Dagstuhl Reports},
  ISSN =	{2192-5283},
  year =	{2024},
  volume =	{14},
  number =	{2},
  publisher =	{Schloss Dagstuhl -- Leibniz-Zentrum f{\"u}r Informatik},
  noaddress =	{Dagstuhl, Germany},
  doi =		{10.4230/DagRep.14.2.71},
}

@Article{1-planar,
author={Grigoriev, Alexander
and Bodlaender, Hans L.},
title={Algorithms for Graphs Embeddable with Few Crossings per Edge},
journal={Algorithmica},
year={2007},
nomonth={Sep},
day={01},
volume={49},
number={1},
pages={1-11},
issn={1432-0541},
doi={10.1007/s00453-007-0010-x},
url={https://doi.org/10.1007/s00453-007-0010-x}
}

@article{recog-k-planar,
title = {Testing gap $k$-planarity is {NP}-complete},
journal = {Information Processing Letters},
volume = {169},
pages = {106083},
year = {2021},
issn = {0020-0190},
doi = {https://doi.org/10.1016/j.ipl.2020.106083},
author = {John C. Urschel and Jake Wellens},
}

@article{cyclomatic, 
title={Parameterized Complexity of 1-Planarity}, 
volume={22},
DOI={10.7155/jgaa.00457},
journal={Journal of Graph Algorithms and Applications},
author={Bannister, Michael and Cabello, Sergio and Eppstein, David}, 
year={2018}, nomonth={Jan.}, pages={23–49} }

@InProceedings{book-2,
  author =	{Agrawal, Akanksha and Cabello, Sergio and Kaufmann, Michael and Saurabh, Saket and Sharma, Roohani and Uno, Yushi and Wolff, Alexander},
  title =	{Eliminating Crossings in Ordered Graphs},
  booktitle =	{19th Scandinavian Symposium and Workshops on Algorithm Theory (SWAT 2024)},
  pages =	{1:1--1:19},
  series =	{Leibniz International Proceedings in Informatics (LIPIcs)},
  ISBN =	{978-3-95977-318-8},
  ISSN =	{1868-8969},
  year =	{2024},
  volume =	{294},
  publisher =	{Schloss Dagstuhl -- Leibniz-Zentrum f{\"u}r Informatik},
  noaddress =	{Dagstuhl, Germany},
  doi =		{10.4230/LIPIcs.SWAT.2024.1},
}

@Article{pach,
author={Pach, J{\'a}nos
and T{\'o}th, G{\'e}za},
title={Graphs drawn with few crossings per edge},
journal={Combinatorica},
year={1997},
nomonth={Sep},
day={01},
volume={17},
number={3},
pages={427-439},
issn={1439-6912},
doi={10.1007/BF01215922},
}

@InProceedings{outer-k-edge,
  author =	{Aichholzer, Oswin and Obenaus, Johannes and Orthaber, Joachim and Paul, Rosna and Schnider, Patrick and Steiner, Raphael and Taubner, Tim and Vogtenhuber, Birgit},
  title =	{Edge Partitions of Complete Geometric Graphs},
  booktitle =	{38th International Symposium on Computational Geometry (SoCG 2022)}, 
  pages =	{6:1--6:16},
  series =	{Leibniz International Proceedings in Informatics (LIPIcs)},
  ISBN =	{978-3-95977-227-3},
  ISSN =	{1868-8969},
  year =	{2022},
  volume =	{224},
  publisher =	{Schloss Dagstuhl -- Leibniz-Zentrum f{\"u}r Informatik},
  noaddress =	{Dagstuhl, Germany},
  doi =		{10.4230/LIPIcs.SoCG.2022.6},
}

@article{Wood2007,
author = {Wood, David R. and Telle, Jan Arne},
journal = {The New York Journal of Mathematics},
language = {eng},
pages = {117-146},
publisher = {University at Albany, Deptartment of Mathematics and Statistics},
title = {Planar decompositions and the crossing number of graphs with an excluded minor.},
url = {https://nyjm.albany.edu/j/2007/13-8.html},
volume = {13},
year = {2007},
}

@book{Kloks,
  title={Treewidth, Computations and Approximations},
  author={Ton Kloks},
  series={Lecture Notes in Computer Science},
  year={1994},
  doi={https://doi.org/10.1007/BFb0045375},
  volume = {842},
  publisher={Springer},
}

@Article{Eto2014,
author={Eto, Hiroshi
and Guo, Fengrui
and Miyano, Eiji},
title={Distance-$d$ independent set problems for bipartite and chordal graphs},
journal={Journal of Combinatorial Optimization}, 
year={2014},
nomonth={Jan},
day={01},
volume={27},
number={1},
pages={88-99},
issn={1573-2886},
doi={10.1007/s10878-012-9594-4},
}

@article{scattered-set-fpt,
title = {Structurally parameterized $d$-scattered set},
journal = {Discrete Applied Mathematics},
volume = {308},
pages = {168-186},
year = {2022},
nonote = {Combinatorial Optimization ISCO 2018},
issn = {0166-218X},
doi = {https://doi.org/10.1016/j.dam.2020.03.052},
author = {Ioannis Katsikarelis and Michael Lampis and Vangelis Th. Paschos},
}

@article{schaefer,
  title={The Graph Crossing Number and its Variants: A Survey},
  author={Marcus Schaefer},
  journal={Electronic Journal of Combinatorics},
  year={2013},
  number={DS21},
  doi={https://doi.org/10.37236/2713}
}

@article{CHANG2026492,
title = {A new width parameter of graphs based on edge cuts: $\alpha$-edge-crossing width},
journal = {Discrete Applied Mathematics},
volume = {380},
pages = {492-510},
year = {2026},
issn = {0166-218X},
doi = {10.1016/j.dam.2025.10.056},
author = {Yeonsu Chang and {O-joung} Kwon and Myounghwan Lee},
}

@Article{Ringel1965,
author={Ringel, Gerhard},
title={Ein {S}echsfarbenproblem auf der {K}ugel},
journal={Abhandlungen aus dem Mathematischen Seminar der Universit{\"a}t Hamburg},
year={1965},
nomonth={Dec},
day={01},
volume={29},
number={1},
pages={107-117},
issn={1865-8784},
doi={10.1007/BF02996313},
}

@article{BOD1998,
title = {A partial $k$-arboretum of graphs with bounded treewidth},
journal = {Theoretical Computer Science},
volume = {209},
number = {1},
pages = {1-45},
year = {1998},
issn = {0304-3975},
doi = {https://doi.org/10.1016/S0304-3975(97)00228-4},
author = {Hans L. Bodlaender},
}

@article{JANSEN1997135,
title = {Generalized coloring for tree-like graphs},
journal = {Discrete Applied Mathematics},
volume = {75},
number = {2},
pages = {135-155},
year = {1997},
issn = {0166-218X},
doi = {https://doi.org/10.1016/S0166-218X(96)00085-6},
author = {Klaus Jansen and Petra Scheffler},
}

@InProceedings{lc-tree,
  author =	{Bodlaender, Hans L. and Groenland, Carla and Jacob, Hugo},
  title =	{List Colouring Trees in Logarithmic Space},
  booktitle =	{30th Annual European Symposium on Algorithms (ESA 2022)},
  pages =	{24:1--24:15},
  series =	{Leibniz International Proceedings in Informatics (LIPIcs)},
  ISBN =	{978-3-95977-247-1},
  ISSN =	{1868-8969},
  year =	{2022},
  volume =	{244},
  noeditor =	{Chechik, Shiri and Navarro, Gonzalo and Rotenberg, Eva and Herman, Grzegorz},
  publisher =	{Schloss Dagstuhl -- Leibniz-Zentrum f{\"u}r Informatik},
  noaddress =	{Dagstuhl, Germany},
  doi =		{10.4230/LIPIcs.ESA.2022.24}
}

@InProceedings{mim-width-bound,
  author =	{Jaffke, Lars and Kwon, O-joung and Telle, Jan Arne},
  title =	{A Unified Polynomial-Time Algorithm for Feedback Vertex Set on Graphs of Bounded Mim-Width},
  booktitle =	{35th Symposium on Theoretical Aspects of Computer Science (STACS 2018)},
  pages =	{42:1--42:14},
  series =	{Leibniz International Proceedings in Informatics (LIPIcs)},
  ISBN =	{978-3-95977-062-0},
  ISSN =	{1868-8969},
  year =	{2018},
  volume =	{96},
  noeditor =	{Niedermeier, Rolf and Vall\'{e}e, Brigitte},
  publisher =	{Schloss Dagstuhl -- Leibniz-Zentrum f{\"u}r Informatik},
  noaddress =	{Dagstuhl, Germany},
  doi =		{10.4230/LIPIcs.STACS.2018.42}
}

@InProceedings{CDS-tw,
author="Dom, Michael
and Lokshtanov, Daniel
and Saurabh, Saket
and Villanger, Yngve",
noeditor="Grohe, Martin
and Niedermeier, Rolf",
title="Capacitated Domination and Covering: A Parameterized Perspective",
booktitle="3rd International Workshop on Parameterized and Exact Computation (IWPEC 2008)",
year="2008",
publisher="Springer",
series = {Lecture Notes in Computer Science}, volume = 5018,
noaddress="Berlin, Heidelberg",
pages="78--90",
isbn="978-3-540-79723-4",
doi="10.1007/978-3-540-79723-4_9"
}

@InProceedings{CDS-planar,
author="Bodlaender, Hans L.
and Lokshtanov, Daniel
and Penninkx, Eelko",
editor="Chen, Jianer
and Fomin, Fedor V.",
title="Planar Capacitated Dominating Set Is {W[1]}-Hard",
booktitle="Parameterized and Exact Computation",
year="2009",
publisher="Springer Berlin Heidelberg",
address="Berlin, Heidelberg",
pages="50--60",
isbn="978-3-642-11269-0",
doi="10.1007/978-3-642-11269-0_4"
}

@article{CDS-red-blue-lb,
author = {Fomin, Fedor V. and Golovach, Petr A. and Lokshtanov, Daniel and Saurabh, Saket},
title = {Almost Optimal Lower Bounds for Problems Parameterized by Clique-Width},
journal = {SIAM Journal on Computing},
volume = {43},
number = {5},
pages = {1541-1563},
year = {2014},
doi = {10.1137/130910932},
}

@article{CVC-recent,
      title={Parameterized Capacitated Vertex Cover Revisited}, 
      author={Michael Lampis and Manolis Vasilakis},
      year={2026},
      journal = {arXiv}, number = {2604.18746},
      primaryClass={cs.DS},
      doi={10.48550/arXiv.2604.18746}
}

@article{DIDIMO201972,
title = {{HV}-planarity: Algorithms and complexity},
journal = {Journal of Computer and System Sciences},
volume = {99},
pages = {72-90},
year = {2019},
issn = {0022-0000},
doi = {https://doi.org/10.1016/j.jcss.2018.08.003},
author = {Walter Didimo and Giuseppe Liotta and Maurizio Patrignani},
}

@article{circulating-orientation,
author       = {Bart M. P. Jansen and
                  Liana Khazaliya and
                  Philipp Kindermann and
                  Giuseppe Liotta and
                  Fabrizio Montecchiani and
                  Kirill Simonov},
  title        = {Upward and Rectilinear Planarity are {W[1]}-Hard Parameterized by Treewidth},
  journal      = {{SIAM} Journal on Discrete Mathematics},
  volume       = {40},
  number       = {2},
  pages        = {706--742},
  year         = {2026},
  nourl          = {https://doi.org/10.1137/24m1679240},
  doi          = {10.1137/24M1679240},
  bibsource    = {dblp computer science bibliography, https://dblp.org}
}

@article{BENZWI201187,
title = {Treewidth governs the complexity of target set selection},
journal = {Discrete Optimization},
volume = {8},
number = {1},
pages = {87-96},
year = {2011},
nonote = {Parameterized Complexity of Discrete Optimization},
issn = {1572-5286},
doi = {https://doi.org/10.1016/j.disopt.2010.09.007},
author = {Oren Ben-Zwi and Danny Hermelin and Daniel Lokshtanov and Ilan Newman},
}

@inproceedings{n-dom,
author = {Fink, J. F. and Jacobson, M. S.},
title = {$n$-{D}omination in Graphs},
year = {1985},
isbn = {0471816353},
publisher = {John Wiley \& Sons},
noaddress = {USA},
booktitle = {Graph Theory with Applications to Algorithms and Computer Science},
pages = {283–300},
numpages = {18},
url={https://dl.acm.org/doi/abs/10.5555/21936.25446}
}

@article{k-dom-tw,
author = {Telle, Jan Arne},
title = {Complexity of domination-type problems in graphs},
year = {1994},
issue_date = {Spring 1994},
publisher = {Publishing Association Nordic Journal of Computing},
address = {FIN},
volume = {1},
number = {1},
issn = {1236-6064},
journal = {Nordic Journal of Computing},
nomonth = mar,
pages = {157–171},
numpages = {15},
url={https://www.cs.helsinki.fi/njc/njc1.html}
}

@article{f-dom,
title = {Upper bounds for $f$-domination number of graphs},
journal = {Discrete Mathematics},
volume = {185},
number = {1},
pages = {239-243},
year = {1998},
issn = {0012-365X},
doi = {https://doi.org/10.1016/S0012-365X(97)00204-5},
author = {Beifang Chen and Sanming Zhou},
}

@book{Sparsity-book,
series = {Algorithms and Combinatorics, 28},
publisher = {Springer},
isbn = {3-642-42776-6},
year = {2012},
title = {Sparsity : Graphs, Structures, and Algorithms},
noedition = {1st ed. 2012.},
language = {eng},
author = {Nešetřil, Jaroslav},
doi={https://doi.org/10.1007/978-3-642-27875-4}
}

@article{treebw,
      title={On a tree-based variant of bandwidth and forbidding simple topological minors}, 
      author={Hugo Jacob and William Lochet and Christophe Paul},
      year={2025},
      journal = {arXiv}, number = {2502.11674},
      archivePrefix={arXiv},
      primaryClass={cs.DM},
      doi={10.48550/arXiv.2502.11674}
}

@InProceedings{tpw-1,
author="Seese, D.",
noeditor="Budach, Lothar",
title="Tree-partite graphs and the complexity of algorithms",
booktitle="5th International Symposium on Fundamentals of Computation Theory (FCT 1985)",
series = {Lecture Notes in Computer Science}, volume = {199},
year="1985",
publisher="Springer",
noaddress="Berlin, Heidelberg",
pages="412--421",
isbn="978-3-540-39636-9",
doi="https://doi.org/10.1007/BFb0028825"
}

@article{tpw-2,
title = {Tree-partitions of infinite graphs},
journal = {Discrete Mathematics},
volume = {97},
number = {1},
pages = {203-217},
year = {1991},
issn = {0012-365X},
doi = {https://doi.org/10.1016/0012-365X(91)90436-6},
author = {R. Halin},
}

@article{cutwidth-1,
author = {Chung, Fan R. K.},
title = {On the Cutwidth and the Topological Bandwidth of a Tree},
journal = {SIAM Journal on Algebraic Discrete Methods},
volume = {6},
number = {2},
pages = {268-277},
year = {1985},
doi = {10.1137/0606026},
}

@article{cutwidth-2,
title = {Cutwidth {I}: A linear time fixed parameter algorithm},
journal = {Journal of Algorithms},
volume = {56},
number = {1},
pages = {1-24},
year = {2005},
issn = {0196-6774},
doi = {https://doi.org/10.1016/j.jalgor.2004.12.001},
author = {Dimitrios M. Thilikos and Maria Serna and Hans L. Bodlaender},
}

@article{tree-cutwidth,
title = {The structure of graphs not admitting a fixed immersion},
journal = {Journal of Combinatorial Theory, Series B},
volume = {110},
pages = {47-66},
year = {2015},
issn = {0095-8956},
doi = {https://doi.org/10.1016/j.jctb.2014.07.003},
author = {Paul Wollan},
}

@article{xnlp2,
  author       = {Hans L. Bodlaender and
                  Carla Groenland and
                  Hugo Jacob and
                  Lars Jaffke and
                  Paloma T. Lima},
  title        = {{XNLP}-Completeness for Parameterized Problems on Graphs with a Linear
                  Structure},
  journal      = {Algorithmica},
  volume       = {87},
  number       = {4},
  pages        = {465--506},
  year         = {2025},
  nourl          = {https://doi.org/10.1007/s00453-024-01274-9},
  doi          = {10.1007/S00453-024-01274-9},
  bibsource    = {dblp computer science bibliography, https://dblp.org}
}

@InProceedings{edge-cut-width,
author="Brand, Cornelius
and Ceylan, Esra
and Ganian, Robert
and Hatschka, Christian
and Korchemna, Viktoriia",
noeditor="Bekos, Michael A.
and Kaufmann, Michael",
title="Edge-Cut Width: An Algorithmically Driven Analogue of Treewidth Based on Edge Cuts",
booktitle="48th International Workshop on Graph-Theoretic Concepts  in Computer Science (WG 2022)",
year="2022",
series = {Lecture Notes in Computer Science}, volume = {13453},
publisher="Springer",
noaddress="Cham",
pages="98--113",
isbn="978-3-031-15914-5",
doi="https://doi.org/10.1007/978-3-031-15914-5_8"
}

@Inproceedings{Distel2024,
author="Distel, Marc
and Wood, David R.",
title="Tree-Partitions with Bounded Degree Trees",
booktitle="2021-2022 MATRIX Annals",
year="2024",
publisher="Springer",
series = {MATRIX Book Series}, volume = 5,
noaddress="Cham",
pages="203--212",
isbn="978-3-031-47417-0",
doi="10.1007/978-3-031-47417-0_11",
}

@article{bw-k-ary,
author = {Smithline, Lawren},
title = {Bandwidth of the complete $k$-ary tree},
year = {1995},
issue_date = {July 15, 1995},
publisher = {Elsevier Science Publishers B. V.},
address = {NLD},
volume = {142},
number = {1–3},
issn = {0012-365X},
doi = {10.1016/0012-365X(93)E0219-T},
journal = {Discrete Mathematics},
nomonth = jul,
pages = {203–212},
numpages = {10}
}

@article{compute-tpw,
    title      = {On the parameterized complexity of computing tree-partitions},
    author     = {Hans L. Bodlaender and Carla Groenland and Hugo Jacob},
    doi        = {10.46298/dmtcs.12540},
    journal    = {Discrete Mathematics \& Theoretical Computer Science},
    issn       = {1365-8050},
    volume     = {vol. 26:3},
    issuetitle = {Discrete Algorithms},
    eid        = 3,
    year       = {2025},
    nomonth      = {Feb},
}

@article{Saxe-bw,
author = {Saxe, James B.},
title = {Dynamic-Programming Algorithms for Recognizing Small-Bandwidth Graphs in Polynomial Time},
journal = {SIAM Journal on Algebraic Discrete Methods},
volume = {1},
number = {4},
pages = {363-369},
year = {1980},
doi = {10.1137/0601042},
}

@Article{min-k-1,
author={Binucci, Carla
and B{\"u}ngener, Aaron
and Di Battista, Giuseppe
and Didimo, Walter
and Dujmovi{\'{c}}, Vida
and Hong, Seok-Hee
and Kaufmann, Michael
and Liotta, Giuseppe
and Morin, Pat
and Tappini, Alessandra},
title={Min-$k$-planar Drawings of Graphs},
journal={Journal of Graph Algorithms and Applications},
year={2024},
nomonth={May},
day={17},
volume={28},
number={2},
pages={1-35},
doi={10.7155/jgaa.v28i2.2925},
}

@Article{min-k-2,
author={Campbell, Rutger
and Clinch, Katie
and Distel, Marc
and Gollin, J. Pascal
and Hendrey, Kevin
and Hickingbotham, Robert
and Huynh, Tony
and Illingworth, Freddie
and Tamitegama, Youri
and Tan, Jane
and Wood, David R.},
title={Product structure of graph classes with bounded treewidth},
journal={Combinatorics, Probability and Computing},
year={2024},
edition={2023/12/07},
publisher={Cambridge University Press},
volume={33},
number={3},
pages={351-376},
issn={0963-5483},
doi={10.1017/S0963548323000457},
}

@article{tw-cdw-1,
title = {Upper bounds to the clique width of graphs},
journal = {Discrete Applied Mathematics},
volume = {101},
number = {1},
pages = {77-114},
year = {2000},
issn = {0166-218X},
doi = {https://doi.org/10.1016/S0166-218X(99)00184-5},
author = {Bruno Courcelle and Stephan Olariu},
}

@article{tw-cdw-2,
author = {Corneil, Derek G. and Rotics, Udi},
title = {On the Relationship Between Clique-Width and Treewidth},
journal = {SIAM Journal on Computing},
volume = {34},
number = {4},
pages = {825-847},
year = {2005},
doi = {10.1137/S0097539701385351},
}

@article{cdw-crossing,
title = {On quasi-planar graphs: Clique-width and logical description},
journal = {Discrete Applied Mathematics},
volume = {278},
pages = {118-135},
year = {2020},
monote = {Eighth Workshop on Graph Classes, Optimization, and Width Parameters},
issn = {0166-218X},
doi = {10.1016/j.dam.2018.07.022},
author = {Bruno Courcelle},
}

@InProceedings{stretch-width,
  author =	{Bonnet, \'{E}douard and Duron, Julien},
  title =	{Stretch-Width},
  booktitle =	{18th International Symposium on Parameterized and Exact Computation (IPEC 2023)},
  pages =	{8:1--8:15},
  series =	{Leibniz International Proceedings in Informatics (LIPIcs)},
  ISBN =	{978-3-95977-305-8},
  ISSN =	{1868-8969},
  year =	{2023},
  volume =	{285},
  noeditor =	{Misra, Neeldhara and Wahlstr\"{o}m, Magnus},
  publisher =	{Schloss Dagstuhl -- Leibniz-Zentrum f{\"u}r Informatik},
  noaddress =	{Dagstuhl, Germany},
  URN =		{urn:nbn:de:0030-drops-194279},
  doi =		{10.4230/LIPIcs.IPEC.2023.8}
}
\bibliographystyle{plainurl}
\end{document}